\documentclass[reqno]{amsart}
\usepackage[margin=1.5in,bottom=1.25in]{geometry}	

\usepackage{amsmath}	
\usepackage{amssymb}	
\usepackage{amsfonts}	
\usepackage{amsthm}	
\usepackage[foot]{amsaddr}	

\usepackage[centercolon=true]{mathtools}	

\mathtoolsset{%
}

\usepackage[utf8]{inputenc}	
\usepackage[T1]{fontenc}	

\usepackage[
cal=cm,
]
{mathalfa}

\usepackage{dsfont}	

\usepackage[light,scaled=.95]{AlegreyaSans}

\usepackage[proportional,tabular,lining,sf,mono=false]{libertine}
\usepackage{acronym}	
\newcommand{\acli}[1]{\textit{\acl{#1}}}	
\newcommand{\acdef}[1]{\define{\acl{#1}} \textup{(\acs{#1})}\acused{#1}}	

\usepackage[labelfont={bf,small},labelsep=colon,font=small]{caption}	
\usepackage{subcaption}	
\usepackage[svgnames]{xcolor}	
\colorlet{MyRed}{Crimson!60!DarkRed}
\colorlet{MyBlue}{DodgerBlue!75!black}
\colorlet{MyGreen}{DarkGreen}
\colorlet{MyViolet}{DarkMagenta}

\colorlet{MyLightBlue}{DodgerBlue!20}
\colorlet{MyLightGreen}{MyGreen!20}

\colorlet{PrimalColor}{MyBlue}
\colorlet{PrimalFill}{MyLightBlue}
\colorlet{DualColor}{MyRed}

\colorlet{AlertColor}{MyRed}	
\colorlet{BadColor}{MyRed}	
\colorlet{GoodColor}{MyGreen}	
\colorlet{LinkColor}{MediumBlue}	
\colorlet{MacroColor}{MyViolet}
\colorlet{RevColor}{MediumBlue}	

\newcommand{\afterhead}{.\;}	
\newcommand{\para}[1]{\smallskip\paragraph{\textbf{#1\afterhead}}}

\usepackage{latexsym}	
\usepackage{fontawesome}	
\usepackage{pifont}	

\usepackage{tikz}		
\usetikzlibrary{calc,patterns,positioning}	

\usepackage{array}	
\usepackage{booktabs}	
\usepackage[inline,shortlabels]{enumitem}	
\setlist[1]{topsep=\smallskipamount,itemsep=\smallskipamount,left=\parindent}
\setlist[2]{left=0pt}

\usepackage[kerning=true]{microtype}	

\usepackage{cancel}	
\usepackage{latexsym}	
\usepackage{tabto}	
\usepackage{xspace}	

\usepackage[authoryear,compress]{natbib}	

\bibpunct[, ]{(}{)}{;}{}{,}{,}

\usepackage{hyperref}
\hypersetup{
final,
colorlinks=true,
linktocpage=true,
pdfstartview=FitH,
breaklinks=true,
pdfpagemode=UseNone,
pageanchor=true,
pdfpagemode=UseOutlines,
plainpages=false,
bookmarksnumbered,
bookmarksopen=false,
bookmarksopenlevel=1,
hypertexnames=true,
pdfhighlight=/O,
hyperfootnotes=false,
urlcolor=LinkColor,linkcolor=LinkColor,citecolor=LinkColor,	
pdftitle={},
pdfauthor={},
pdfsubject={},
pdfkeywords={},
pdfcreator={pdfLaTeX},
pdfproducer={LaTeX with hyperref}
}

\usepackage[sort&compress,capitalize,nameinlink]{cleveref}	

\crefname{algo}{Algorithm}{Algorithms}
\crefname{assumption}{Assumption}{Assumptions}

\usepackage{algorithm}	
\usepackage{algpseudocode}	

\usepackage{thmtools}	
\usepackage{thm-restate}	

\theoremstyle{plain}
\newtheorem{theorem}{Theorem}	
\newtheorem{corollary}{Corollary}	
\newtheorem{lemma}{Lemma}	
\newtheorem{proposition}{Proposition}	

\newtheorem*{corollary*}{Corollary}	

\theoremstyle{definition}
\newtheorem{definition}{Definition}	
\newtheorem*{definition*}{Definition}	

\newtheorem*{assumption*}{Assumptions}	

\newtheorem{example}{Example}	
\newtheorem*{example*}{Example}	

\theoremstyle{remark}
\newtheorem{remark}{Remark}	
\newtheorem*{remark*}{Remark}	

\def\endenv{\hfill\raisebox{1pt}{\S}}

\usepackage[showdeletions]{color-edits}	
\newcommand{\draft}[1]{#1}	

\newcommand{\define}[1]{\emph{\draft{#1}}}	
\newcommand{\good}[1]{{\color{GoodColor}#1}}	

\newcommand{\explain}[1]{\tag*{\#\:#1}}

\newcommand{\newmacro}[2]{\newcommand{#1}{\draft{#2}}}	
\newcommand{\newop}[2]{\DeclareMathOperator{#1}{\draft{#2}}}	

\DeclarePairedDelimiter{\braces}{\{}{\}}	
\DeclarePairedDelimiter{\bracks}{[}{]}	
\DeclarePairedDelimiter{\parens}{(}{)}	

\DeclarePairedDelimiter{\abs}{\lvert}{\rvert}	
\DeclarePairedDelimiter{\pospart}{[}{]_{+}}	

\DeclarePairedDelimiterX{\setdef}[2]{\{}{\}}{#1:#2}	
\DeclarePairedDelimiterXPP{\exclude}[1]{\mathopen{}\setminus}{\{}{\}}{}{#1}

\newcommand{\R}{\mathbb{R}}	

\DeclareMathOperator*{\argmax}{arg\,max}	
\DeclareMathOperator*{\intersect}{\bigcap}	
\DeclareMathOperator*{\union}{\bigcup}	

\DeclareMathOperator{\one}{\mathds{1}}	
\DeclareMathOperator{\relint}{ri}	
\DeclareMathOperator{\supp}{supp}	

\newcommand{\cf}{cf.\xspace}	
\newcommand{\eg}{e.g.,\xspace}	
\newcommand{\ie}{i.e.,\xspace}	
\newcommand{\vs}{vs.\xspace}	
\newcommand{\viz}{viz.\xspace}	

\newcommand{\textpar}[1]{\textup(#1\textup)}	

\newcommand{\dis}{\displaystyle}	
\newcommand{\txs}{\textstyle}	

\newcommand{\alt}[1]{#1'}	
\newcommand{\altalt}[1]{#1''}	

\newmacro{\dd}{\:d}	
\newcommand{\eps}{\varepsilon}	
\newcommand{\pd}{\partial}	

\newcommand{\insum}{\sum\nolimits}	

\newmacro{\const}{c}	
\newmacro{\coef}{\lambda}	
\newmacro{\param}{\theta}	
\newmacro{\params}{\Theta}	

\newmacro{\step}{\gamma}	

\newmacro{\pexp}{p}	
\newmacro{\qexp}{q}	
\newmacro{\rexp}{r}	

\newmacro{\beforestart}{0}	
\newmacro{\start}{1}	
\newmacro{\afterstart}{2}	
\newmacro{\running}{\start,\afterstart,\dotsc}	

\newmacro{\run}{n}	
\newmacro{\runalt}{s}	
\newmacro{\runaltalt}{k}	
\newmacro{\nRuns}{T}	
\newmacro{\runs}{\mathcal{\nRuns}}	

\newmacro{\state}{X}	
\newmacro{\statealt}{Y}	
\newmacro{\statealtalt}{Z}	

\newop{\Nash}{NE}	
\newop{\CE}{CE}	
\newop{\CCE}{CCE}	
\newop{\NI}{NI}	

\newop{\brep}{br}	
\newop{\reg}{Reg}	
\newop{\preg}{\overline{Reg}}	
\newop{\val}{val}	

\newmacro{\play}{i}	
\newmacro{\playalt}{j}	
\newmacro{\playaltalt}{k}	
\newmacro{\nPlayers}{N}	
\newmacro{\players}{\mathcal{I}}	

\newmacro{\pure}{\alpha}	
\newmacro{\purealt}{\beta}	
\newmacro{\purealtalt}{\gamma}	
\newmacro{\nPures}{m}	
\newmacro{\pures}{\mathcal{A}}	

\newmacro{\strat}{x}	
\newmacro{\stratalt}{\alt\strat}	
\newmacro{\strataltalt}{\altalt\strat}	
\newmacro{\strats}{\mathcal{X}}	
\newmacro{\intstrats}{\strats^{\circle}}	

\newcommand{\eq}{\sol}	

\newmacro{\loss}{\ell}	
\newmacro{\py}{\pi}	
\newmacro{\pay}{u}	
\newmacro{\payv}{v}	
\newmacro{\classpay}{u}	
\newmacro{\pot}{F}	
\newmacro{\paydiff}{c}	

\newmacro{\game}{\mathcal{G}}	
\newmacro{\gamefull}{\game(\pures,\payv)}	

\newmacro{\fingame}{\Gamma}	
\newmacro{\fingameall}{\Gamma(\players,\pures,\loss)}	

\newmacro{\vertex}{p}	
\newmacro{\vertexalt}{q}	
\newmacro{\vertexaltalt}{\altalt\vertex}	
\newmacro{\nVertices}{V}	
\newmacro{\vertices}{\mathcal{\nVertices}}	

\newmacro{\edge}{e}	
\newmacro{\edgealt}{\alt\edge}	
\newmacro{\edgealtalt}{\altalt\edge}	
\newmacro{\nEdges}{E}	
\newmacro{\edges}{\mathcal{\nEdges}}	

\newmacro{\graph}{\mathcal{G}}	
\newmacro{\graphall}{\graph(\vertices,\edges)}	

\newmacro{\vecspace}{\mathcal{V}}	
\newmacro{\subspace}{\mathcal{W}}	

\newmacro{\bvec}{e}	
\newmacro{\bvecs}{\mathcal{E}}	

\newmacro{\coord}{i}	
\newmacro{\coordalt}{j}	
\newmacro{\coordaltalt}{k}	
\newmacro{\nCoords}{d}	
\newmacro{\dims}{\nCoords}	
\newmacro{\vdim}{\nCoords}	

\newmacro{\pspace}{\mathcal{X}}	
\newmacro{\dspace}{\mathcal{Y}}	

\newmacro{\ppoint}{x}	
\newmacro{\ppointalt}{\alt\ppoint}	
\newmacro{\ppointaltalt}{\altalt\ppoint}	
\newmacro{\ppoints}{\mathcal{X}}	
\newmacro{\pstate}{X}	

\newmacro{\dpoint}{y}	
\newmacro{\dpointalt}{\alt\dpoint}	
\newmacro{\dpointaltalt}{\altalt\dpoint}	
\newmacro{\dpoints}{\mathcal{Y}}	
\newmacro{\dstate}{Y}	

\newmacro{\pvec}{z}	
\newmacro{\dvec}{w}	

\newmacro{\mat}{M}	
\newmacro{\hmat}{H}	

\newmacro{\ones}{\mathbf{1}}	
\newmacro{\eye}{I}	
\newmacro{\zer}{\mathbf{0}}	
\newmacro{\ttop}{{\!\top\!}}	

\DeclarePairedDelimiterXPP{\dnorm}[1]{}{\lVert}{\rVert}{_{\ast}}{#1}	

\DeclarePairedDelimiterXPP{\onenorm}[1]{}{\lVert}{\rVert}{_{1}}{#1}	
\DeclarePairedDelimiterXPP{\twonorm}[1]{}{\lVert}{\rVert}{_{2}}{#1}	
\DeclarePairedDelimiterXPP{\supnorm}[1]{}{\lVert}{\rVert}{_{\infty}}{#1}	

\DeclarePairedDelimiterX{\braket}[2]{\langle}{\rangle}{#1,#2}	

\DeclarePairedDelimiterX{\inner}[2]{\langle}{\rangle}{#1,#2}	

\newcommand{\defeq}{\coloneqq}	
\newcommand{\eqdef}{\eqqcolon}	

\newcommand{\from}{\colon}	

\newop{\Opt}{Opt}	
\newop{\Sol}{Sol}	
\newop{\gap}{Gap}	
\newop{\orcl}{Or}	

\newmacro{\obj}{f}	
\newmacro{\objalt}{g}	
\newmacro{\sobj}{F}	

\newmacro{\gvec}{g}	
\newmacro{\oper}{A}	
\newmacro{\vecfield}{v}	

\newcommand{\sol}[1][\point]{#1^{\ast}}	

\newmacro{\vbound}{G}	
\newmacro{\lips}{L}	
\newmacro{\strong}{\alpha}	
\newmacro{\smooth}{\beta}	

\newop{\tspace}{T}	
\newop{\tcone}{TC}	
\newop{\dcone}{\tcone^{\ast}}	
\newop{\ncone}{NC}	
\newop{\pcone}{PC}	
\newop{\hull}{\Delta}	

\newmacro{\cvx}{\mathcal{C}}	
\newmacro{\subd}{\partial}	

\newop{\Eucl}{\Pi}	
\newop{\logit}{\Lambda}	
\newop{\dkl}{KL}	

\newmacro{\hreg}{h}	
\newmacro{\hconj}{\hreg^{\ast}}	
\newmacro{\breg}{D}	
\newmacro{\mprox}{P}	
\newmacro{\mirror}{Q}	
\newmacro{\fench}{F}	
\newmacro{\hstr}{K}	
\newmacro{\hrange}{H}	
\newmacro{\proxdom}{\points^{\hreg}}	

\DeclarePairedDelimiterXPP{\proxof}[2]{\mprox_{#1}}{(}{)}{}{#2}	

\newmacro{\point}{x}	
\newmacro{\pointalt}{\alt\point}	
\newmacro{\pointaltalt}{\altalt\point}	
\newmacro{\points}{\mathcal{K}}	
\newmacro{\intpoints}{\relint\points}	

\newmacro{\base}{p}	
\newmacro{\basealt}{q}	

\newmacro{\open}{\mathcal{U}}	
\newmacro{\closed}{\mathcal{C}}	
\newmacro{\cpt}{\mathcal{K}}	
\newmacro{\nhd}{\mathcal{U}}	

\newop{\ex}{\mathbb{E}}	
\newop{\prob}{\mathbb{P}}	
\newop{\Var}{Var}	
\newop{\simplex}{\hull}	

\providecommand{\given}{\mathbin{}{\nonscript\mkern-\medmuskip}\vert{\nonscript\mkern-\medmuskip} \mathbin{}}

\DeclarePairedDelimiterXPP{\exof}[1]{\ex}{[}{]}{}{
\renewcommand\given{\nonscript\,\delimsize\vert\nonscript\,\mathopen{}} #1}

\DeclarePairedDelimiterXPP{\probof}[1]{\prob}{(}{)}{}{
\renewcommand\given{\nonscript\:\delimsize\vert\nonscript\:\mathopen{}} #1}

\DeclarePairedDelimiterXPP{\oneof}[1]{\one}{\{}{\}}{}{#1}

\newmacro{\sample}{\omega}	
\newmacro{\samples}{\Omega}	

\newmacro{\filter}{\mathcal{F}}	
\newmacro{\probspace}{(\samples,\filter,\prob)}	

\newmacro{\event}{E}       
\newmacro{\eventalt}{H}       
\newmacro{\mean}{\mu}	
\newmacro{\sdev}{\sigma}	
\newmacro{\variance}{\sdev^{2}}	

\newmacro{\error}{Z}	
\newmacro{\noise}{Z}	
\newmacro{\offset}{b}	

\newmacro{\serror}{\theta}	
\newmacro{\snoise}{\xi}	
\newmacro{\sbias}{\psi}	

\newmacro{\sbound}{M}	
\newmacro{\bbound}{B}	
\newmacro{\noisepar}{\sdev}	
\newmacro{\noisevar}{\variance}	

\newmacro{\elem}{e}	
\newmacro{\elemalt}{ee}	
\newmacro{\elems}{\mathcal{E}}	
\newmacro{\iElem}{i}	
\newmacro{\nElems}{\nPures}	

\newmacro{\partition}{\mathcal{S}}	
\newcommand{\refines}{\mathrel{}\preccurlyeq\mathrel{}}	
\newcommand{\refined}{\mathrel{}\succcurlyeq\mathrel{}}	

\newmacro{\set}{\mathcal{S}}	
\newmacro{\setalt}{\alt\set}	

\newmacro{\class}{\mathcal{S}}	
\newmacro{\classalt}{\alt\class}	
\newmacro{\classes}{\mathcal{P}}	
\newmacro{\iClass}{i}	
\newmacro{\jClass}{j}	
\newmacro{\nClasses}{m}	
\newmacro{\allClasses}{\classes\atlvl{\!\scriptscriptstyle\bullet}}	
\newmacro{\struct}{\allClasses}	

\newcommand{\classof}[2][\hspace{-.5pt}]{[#2]\atlvl{#1}}

\newmacro{\basepay}{r}
\newmacro{\range}{R}

\newmacro{\lvl}{\ell}	
\newmacro{\lvlalt}{k}	
\newmacro{\lolvl}{\lvlalt}	
\newmacro{\lolvls}{\braces{0,\dotsc,\nLvls-1}}	
\newmacro{\lvls}{\braces{0,\dotsc,\nLvls}}	
\newmacro{\nLvls}{L}	
\newcommand{\atlvl}[1]{_{#1}}

\newmacro{\source}{\pures}	

\DeclarePairedDelimiterXPP{\attof}[1]{\draft{\mathrm{lvl}}}{(}{)}{}{#1}

\newmacro{\parent}{\class}	
\newmacro{\parentset}{\set}	
\newcommand{\parentof}{\mathrel{}\vartriangleright\mathrel{}}	

\newmacro{\child}{\classalt}	
\newmacro{\childset}{\setalt}	
\newcommand{\childof}{\mathrel{}\vartriangleleft\mathrel{}}	

\newcommand{\sibling}[1][]{\overset{\cramped{#1}}{\sim}}	
\newop{\children}{ch}
\newmacro{\nChildren}{m}	

\newcommand{\desc}[1][]{\mathrel{}\prec_{#1}\mathrel{}}	
\newcommand{\anc}[1][]{\mathrel{}\succ_{#1}\mathrel{}}	
\newcommand{\desceq}[1][]{\mathrel{}\preccurlyeq_{#1}\mathrel{}}	

\newmacro{\temp}{\mu}	
\newmacro{\diff}{\delta}	
\newmacro{\learn}{\eta}	
\newmacro{\score}{y}	
\newmacro{\allScores}{\score_{\bullet}}	
\newmacro{\pf}{Z}	

\newcommand{\choice}[1][]{%
\renewcommand\given{\mathrel{}{\nonscript\mkern-\medmuskip}|{\nonscript\mkern-\medmuskip} \mathrel{}}%
\draft P_{#1}}	

\newmacro{\energy}{E}	
\newmacro{\test}{p}	

\newmacro{\dynfield}{F}	
\newcommand{\olim}[1][\strat]{\hat#1}	
\newmacro{\lyap}{W}	

\renewcommand{\time}{\draft t}	
\newmacro{\timealt}{\tau}	
\newmacro{\tstart}{0}	
\newmacro{\horizon}{T}	

\newcommand{\orbit}[2][]{\strat_{#1}(#2)}	

\newmacro{\flowmap}{\Theta}	

\newcommand{\all}{\draft{\textsf{all}}\xspace}
\newcommand{\public}{\draft{\textsf{public}}\xspace}
\newcommand{\private}{\draft{\textsf{private}}\xspace}
\newcommand{\bus}{\draft{\textsf{bus}}\xspace}
\newcommand{\metro}{\draft{\textsf{metro}}\xspace}
\newcommand{\tram}{\draft{\textsf{tram}}\xspace}
\newcommand{\car}{\draft{\textsf{car}}\xspace}
\newcommand{\bike}{\draft{\textsf{bike}}\xspace}
\newcommand{\busOne}{\draft{\textsf{bus line 1}}\xspace}

\newcommand{\bikeOne}{\draft{\textsf{bike path 1}}\xspace}

\newmacro{\redbus}{\textsf{red bus}\xspace}
\newmacro{\bluebus}{\textsf{blue bus}\xspace}

\addauthor[\textbf{Bill}]{WHS}{DarkGreen}

\addauthor[\textbf{Pan}]{PM}{Blue}

\newop{\cand}{cand}	
\newmacro{\real}{x}	
\newmacro{\switch}{\varrho}	
\newmacro{\pdist}{P}	
\newmacro{\imit}{r}	
\newmacro{\Extr}{\Theta}	
\newmacro{\extr}{\theta}	
\newmacro{\growth}{g}	
\newmacro{\rate}{\lambda}	
\newmacro{\sumrate}{\Lambda}
\newmacro{\sumdiff}{\Delta}

\newcommand{\pcand}{\pdist^{\textup{cand}}}
\newcommand{\pmeet}{\pdist^{\textup{meet}}}

\begin{document}


\title
[Nested Replicator Dynamics and Similarity-Based Learning]
{Nested Replicator Dynamics, Nested Logit Choice,\\
and Similarity-Based Learning}

\author[P.~Mertikopoulos]{Panayotis Mertikopoulos$^{\ast}$}
\address
{$^{\ast}$%
Univ. Grenoble Alpes, CNRS, Inria, Grenoble INP, LIG, 38000 Grenoble, France.}
\email{\href{mailto:panayotis.mertikopoulos@imag.fr}{panayotis.mertikopoulos@imag.fr}}

\author[W.~H.~Sandholm]{William H. Sandholm$^{\S}$}
\address
{$^{\S}$\hspace{.5pt}Department of Economics, University of Wisconsin, 1180 Observatory Drive, Madison WI 53706, USA.}
\email{\href{mailto:whs@ssc.wisc.edu}{whs@ssc.wisc.edu}}

\thanks{The original idea for this paper was due to the second author, and the main results were obtained during a visit of the first author to UW Madison in 2017.
Any mistakes or awkward turns of phrase are the sole responsbility of the first author.}




\newacro{LHS}{left-hand side}
\newacro{RHS}{right-hand side}
\newacro{iid}[i.i.d.]{independent and identically distributed}
\newacro{lsc}[l.s.c.]{lower semi-continuous}

\newacro{DA}{dual averaging}
\newacro{IIA}{independence from irrelevant alternatives}
\newacro{EW}{exponential weights}
\newacro{NEW}{nested exponential weights}
\newacro{RL}{regularized learning}
\newacro{FTRL}{``follow the regularized leader''}
\newacro{ARUM}{additive random utility model}
\newacro{RUM}{random utility model}
\newacro{PPI}{pairwise proportional imitation}
\newacro{SPPI}{similarity-based pairwise proportional imitation}
\newacro{NPPI}{nested pairwise proportional imitation}
\newacro{RD}{replicator dynamics}
\newacro{2RD}{two-level replicator dynamics}
\newacro{NRD}{nested replicator dynamics}
\newacro{ES}{evolutionarily stable}
\newacro{ESS}{evolutionarily stable state}
\newacro{GES}{globally evolutionarily stable}
\newacro{GESS}{globally evolutionarily stable state}
\newacro{LC}{logit choice}
\newacro{NLC}{nested logit choice}
\newacro{XL}{exponential learning}
\newacro{NXL}{nested exponential learning}
\newacro{NRL}{nested regularized learning}
\newacro{HR}{Hess\-i\-an Rie\-man\-ni\-an}
\newacro{RPS}{Rock-Paper-Scissors}
\newacro{MP}{Matching Pennies}
\newacro{KL}{Kull\-back\textendash Le\-ib\-ler}
\newacro{LHS}{left-hand side}
\newacro{RHS}{right-hand side}
\newacro{NE}{Nash equilibrium}
\newacroplural{NE}[NE]{Nash equilibria}
\newacro{RE}{restricted equilibrium}
\newacroplural{RE}[RE]{restricted equilibria}

\begin{abstract}

We consider a model of learning and evolution in games whose action sets are endowed with a partition-based similarity structure intended to capture exogenous similarities between strategies.
In this model,
revising agents have a higher probability of comparing their current strategy with other strategies that they deem similar,
and they switch to the observed strategy with probability proportional to its payoff excess.
Because of this implicit bias toward similar strategies, the resulting dynamics \textendash\ which we call the \acli{NRD} \textendash\ do not satisfy any of the standard monotonicity postulates for imitative game dynamics;
nonetheless, we show that they retain the main long-run rationality properties of the replicator dynamics, albeit at quantitatively different rates.
We also show that the induced dynamics can be viewed as a stimulus-response model in the spirit of \citet{ER98}, with choice probabilities given by the \acl{NLC} rule of \citet{BA73} and \citet{McF78}.
This result generalizes an existing relation between the replicator dynamics and the \acl{EW} algorithm in online learning, and provides an additional layer of interpretation to our analysis and results.
\end{abstract}

\maketitle
\allowdisplaybreaks	
\acresetall	

\section{Introduction}
\label{sec:introduction}

In the mass-action interpretation of evolutionary game theory, a nonatomic population of strategically-minded agents interact repeatedly via a set of simple myopic rules, and the governing dynamics arise from the aggregation of these individual interactions.
Typically, these rules are described by a \emph{revision protocol} \textendash\ that is, a rule which determines how frequently agents playing a given strategy consider switching to a different one, and with what criteria.
Thus, taken together, a population game and a revision protocol collectively define an aggregate dynamical system which uniquely prescribes the course of play starting from any initial population state.

A common feature of many revision protocols is the notion of \emph{imitation:}
the revising agent observes the behavior of a randomly selected individual, and then switches to the strategy of the observed agent with a probability that may depend on the revising agent's incumbent payoff, the payoff of the observed strategy, or both.
The archetypal example of this process is the \acdef{PPI} protocol of \citet{Hel92}, where the switching probability is proportional to the payoff excess of the observed strategy over the incumbent one, and the distribution of actions in the population evolves according to the \acl{RD} of \citet{TJ78}.
Beyond this simple revision rule, different models of imitation have been considered by \citet{BW96}, \citet{BinSam97}, \citet{Sch98}, \citet{FI06,FI08}, \citet{MV22}, and many others;
for a summary, see \citet{Wei95}, \citet{San10}, \citet{HLMS22}, and references therein.

Invariably, in all these revision rules, the revising agent inspects the strategy of ``the first person they meet in the street'' \citep{Wei95}.
This premise is certainly reasonable when there are no prior beliefs or selection biases in the population;
however, when the players' strategies share a set of exogenous similarities,
an agent may be biased toward sampling similar strategies rather than dissimilar ones.
[For example, when commuting from home to work, an agent may prefer to compare routes or departure times for the same mode of transport instead of a different, less familiar one.]
We are thus led to the following question:
\begin{quote}
\centering
\itshape
How are the dynamics affected if players are more likely to observe\\
and compare strategies that are more similar to their current ones?
\end{quote}

\para{Our contributions}

To address this question, we examine a simple model of evolution and learning in games whose action sets are endowed with a partition-based \emph{similarity structure}, \ie a hierarchy of nested partitions \textendash\ or \emph{similarity classes} \textendash\ that capture the degree of exogenous similarity between distinct strategies.
In our model, the revising agent has a higher probability of comparing payoffs with agents playing similar strategies, and the greater the degree of similarity between strategies, the greater the probability of such a comparison being made.
Then, after observing their opponent's payoff,
the agent considers imitating the observed strategy only if its payoff is higher than their current one,
and they consummate the switch with probability proportional to the difference in payoffs.

When there are no exogenous similarities between strategies, this model reduces to the \acl{PPI} protocol of \citet{Hel92} mentioned above, and the population follows the \acl{RD} of \citet{TJ78}.
In our approach, the additional affinity toward similar, more familiar strategies introduces an implicit selection bias, which in turn leads to quite different aggregate dynamics.
We call the resulting law of evolution the \acdef{NRD}, in reference to the nested similarity structure that drives the players' payoff comparisons.%
\footnote{\citet{MerSan18} derived a version of these dynamics as a special case of a ``Hessian Riemannian system'', and deferred the dynamics' microfoundations and analysis to the current paper.
This link with Riemannian game dynamics can be seen as an additional layer of interpretation.}

To have a baseline in mind, the standard \acl{RD} exhibit a wide range of rationality properties,
including the elimination of strictly dominated strategies,
the local stability of \aclp{ESS},
global convergence in potential and strictly contractive games,
etc.%
\footnote{For a detailed discussion, see \cite{HS88}, \cite{Wei95} and \cite{San10,San15}.}
A subset of these properties extend to a more general class of dynamics exhibiting monotone percentage growth rates, the canonical monotonicity condition for dynamics based on imitation.%
\footnote{For early appearances of this, see \cite{Nac90,Fri91} and \cite{SZ92}.}
Nevertheless, because of the increased affinity toward similar strategies, the \acl{NRD} fail even the most basic versions of monotonicity:
unless the implicit bias toward similar strategies is very small, a strategy that performs well relative to the full collection of strategies but poorly compared to similar strategies will tend to become less common in the population.
It is therefore unclear whether the \acl{NRD} possess \emph{any} appealing rationality properties whatsoever.

Despite this handicap, we show that the \acl{NRD} provide remarkably strong support for traditional solution concepts.
More precisely, our first result (\cref{thm:NRD}) shows that they retain all the mainstay rationality properties of the replicator dynamics, including in particular that
\begin{enumerate*}
[(\itshape i\hspace*{1pt}\upshape)]
\item
strictly dominated strategies become extinct in the long run;
\item
if the dynamics converge, their limit is a \acl{NE};
\item
dynamically stable states are \aclp{NE};
\item
strict \aclp{NE} and, more generally, \acp{ESS} are dynamically stable and attracting (with global \acp{ESS} being globally attracting);
and
\item
the dynamics converge to \acl{NE} in potential and strictly contractive games.
\end{enumerate*}
However, from a quantitative viewpoint, there are visible differences in the rate at which these properties emerge, \eg as a function of whether a strategy is dominated by a similar or a dissimilar strategy.

In addition to the above, we also establish a fundamental connection between the \acl{NRD} and stimulus-response learning in the spirit of \citet{ER98}.
The precise setting of this connection is a continuous-time model of learning through reinforcement in which agents track the cumulative payoffs of their strategies and select strategies by feeding these cumulative scores into a probabilistic choice rule.
It is known from the work of \cite{Rus99} and \cite{HSV09} that if players choose mixed strategies using the logit choice rule, the evolution of these strategies follows the replicator dynamics.
In \cref{sec:learning}, we derive a considerable generalization of this result:
if players employ the \acli{NLC} rule \citep{BA73,McF78} their mixed strategies follow the \acli{NRD} (\cref{thm:NEW}).

This relation is particularly surprising in light of the fundamental differences between the two nesting constructions.
This is partially explained by our results in \cref{sec:FTRL}, where we describe a variant model of similarity-based learning, and we provide a nonrecursive representation of the \acl{NLC} rule as a ``regularized best-response'' model.
This last representation shows that the \acl{NRD} can also be seen as a special case of the widely used \ac{FTRL} family of online learning methods \citep{SS11}, providing in this way a triple equivalence between the revisionist, stimulus-response, and regularized learning viewpoints.

\para{Related work on similarity-based reasoning}

The idea that similarities play an important role for decision-making goes back at least to \citet{Luc56}, with \citet{Deb60} pointing out that realistic models of probabilistic choice should allow for violations of the principle of ``irrelevance of added alternatives'' \citep{Luc59}.
This issue is illustrated clearly by a canonical example from transportation due to \citet{McF74a}, where two similar options (a blue bus and a red bus) are each chosen less frequently than a dissimilar one (a car), even when all options are otherwise equally desirable in terms of utilities.
In this regard, the need to account for similarities when taking decisions has had a lasting impact in lasting impact in theoretical and econometric models of discrete choice \citep{McF81,McF01,BAL85,AdPT92}.

In a learning context, \citet{GS95} considered a model where agents reason by drawing analogies to similar situations encountered in the past, and they provided a concrete axiomatization of decision rules that output the ``best'' action based on its past performance in similar cases.
More recently, \citet{SteSte08} examined an analogous setting where players continually encounter new strategic, normal-form games and extrapolate beliefs from similar past situations, ultimately generating contagion of actions across different games through a ``transfer learning'' process.
This idea was taken further by \citet{Men12} who introduced a partition of all games encountered into categories, and considered a process of simultaneous learning of actions and partitions.
A crucial point where our conclusions differ with these works is that the dynamics of \citet{Men12} \emph{destabilize} strict equilibria, whereas strict equilibria \emph{are always stable} under the \acl{NRD}.

One reason for this discrepancy is that the similarities mentioned above concern events and situations encountered in the past, not the players' actions per se \textendash\ in this sense, similarities are \emph{endogenous}, not \emph{exogenous}.
In the context of extensive-form games, \citet{Jeh05} \textendash\ and, more recently, \citet{JS05,JS07} \textendash\ examined a dynamic behavioral model where agents bundle together decision nodes into distinct, static similarity classes, and agents only try to learn based on the average behavior in each class.
While the model of \citet{JS05,JS07} is related in spirit to our approach (in both cases, the notion of similarity is determined exogenously and it does not fluctuate over time), their focus is on partitions of sequences of moves in extensive-form games, which have no analogue in the partition of one-shot actions in a population game.

Finally, in a recent unpublished manuscript, \citet{BorMai20} introduced a model of learning in finite games with exogenous similarities in the signals received by the players.%
\footnote{We thank Tilman Börgers for bringing this work to our attention;
our paper's title is partially due to this work.}
These similarities need not be transitive or symmetric (so they do not necessarily comprise a partition-based structure), leading to the novel solution concept of a \emph{similarity equilibrium}.%
\footnote{When the similarities are partition-based, similarity equilibria boil down to correlated equilibria.}
\citet{BorMai20} subsequently showed that similarity-based learning only converges to similarity equilibria \textendash\ \ie if the players' learning process converges, the limit is a similarity equilibrium.
This property is reminiscent of the internal consistency of the \acl{NRD} \textendash\ \ie if the \acl{NRD} converge, they converge to a \acl{NE} \textendash\ but the setting of \citet{BorMai20} is othwerwise quite different to our own, so no further analogies can be drawn.

\section{Preliminaries}
\label{sec:prelims}

\subsection{Population games}
\label{sec:game}

Throughout this paper, we focus on single-population games, \ie the set of \emph{players} is modeled by the unit interval $\players = [0, 1]$
and
each agent selects (in a measurable way) an \emph{action} \textendash\ or \emph{pure strategy} \textendash\ from some finite set $\pures = \setdef{\pure_{\iElem}}{\iElem = 1,\dotsc,\nElems}$.%
\footnote{Our analysis extends to multi-population settings without much effort, but we do not treat this case to avoid overburdening the notation and the presentation.}
The players' payoffs are then determined by their choice of action and the population shares of each action (so any measure-zero set of players has no impact on the game's payoffs).

In more detail, if the mass of agents playing $\pure\in\pures$ is $\strat_{\pure} \in [0,1]$, the \emph{state of the population} is defined as the overall distribution of actions $\strat = (\strat_{\pure})_{\pure\in\pures}$, viewed here as an element of the $\nElems$-simplex
\begin{equation}
\label{eq:simplex}
\txs
\strats
	\equiv \simplex(\pures)
	= \setdef*{\strat\in\R_{+}^{\nElems}}{\sum_{\pure\in\pures} \strat_\pure =1}.
\end{equation}
For posterity, given a subset $\class$ of $\pures$, we will also write
\begin{align}
\label{eq:strat-class}
\strat_{\class}
	&= \sum_{\pure\in\class} \strat_{\pure}
\end{align}
for the mass of agents playing actions drawn from $\class$, and, analogously, we will write
\begin{align}
\label{eq:strat-cond}
\strat_{\classalt \vert \class}
	&= \frac{\strat_{\class\cap\classalt}}{\strat_{\class}}
	= \frac{\sum_{\pure\in\class\cap\classalt} \strat_{\pure}}{\sum_{\pure\in\class} \strat_{\pure}}
\end{align}
for the relative share of agents drawing an action from $\classalt\subseteq\pures$ within the sub-population of agents playing an action in $\class$ (obviously, $\strat_{\classalt \vert \class} = 0$ if $\class\cap\classalt = \varnothing$).%
\footnote{For simplicity, we will make no distinction between $\strat_{\{\pure\}}$ and $\strat_{\pure}$ for any individual action $\pure\in\pures$.
We also assume above that $\strat$ has full support, \ie $\strat_{\pure}>0$ for all $\pure\in\pures$.}

The nonatomic nature of the game is captured by the fact that the players' rewards only depend on the state of the population and not the individual choices of each player.
More precisely, the players' rewards are determined by an ensemble of \emph{payoff functions} $\payv_{\pure}\from\strats\to\R$, $\pure\in\pures$,  with $\payv_{\pure}(\strat)$ representing the payoff to $\pure$-strategists when the population's state is $\strat$.
Collectively,
we will write $\payv(\strat) = (\payv_{\pure}(\strat))_{\pure\in\pures}$ for the population's \emph{payoff vector} at state $\strat\in\strats$,
so the associated population mean payoff is given by
\begin{equation}
\label{eq:pay}
\pay(\strat)
	= \braket{\payv(\strat)}{\strat}
	= \sum_{\pure\in\pures} \strat_{\pure} \payv_{\pure}(\strat),
\end{equation}
with $\braket{\py}{\strat} \equiv \sum_{\pure\in\pures} \strat_{\pure} \py_{\pure}$ denoting the canonical pairing between $\py$ and $\strat$ (both viewed here as elements of $\R^{\nElems}$).
Putting all this together, a \emph{population game} is a tuple $\game \equiv \gamefull$ with $\pures$ and $\payv$ defined as above.

\begin{example}
[Symmetric random matching]
\label{ex:matching}
One of the most widely studied classes of population games is given by \emph{random matching} in finite games \citep{Wei95,San10}.
In this setting, two players are drawn uniformly at random from the population and, subsequently, they are matched to play a symmetric two-player game with payoff matrix $\mat\in\R^{\nElems \times \nElems}$.
Thus, averaging over the population, the mean payoff to $\pure$-strategists at state $\strat\in\strats$ is $\payv_{\pure}(\strat) = \sum_{\purealt\in\pures} \mat_{\pure\purealt} \strat_{\purealt}$
and the game's payoff field is $\payv(\strat) = \mat\strat$.
\endenv
\end{example}

\subsection{Solution concepts}
\label{sec:solutions}

In our analysis, we will consider an array of different solution concepts and notions for nonatomic games.
For completeness, we recall some basic definitions below:

\begin{definition}
\label{def:solutions}
Consider a population game $\game\equiv\gamefull$.
Then:
\begin{enumerate}

\item
A pure strategy $\pure\in\pures$ is \emph{strictly dominated} by $\purealt\in\pures$ if
\begin{equation}
\label{eq:dominated}
\payv_{\pure}(\strat)
	< \payv_{\purealt}(\strat)
	\qquad
	\text{for all $x\in\strats$}.
\end{equation}

\item
A population state $\eq\in\strats$ is a \acdef{NE} of $\game$ if
\begin{equation}
\label{eq:Nash}
\payv_{\pure}(\eq)
	\geq \payv_{\purealt}(\eq)
	\qquad
	\text{for all $\pure,\purealt\in\pures$ such that $\eq_{\pure}>0$}
\end{equation}
or, equivalently, if
\begin{equation}
\label{eq:Nash-var}
\braket{\payv(\eq)}{\strat - \eq}
	\leq 0
	\qquad
	\text{for all $x\in\strats$}.
\end{equation}
Finally, $\eq$ is a \acli{RE} of $\game$ if $\payv_{\pure}(\eq) = \payv_{\purealt}(\eq)$ for all $\pure,\purealt\in\supp(\eq)$.

\item
A population state $\eq\in\strats$ is an \acdef{ESS} of $\game$ if
\begin{equation}
\label{eq:ESS}
\braket{\payv(\strat)}{\strat - \eq}
	< 0
	\qquad
	\text{for all $\strat\neq\eq$ close to $\eq$},
\end{equation}
In particular, $\eq$ is a \acdef{GESS} if \eqref{eq:ESS} holds for all $\strat\in\strats\setminus\{\eq\}$.

\item
$\game$ is a \emph{potential game} if $\payv(\strat) = \nabla\pot(\strat)$ for some \emph{potential function} $\pot\from\strats\to\R$.

\item
$\game$ is \emph{strictly monotone} \textendash\ or \emph{strictly contractive} \textendash\ if
\begin{equation}
\label{eq:monotone}
\braket{\payv(\strat') - \payv(\strat)}{\stratalt - \strat}
	\leq 0
	\qquad
	\text{for all $\strat,\stratalt\in\strats$},
\end{equation}
with equality if and only if $\strat = \stratalt$.%
\footnote{%
The terminology ``contractive'' follows \citet{San15};
by contrast, \citet{HS09} and \citet{SW16} respectively use the term \emph{stable} and \emph{dissipative} for essentially the same concept.
Our choice of terminology reflects the long history of \eqref{eq:monotone} in convex analysis and optimization, where it is known as ``\emph{monotonicity}'', \cf \citet{HLMS22}, \citet{DMSV23}, and references therein.
The first author hopes that the second author would have found this compromise acceptable.
}
\end{enumerate}
\end{definition}

Some well-known facts that will come in handy later are as follows:
\begin{itemize}
\item
If $\eq$ is \acl{ES}, it is \emph{a fortiori} a \acl{NE} \citep{HSS79}.
\item
If $\game$ is potential, any local maximizer of $\pot$ is a \acl{NE} \citep{San01}.
\item
If a game admits a (strictly) concave potential, it is (strictly) monotone;
also, a strictly monotone game admits a unique \acl{NE} which is a \ac{GESS} \citep{San10}.
\end{itemize}
For an in-depth discussion, see \citet{Wei95}, \citet{San10} and references therein.

\subsection{Revision protocols and evolutionary dynamics}
\label{sec:dynamics}

Evolutionary game dynamics are commonly derived from explicit microfoundations describing the agents' efforts to increase their individual payoffs \citep{Wei95,BW03,San10,San10b}.
In this context, each agent occasionally receives opportunities to switch actions \textendash\ say, based on the rings of a Poisson alarm clock \textendash\ and, at such moments, they choose a new action by applying a \emph{revision protocol}.

Typically, a revision protocol is defined by specifying the \emph{conditional switch rates} $\switch_{\pure\purealt}(\py,\strat)$ at which an $\pure$-strategist might consider switching to strategy $\purealt$ when the population is at state $\strat\in\strats$ and the payoff vector is $\py\in\R^{\nElems}$.
In this way, for any given population game $\game\equiv\gamefull$, a revision protocol $\switch$ induces the \emph{mean dynamics}
\begin{equation}
\label{eq:MD}
\tag{MD}
\dot \strat_{\pure}
	= \sum_{\purealt\in\pures} \strat_\purealt \switch_{\purealt\pure}(\payv(\strat),\strat)
	- \strat_{\pure} \sum_{\purealt\in\pures} \switch_{\pure\purealt}(\payv(\strat),\strat).
\end{equation}
The dynamics \eqref{eq:MD} specify the absolute rate of change in the use of each action $\pure\in\pures$ as the difference between inflows \emph{to} $\pure$ from other actions and outflows \emph{from} $\pure$ to other actions.
To motivate the analysis to come, we show how this approach provides microfoundations for the \acli{RD}, the most widely studied dynamics in evolutionary game theory.%

\begin{example*}
[\Acl{PPI}]
\label{ex:PPI}
Following \citet{Hel92}, consider the \acli{PPI} protocol
\acused{PPI}
\begin{equation}
\label{eq:PPI}
\tag{PPI}
\switch_{\pure\purealt}(\py,\strat)
	= \strat_{\purealt} \pospart{\py_\purealt - \py_{\pure}}
\end{equation}
where $\pospart{\real} \defeq \max\{0,\real\}$ denotes the positive part of $\real\in\R$.
Under \eqref{eq:PPI}, a revising $\pure$-strategist meets another individual \textendash\ a \emph{mentor} or \emph{comparator} \textendash\ drawn from the population uniformly at random;
thus, if the population is at state $\strat\in\strats$, a $\purealt$-strategist is encountered with probability $\strat_{\purealt}$.
If the payoff of the revising agent is lower than the comparator's, the agent switches to the comparator strategy with probability proportional to the payoff excess $\pospart{\py_{\purealt} - \py_{\pure}}$;
otherwise, the revising agent sticks to their current choice.
Accordingly, the \emph{conditional imitation rate} of $\purealt$-strategists by $\pure$-strategists under \eqref{eq:PPI} is
\begin{equation}
\label{eq:imit-PPI}
\imit_{\pure\purealt}(\py,\strat)
	= \pospart{\py_{\purealt} - \py_{\pure}}
\end{equation}
so $\switch_{\pure\purealt} = \strat_{\purealt} \imit_{\pure\purealt}$.
Then, substituting protocol \eqref{eq:PPI} into \eqref{eq:MD} and rearranging yields the \acli{RD} of \citet{TJ78}, \viz
\begin{equation}
\label{eq:RD}
\tag{RD}
\dot \strat_{\pure}
	= \strat_{\pure} \bracks*{\payv_{\pure}(\strat) - \pay(\strat)}
\end{equation}
with $\pay(\strat) = \sum_{\pure\in\pures} \strat_{\pure} \payv_{\pure}(\strat)$ denoting the mean population payoff under $\strat$.
This derivation plays a central role in our analysis, so we will return to it several times in the sequel.%
\endenv
\end{example*}

\section{The model: Imitation driven by similarities}
\label{sec:model}

We now turn to our paper's main objective, that is, the study of how similarities between strategies could impact the players' aggregate behavior over time.
Our model for this has two main components:
\begin{enumerate*}
[(\itshape i\hspace*{1pt}\upshape)]
\item
a \emph{similarity structure} intended to capture varying degrees of similarity between strategies;
and
\item
the \acdef{NPPI} protocol, which provides a similarity-driven variant of the \ac{PPI} protocol of \citet{Hel92}.
\end{enumerate*}
Taken together, these two components define the \acdef{NRD},
which we introduce later in this section and analyze in detail in \cref{sec:analysis}.

\subsection{Similarity structures}
\label{sec:struct}

To set the stage for the sequel \textendash\ and because the notation involved can become somewhat cumbersome at times \textendash\ it will help to keep in mind as a running example a population of commuters who consider different alternatives to go from home to work each morning.
At the highest level, these alternatives could be grouped into public and private transport modes;
subsequently, different alternatives could be divvied up into smaller subgroups based on the individual transport mode chosen (bus, car, metro,\dots), the choice of route and/or departure time, etc.
The result is a tree-like grouping of alternatives into nested clusters and subclusters, which we depict in \cref{fig:transport} above.


\begin{figure}[tbp]

\begin{tikzpicture}
[scale=1.2,
empty/.style={inner sep=2pt,minimum height=3ex},
class/.style={rounded corners=2pt,draw,inner sep=2pt,minimum height=3ex},
desc/.style={text width=15em,text height=1ex,minimum height=3ex,align=left},
edgestyle/.style={-},
>=stealth]

\def\side{1}
\def\costhirty{0.8660256}
\def\cosfortyfive{0.7071068}

\small
\sf

\node [desc] (lvl0) at (5.5,0) {$\classes\atlvl{0}\colon \braces{\all}$\vphantom{pl}};
\node [desc] (lvl1) [below=2em of lvl0.south] {$\classes\atlvl{1}\colon \braces{\public,\private}$\vphantom{pl}};
\node [desc] (lvl2) [below=2em of lvl1.south] {$\classes\atlvl{2}\colon \braces{\bus,\metro,\dotsc,\bike}$\vphantom{pl}};
\node [desc] (lvl3) [below=2em of lvl2.south] {$\classes\atlvl{3}\colon \braces{\busOne,\dotsc,\bikeOne,\dotsc}$\vphantom{pl}};

\node [class] (root) at (0,0) {transport alternatives};

\node [class] (public) [below left=2em and 3em of root.south] {\public\vphantom{pl}};
\node [class] (private) [below right=2em and 3em of root.south] {\private\vphantom{pl}};

\node [class] (bus) [below left=2em and 1em of public] {\bus\vphantom{pl}};
\node [class] (metro) [below=2em of public] {\metro\vphantom{pl}};
\node [class] (tram) [below right=2em and 1em of public] {\tram\vphantom{pl}};

\node [class] (car) [below left=2em and 0em of private] {\car\vphantom{pl}};
\node [class] (bike) [below right=2em and 1em of private] {\bike\vphantom{pl}};

\node [class] (bus1) [below left=2em and 1em of bus] {line 1\vphantom{pl}};
\node [class] (bus2) [below=2em and 0em of bus] {line 2\vphantom{pl}};
\node [class] (busX) [below right=2em and 1em of bus] {$\dotsm$\vphantom{pl}};

\node [class] (bike1) [below left=2em and 1em of bike] {path 1\vphantom{pl}};
\node [class] (bike2) [below=2em and 1em of bike] {path 2\vphantom{pl}};
\node [class] (bikeX) [below right=2em and 1em of bike] {$\dotsm$\vphantom{pl}};

\node [empty] (tramX) [below=2em of tram] {\large$\dots$};

\draw [edgestyle] (root) to (public);
\draw [edgestyle] (root) to (private);

\draw [edgestyle] (public) to (bus);
\draw [edgestyle] (public) to (tram);
\draw [edgestyle] (public) to (metro);

\draw [edgestyle] (private) to (car);
\draw [edgestyle] (private) to (bike);

\draw [edgestyle] (bus) to (bus1);
\draw [edgestyle] (bus) to (bus2);
\draw [edgestyle] (bus) to (busX);

\draw [edgestyle] (bike) to (bike1);
\draw [edgestyle] (bike) to (bike2);
\draw [edgestyle] (bike) to (bikeX);

\draw [edgestyle,dashed] (metro) to (tramX.north west);
\draw [edgestyle,dashed] (tram) to (tramX.north);
\draw [edgestyle,dashed] (car) to (tramX.north east);


\end{tikzpicture}
\caption{Grouping of alternative modes of transport by similarity.}
\label{fig:transport}
\end{figure}
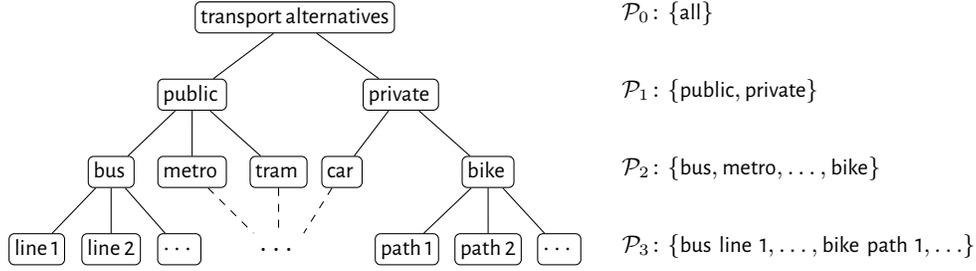


To put all this on a formal footing, we will assume that the players' set of strategies comes equipped with a hierarchy of increasingly finer partitions as follows:

\begin{definition}
\label{def:struct}
A \define{similarity structure} \textendash\ or \define{hierarchy} \textendash\ on $\pures$ is a tower of nested \emph{similarity partitions} $\classes\atlvl{\lvl}$, $\lvl=0,\dotsc,\nLvls$, of $\pures$ such that
\begin{equation}
\{\pures\}
	\eqdef \classes\atlvl{0}
	\refined \classes\atlvl{1}
	\refined \dotsb
	\refined \classes\atlvl{\nLvls}
	\defeq \setdef{\{\pure\}}{\pure \in \pures}
\end{equation}
where the notation ``$\mathcal{P}\refined \mathcal{Q}$'' means that the partition $\mathcal{Q}$ is \emph{finer} than $\mathcal{P}$, \ie every element of $\mathcal{Q}$ is a subset of some element of $\mathcal{P}$.%
\footnote{The trivial partition $\classes\atlvl{0} = \{\pures\}$ does not carry any information in itself, but it has been included for completeness and notational convenience later on.}
\end{definition}

In this definition, each partition $\classes\atlvl{\lvl}$, $\lvl=1,\dotsc,\nLvls$, captures succcessively finer traits of the elements of $\pures$:
for instance, in the commuting example of \cref{fig:transport},
the coarsest non-trivial partition $\classes\atlvl{1}$ would correspond to the distinction between public and private transport modes;
the second partition $\classes\atlvl{2} \refines \classes\atlvl{1}$ to different individual modes (bus, metro, car, \dots);
the third partition $\classes\atlvl{3} \refines \classes\atlvl{2}$ to different routes for each mode;
etc.
In this interpretation, the index $\lvl=1,\dotsc,\nLvls$ indicates the degree of similarity between two strategies:
two strategies that cannot be resolved by $\classes\atlvl{\lvl}$ are said to be \define{similar at level $\lvl$},%
\footnote{By ``resolve'' we mean here that the two strategies in question belong to different elements of $\classes\atlvl{\lvl}$.
By construction, any two distinct strategies can be resolved by $\classes\atlvl{\nLvls}$, but no strategies can be resolved by $\classes\atlvl{0}$.}
so the strategies that are similar at level $\nLvls-1$ are the \emph{most similar} (because they share all attributes of $\classes\atlvl{\nLvls-1}$), whereas the strategies that are not similar even at level $1$ are the \emph{most dissimilar} ones (since they do not share any common attributes whatsoever).
Thus, in our running example, two different bus lines have $2$ degrees of similarity, but a bike path and a bus line are not similar at any level.

For concreteness, each constituent set $\class$ of a partition $\classes\atlvl{\lvl}$ will be referred to as a \emph{similarity class of level $\lvl$}.
A similarity structure on $\pures$ may thus be represented as a disjoint union of all class/level pair $(\set,\lvl)$ for $\set\in\classes\atlvl{\lvl}$, \viz
\begin{equation}
\label{eq:struct}
\struct
	\defeq \coprod\nolimits_{\lvl=0}^{\nLvls} \classes\atlvl{\lvl}
	\equiv \union\nolimits_{\lvl=0}^{\nLvls} \setdef{(\set,\lvl)}{\set\in\classes\atlvl{\lvl}}.
\end{equation}
To simplify notation, if there is no danger of confusion, we will not differentiate between the class $\class$ and the class/level pair $(\class,\lvl)$;
by contrast, when we need to unambiguously keep track of the level index $\lvl$ of $\class$, we will write $\lvl = \attof{\class}$.

Moving forward, if a class $\parent\in\classes\atlvl{\lvl}$ contains the class $\child\in\classes\atlvl{\lvlalt}$ for some $\lvlalt>\lvl$, we will say that $\child$ is a \emph{descendant} of $\parent$ and we will write ``$\child \desc \parent$'' (resp.``$\parent \anc \child$'' when $\parent$ is an \emph{ancestor} of $\child$).
As a special case of this relation, if $\child \desc \parent$ and $\lvlalt = \lvl+1$, we will say that $\child$ is a \emph{child} of $\parent$ and we will write ``$\child \childof \parent$'' (resp.~$\parent$ is a \emph{parent} of $\child$ when $\parent \parentof \child$).
More generally, when we wish to focus on descendants sharing a certain attribute, we will write ``$\child \desceq[\lvl] \parent$'' as shorthand for the predicate ``$\child \desceq \parent$ and $\attof{\child} = \lvl$''.
In a similar vein, given an arbitrary class $\class\in\struct$, we will write $\classof[\lvl]{\class}$ for the $\lvl$-th type ancestor of $\class$, \ie the (necessarily unique) similarity class $\class\atlvl{\lvl}\in\classes\atlvl{\lvl}$ containing $\class$ (obviously, we assume here that $\attof{\class} \geq \lvl$).
Finally, we will say that two classes $\class,\classalt\in\pures$ are \emph{similar with respect to $\classes\atlvl{\lvl}$} if and only if $\classof[\lvl]{\class} = \classof[\lvl]{\classalt}$, in which case we will write $\class \sibling[\lvl] \classalt$.%
\footnote{The notations $\classof{\pure}$ and $\pure\sibling\purealt$ will be reserved for elements of $\pures$ (\ie when $\lvl = \nLvls-1$).}

Building on the above, a similarity structure on $\pures$ can be represented graphically as a rooted directed tree \textendash\ or \emph{arborescence} \textendash\ by connecting two classes $\parent,\child\in\struct$ with a directed edge $\parent\to\child$ whenever $\parent \parentof \child$ (\cf \cref{fig:tree} above).
By construction, the root of this tree is $\source$ itself,
and
the unique directed path $\source \equiv \class\atlvl{0} \parentof \class\atlvl{1} \parentof \dotsb \parentof \class\atlvl{\lvl} \equiv \class$ from $\source$ to any class $\class\in\struct$ will be referred to as the \emph{lineage} of $\class$.
For notational simplicity, we will not distinguish between $\struct$ and its graphical representation, and we will use the two interchangeably.


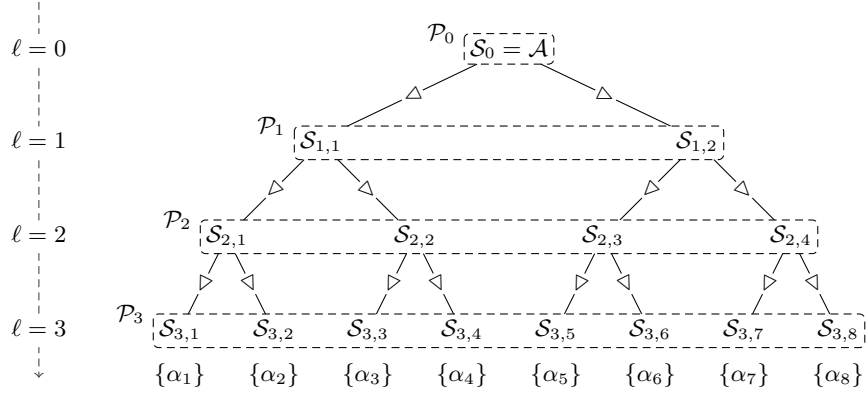
\begin{figure}[tbp]

\begin{tikzpicture}
[scale=1.25,
class/.style={inner sep=2pt},
desc/.style={rounded corners=2pt,fill=white,inner sep=2pt,align=left},
edgestyle/.style={->},
label/.style={circle,fill=white,inner sep=0pt},
nest/.style={densely dashed,rounded corners=2pt,draw=black},
connect/.style={draw=black,-},
allow upside down]

\small

\def\side{1}
\def\legendpos{-5}
\def\costhirty{0.8660256}
\def\cosfortyfive{0.7071068}

\draw [densely dashed,draw=darkgray,->] (\legendpos,.5) to (\legendpos,-3.5);
\node [desc] (lvl0) at (\legendpos,0) {$\lvl = 0 \vphantom{S_{0}}$};
\node [desc] (lvl1) at (\legendpos,-1) {$\lvl = 1 \vphantom{S_{1,1}}$};
\node [desc] (lvl2) at (\legendpos,-2) {$\lvl = 2\vphantom{S_{2,1}}$};
\node [desc] (lvl3) at (\legendpos,-3) {$\lvl = 3 \vphantom{S_{3,1}}$};

\node [class] (root) at (0,0) {$\class\atlvl{0} = \pures$};

\node [class] (11) at (-2,-1) {$\class\atlvl{1,1}$};
\node [class] (12) at (2,-1) {$\class\atlvl{1,2}$};

\node [class] (21) at (-3,-2) {$\class\atlvl{2,1}$};
\node [class] (22) at (-1,-2) {$\class\atlvl{2,2}$};
\node [class] (23) at (1,-2) {$\class\atlvl{2,3}$};
\node [class] (24) at (3,-2) {$\class\atlvl{2,4}$};

\node [class] (31) at (-3.5,-3) {$\class\atlvl{3,1}$};
\node [class] (32) at (-2.5,-3) {$\class\atlvl{3,2}$};
\node [class] (33) at (-1.5,-3) {$\class\atlvl{3,3}$};
\node [class] (34) at (-0.5,-3) {$\class\atlvl{3,4}$};
\node [class] (35) at (0.5,-3) {$\class\atlvl{3,5}$};
\node [class] (36) at (1.5,-3) {$\class\atlvl{3,6}$};
\node [class] (37) at (2.5,-3) {$\class\atlvl{3,7}$};
\node [class] (38) at (3.5,-3) {$\class\atlvl{3,8}$};

\node [class] (1) [below=1 ex of 31] {$\braces{\pure_{1}}$};
\node [class] (2) [below=1 ex of 32] {$\braces{\pure_{2}}$};
\node [class] (3) [below=1 ex of 33] {$\braces{\pure_{3}}$};
\node [class] (4) [below=1 ex of 34] {$\braces{\pure_{4}}$};
\node [class] (5) [below=1 ex of 35] {$\braces{\pure_{5}}$};
\node [class] (6) [below=1 ex of 36] {$\braces{\pure_{6}}$};
\node [class] (7) [below=1 ex of 37] {$\braces{\pure_{7}}$};
\node [class] (8) [below=1 ex of 38] {$\braces{\pure_{8}}$};

\draw [connect] (root) to node [label,midway,sloped] {$\parentof$} (11.north east);
\draw [connect] (root) to node [label,midway,sloped] {$\parentof$} (12.north west);

\draw [connect] (11) to node [label,midway,sloped] {$\parentof$} (21);
\draw [connect] (11) to node [label,midway,sloped] {$\parentof$} (22);
\draw [connect] (12) to node [label,midway,sloped] {$\parentof$} (23);
\draw [connect] (12) to node [label,midway,sloped] {$\parentof$} (24);

\draw [connect] (21) to node [label,midway,sloped] {$\parentof$} (31);
\draw [connect] (21) to node [label,midway,sloped] {$\parentof$} (32);
\draw [connect] (22) to node [label,midway,sloped] {$\parentof$} (33);
\draw [connect] (22) to node [label,midway,sloped] {$\parentof$} (34);
\draw [connect] (23) to node [label,midway,sloped] {$\parentof$} (35);
\draw [connect] (23) to node [label,midway,sloped] {$\parentof$} (36);
\draw [connect] (24) to node [label,midway,sloped] {$\parentof$} (37);
\draw [connect] (24) to node [label,midway,sloped] {$\parentof$} (38);

\draw [nest] (root.north west) node [left] {$\classes\atlvl{0}$} rectangle (root.south east);
\draw [nest] (11.north west) node [left] {$\classes\atlvl{1}$} rectangle (12.south east);
\draw [nest] (21.north west) node [left] {$\classes\atlvl{2}$} rectangle (24.south east);
\draw [nest] (31.north west) node [left] {$\classes\atlvl{3}$} rectangle (38.south east);

\end{tikzpicture}
\caption{Graphical representation of a $3$-tier similarity structure as a rooted tree.}
\label{fig:tree}
\end{figure}


\begin{example*}
Going back to the commuting example of \cref{fig:transport}, the class ``\bus'' is a child of the parent class ``\public'' (formally, $\bus\childof\public$), and it is a sibling\,/\,similar to the classes ``\metro'' and ``\tram'', both of which are also children of the class ``\public'' (so $\bus \sibling[2] \metro \sibling[2] \tram$ in the hierarchy depicted in \cref{fig:transport}).
Accordingly, the lineage of the class ``\bus'' is the two-level parent/child chain $\all \parentof \public \parentof \bus$, where the class ``\all'' bundles together all transportation modes in a single (trivial) supercluster.
\endenv
\end{example*}

\subsection{\Acl{NPPI} and the \acl{NRD}}
\label{sec:NRD}

Building on the \ac{PPI} protocol described in the previous section, we introduce below a more general family of revision protocols under which agents give additional consideration to switches to ``similar'', more familiar actions.
For example,
a commuter may be more likely to contemplate switching to a different instance of their current transport mode than a different mode altogether \citep[\eg taking a different combination of metro lines rather than switching to a bus, \cf][]{BA73,BAL85,McF81,McF01}.
Likewise,
when deciding to replace a computer or other technology with network externalities, a consumer may be more likely to consider models manufactured by the maker of their current machine than those of other manufacturers, particularly in contexts with switching costs \citep{KS94,FarKle07};
etc.

Formally, given a similarity structure $\classes\atlvl{0} \refined \classes\atlvl{1} \refined \dotsb \refined \classes\atlvl{\nLvls}$ of $\pures$ as above, we will consider a multi-stage revision protocol that unfolds as follows:

\begin{enumerate}
\item
When given an opportunity to switch strategies, a revising agent playing $\pure\in\pures$ first decides which strategies to consider as candidates for imitation.
To account for the increased familiarity of more similar strategies, the pool of candidates from $\pure$ is
\begin{equation}
\label{eq:revise}
\cand(\pure)
	= \begin{cases*}
	\classof[0]{\pure}
		&with probability $\rate_{0}$,
		\\
	\classof[1]{\pure}
		&with probability $\rate_{1}$,
		\\
	\dotso
		\\
	\classof[\nLvls-1]{\pure}
		&with probability $\rate_{\nLvls-1}$,
	\end{cases*}
\end{equation}
where
the ``intra-level'' sampling probabilities $\rate_{\lolvl}\geq0$ satisfy
\begin{equation}
\label{eq:normal}
\rate\atlvl{0} + \rate\atlvl{1} + \dotsb + \rate\atlvl{\nLvls-1}
	= 1
\end{equation}
and, as before,
\begin{equation}
\label{eq:classof}
\classof[\lolvl]{\pure}
	= \setdef{\class\atlvl{\lolvl}\in\classes\atlvl{\lolvl}}{\class\atlvl{\lolvl}\ni\pure}
\end{equation}
is the set of strategies of $\pures$ that are similar to $\pure$ at level $\lolvl$.

In words, conditioning on the degree of similarity between strategies, the revising agent considers
all possible switches \textendash\ \ie to \emph{any} strategy, irrespective of similarity with the agent's incumbent strategy \textendash\ with conditional probability $\rate_{0}$,
all switches to a strategy with at least one degree of similarity with conditional probability $\rate_{1}$,
and more generally,
all strategies with at least $\lolvl$ degrees of similarity with conditional probability $\rate_{\lolvl}$.
Specifically, in our running example, an agent taking the bus to work would first contemplate whether they want to only consider switching to another bus line ($2$ degrees of similarity),
or to some other means of public transportation (at least $1$ degree of similarity),
or to consider all possible means of commuting from home to work (\ie with no regard for similarities).

\item
Conditioned on the above choice of $\lolvl\in\lolvls$, the revising $\pure$-strategist meets another individual \textendash\ the comparator \textendash\ drawn uniformly at random from the sub-population of agents playing a strategy from the chosen similarity class $\classof[\lolvl]{\pure}$.
Hence, if the population is at state $\strat\in\strats$, the probability that a revising $\pure$-strategist encounters a $\purealt$-strategist is
\begin{align}
\label{eq:pmeet}
\pmeet_{\pure\purealt}(\strat)
	&= \rate\atlvl{0} \, \strat_{\purealt}
		+ \rate\atlvl{1} \strat_{\purealt} \frac{\oneof{\purealt\sibling[1]\pure}}{\sum_{\alt\pure\sibling[1]\pure} \strat_{\alt\pure}}
		+ \dotsb
		+ \rate\atlvl{\nLvls-1} \strat_{\purealt} \frac{\oneof{\purealt\sibling[\nLvls-1]\pure}}{\sum_{\alt\pure\sibling[\nLvls-1]\pure} \strat_{\alt\pure}}
	\notag\\
	&= \sum_{\lolvl=0}^{\nLvls-1} \rate\atlvl{\lolvl} \strat_{\purealt \given \classof[\lolvl]{\pure}}
\end{align}
where, as per \eqref{eq:strat-cond}, the notation
\begin{equation}
\label{eq:strat-cond}
\strat_{\purealt \given \classof[\lolvl]{\pure}}
	= \oneof{\purealt \sibling[\lolvl] \pure} \frac{\strat_{\purealt}}{\strat_{\classof[\lolvl]{\pure}}}
\end{equation}
stands for the relative frequency of $\purealt$-strategists within the sub-population of $\classof[\lolvl]{\pure}$-type strategists (\ie the part of the population playing a strategy with at least $\lolvl$ degrees of similarity to $\pure$).

\item
Finally, as in the case of \eqref{eq:PPI}, the revising agent observes the comparator's payoff:
if it is higher than their own, the revising agent switches to the comparator's strategy with probability proportional to the observed payoff excess;
otherwise, the revising agent skips the opportunity to switch.
\end{enumerate}

Hence, putting together all of the above yields the \acli{NPPI} protocol
\begin{equation}
\label{eq:NPPI}
\tag{NPPI}
\switch_{\pure\purealt}(\py,\strat)
	= \pmeet_{\pure\purealt}(\strat) \cdot \pospart{\py_{\purealt} - \py_{\pure}}
	= \sum_{\lolvl=0}^{\nLvls-1}
		\rate\atlvl{\lolvl}\,
		\strat_{\purealt \given \classof[\lolvl]{\pure}}
		\pospart{\py_{\purealt} - \py_{\pure}},
\end{equation}
where, to avoid trivialities, we assume that $\rate\atlvl{0} > 0$ (so all strategies are sampled with positive probability).
Then, to derive the induced dynamics, a direct substitution into \eqref{eq:MD} gives
\begin{align}
\dot\strat_{\pure}
	&= \sum_{\lolvl=0}^{\nLvls-1}
		\rate\atlvl{\lolvl} \strat_{\pure}
		\bracks*{
			\sum_{\purealt\in\pures}
				\strat_{\purealt\given\classof[\lolvl]{\pure}}
				\pospart{\payv_{\pure}(\strat) - \payv_{\purealt}(\strat)}
			- \sum_{\purealt\in\pures}
				\strat_{\purealt\given\classof[\lolvl]{\pure}}
				\pospart{\payv_{\purealt}(\strat) - \payv_{\pure}(\strat)}
		} 
\end{align}
so,
after rearranging,
we obtain the \acli{NRD}
\begin{equation}
\label{eq:NRD}
\tag{NRD}
\dot \strat_{\pure}
	= \strat_{\pure}
		\sum_{\lolvl=0}^{\nLvls-1}
		\rate\atlvl{\lolvl}\,
		\bracks{\payv_{\pure}(\strat) - \classpay_{\classof[\lolvl]{\pure}}(\strat)}
\end{equation}
where
\begin{equation}
\label{eq:pay-class}
\classpay_{\class}(\strat)
	= \sum_{\pure\in\class} \strat_{\pure \given \class} \payv_{\pure}(\strat)
	= \frac{1}{\strat_{\class}} \sum_{\pure\in\class} \strat_{\pure} \payv_{\pure}(\strat)
\end{equation}
denotes the mean payoff of the class $\class \subseteq \pures$.
The protocol \eqref{eq:NPPI} and the dynamics \eqref{eq:NRD} will be our main object of study, so
we close this section with some remarks intended to highlight their most salient features.

\begin{remark}
We begin by noting that \eqref{eq:PPI} and \eqref{eq:RD} are immediately recovered from \eqref{eq:NPPI} and \eqref{eq:NRD} respectively by taking $\rate\atlvl{0}=1$ and $\rate\atlvl{1} = \dotsb = \rate\atlvl{\nLvls-1} = 0$.
In this case, revising agents do not take into account similarities between strategies when selecting a comparator, so it stands to reason that the induced dynamics boil down to \eqref{eq:RD}.
\endenv
\end{remark}

\begin{remark}
Moving forward, it is important to note that comparison probabilities under \eqref{eq:NPPI} are always higher between similar strategies compared to dissimilar ones, and the closer the similarity, the higher the probability of carrying out a comparison.
To make this precise, define the \define{degree of similarity} between $\pure$ and $\purealt$ as
\begin{equation}
\label{eq:degree}
\deg(\pure,\purealt)
	\defeq \max\setdef{\lolvl=0,\dotsc,\nLvls-1}{\pure \sibling[\lolvl] \purealt}
\end{equation}
\ie $\pure$ and $\purealt$ are similar at level $\deg(\pure,\purealt)$ but not at level $\deg(\pure,\purealt)+1$.
Then the probability that $\purealt$ is considered by an $\pure$-strategist as a candidate for imitation is
\begin{equation}
\label{eq:pcand}
\pcand_{\pure\purealt}
	= \probof{\purealt \in \cand(\pure)}
	= \rate_{0} + \rate_{1} + \dotsb + \rate_{\deg(\pure,\purealt)}
\end{equation}
so $\pcand_{\pure\purealt}$ is increasing in $\deg(\pure,\purealt)$, just like the encounter probability $\pmeet_{\pure\purealt}$ of \eqref{eq:pmeet}.

On the other hand, the agents' imitation rates under \eqref{eq:NPPI} also depend on the payoff excess $\pospart{\py_{\purealt} - \py_{\pure}}$ of $\purealt$ over $\pure$;
more precisely, by \eqref{eq:strat-cond} we have
\begin{equation}
\label{eq:imit-NPPI}
\imit_{\pure\purealt}(\py,\strat)
	= \frac{\switch_{\pure\purealt}(\py,\strat)}{\strat_{\purealt}}
	= \frac{\pmeet_{\pure\purealt}(\strat)}{\strat_{\purealt}}
		\cdot \pospart{\py_{\purealt} - \py_{\pure}}
	= \pospart{\py_{\purealt} - \py_{\pure}}
		\cdot \!\!\sum_{\lolvl=0}^{\deg(\pure,\purealt)}\!\!
			\frac{\rate\atlvl{\lolvl}}{\strat_{\classof[\lolvl]{\pure}}}
\end{equation}
so the total imitation rate is not always higher for more similar strategies (after all, a less familiar strategy could still be a favorable candidate for imitation, even if it is encountered with lower probability).
However, for any \emph{given} value of $\pospart{\py_{\purealt} - \py_{\pure}}$, the imitation rate of $\purealt$-strategists by $\pure$-strategists is highest when $\pure$ and $\purealt$ are \emph{most} similar (\eg different metro lines), and lowest when $\pure$ and $\purealt$ are \emph{least} similar (\eg different transport modes altogether).
This increased affinity toward similar strategies is an \emph{intrinsic} feature of the agents' revision mechanism and it is due to the fact that, ceteris paribus, an agent is more likely to
compare strategies with another agent if the other agent's strategy is more similar to their own. 
\endenv
\end{remark}

\begin{remark}
Building on the previous remark, it is also natural to ask what happens if a revising agent meets another agent from the population uniformly at random (that is, $\pmeet_{\pure\purealt}(\strat) = \strat_{\purealt}$) but is \emph{extrinsically} more likely to imitate strategies that are more similar to their own.
In this case, the agents' imitation rates can be written as
\begin{equation}
\label{eq:imit-extr}
\imit_{\pure\purealt}(\py,\strat)
	= \Extr_{\pure\purealt}
		\cdot \frac{\pmeet_{\pure\purealt}(\strat)}{\strat_{\purealt}}
		\cdot \pospart{\py_{\purealt} - \py_{\pure}}
	= \Extr_{\pure\purealt} \, \pospart{\py_{\purealt} - \py_{\pure}}
\end{equation}
where the imitation coefficients $\Extr_{\pure\purealt}$ are of the form
\begin{equation}
\label{eq:imit-coeff}
\Extr_{\pure\purealt}
	= \extr_{0} + \extr_{1} + \dotsb + \extr_{\deg(\pure,\purealt)}
\end{equation}
for some $\extr_{0},\dotsc,\extr_{\nLvls-1} \geq 0$ (so $\Extr_{\pure\purealt}$ increases with the degree of similarity between $\pure$ and $\purealt$).

The expression \eqref{eq:imit-extr} for $\imit_{\pure\purealt}$ inflates the probability of imitating a similar strategy by a factor of $\Extr_{\pure\purealt}$, so the revising agent is more likely to imitate the strategy of another agent if the other agent's strategy is more similar to their own.
Mutatis mutandis, this observation also applies to the agents' imitation rates \eqref{eq:imit-NPPI} under \eqref{eq:NPPI}, but where the two models differ is that, in \eqref{eq:imit-NPPI}, the reinforcement of similar strategies arises \emph{intrinsically}, as a consequence of the fact that an agent is more likely to compare strategies with another agent if the other agent's strategy is more similar to their own.
By contrast, the reinforcement of similar strategies in \eqref{eq:imit-extr} is \emph{extrinsic}, and it occurs \emph{despite} the fact that players meet other players uniformly at random.

The difference between these two models \textendash\ intrinsic \vs extrinsic \textendash\ can also be seen in the induced dynamics.
Indeed, substituting $\switch_{\pure\purealt} = \imit_{\pure\purealt} \strat_{\purealt}$ in \eqref{eq:MD} and rearranging yields the dynamics
\begin{equation}
\label{eq:NRD-extr}
\tag{\ref*{eq:NRD}$_{\textup{extr}}$}
\dot \strat_{\pure}
	= \strat_{\pure}
		\sum_{\lolvl=0}^{\nLvls-1}
		\extr\atlvl{\lolvl} \, \strat_{\classof[\lolvl]{\pure}}
		\bracks{\payv_{\pure}(\strat) - \classpay_{\classof[\lolvl]{\pure}}(\strat)}
\end{equation}
where the subscript ``$\textrm{extr}$'' alludes to the extrinsic nature of the model \eqref{eq:imit-extr}.
Compared to \eqref{eq:NRD}, the main factor setting \eqref{eq:NRD-extr} apart is that the state-\emph{independent} coefficients $\rate\atlvl{\lolvl}$ in the former are replaced by the state-\emph{dependent} factors $\extr\atlvl{\lolvl} \strat_{\classof[\lolvl]{\pure}}$ in the latter.
Since $\strat_{\classof[\lolvl]{\pure}} \leq 1$, this adjustment means that the impact of the intra-class revision terms $\payv_{\pure}(\strat) - \classpay_{\classof[\lolvl]{\pure}}(\strat)$ is toned down in \eqref{eq:NRD-extr} relative to \eqref{eq:NRD}.
In turn, this reflects the fact that comparisons to similar strategies occur less frequently in the extrinsic model \eqref{eq:imit-extr} rather than the intrinsic model \eqref{eq:imit-NPPI}, so the relative ranking of strategies within a similarity class is likewise less impactful in \eqref{eq:NRD-extr} than it is in \eqref{eq:NRD}.

The extrinsic dynamics \eqref{eq:NRD-extr} are clearly interesting in their own right, but a full analysis would take us too far afield, so we do not undertake it here.
\endenv
\end{remark}

\begin{remark}
It is also worth noting that summing \eqref{eq:NRD} over all alternatives in a given similarity class $\class\in\classes\atlvl{\lvl}$ yields the similarity class dynamics
\begin{equation}
\label{eq:NRD-class}
\dot\strat_{\class}
	= \sum_{\pure\in\class} \dot\strat_{\pure}
	= \strat_{\class}
		\sum_{\lolvl=0}^{\lvl-1}
		\rate\atlvl{\lolvl}\,
		\bracks{\classpay_{\class}(\strat) - \classpay_{\classof[\lolvl]{\class}}(\strat)}
\end{equation}
where, as before, $\classof[\lolvl]{\class}$ denotes the $\lolvl$-th level ancestor of $\class$.
These dynamics are formally analogous to \eqref{eq:NRD} with $\lvl$ levels of nesting instead of $\nLvls$, and with a coarser similarity hierarchy that does not differentiate actions beyond $\lvl$ degrees of similarity.
However, the average payoffs $\classpay_{\class}(\strat)$ and $\classpay_{\classof[\lolvl]{\pure}}(\strat)$ appearing in \eqref{eq:NRD-class} depend on the \emph{complete} population state $\strat \in \strats$, not only on the aggregate population shares $\strat_{\class}$.
Because of this, the evolution of the population cannot be described by \eqref{eq:NRD-class} alone if $\lvl < \nLvls$, \ie \eqref{eq:NRD-class} does not stand on its own as an autonomous dynamical system.
\endenv
\end{remark}

\begin{remark}
As a final point, we note that the per capita growth rate $\growth_{\pure}(\strat) = \dot\strat_{\pure}/\strat_{\pure}$ of $\pure\in\pures$ under \eqref{eq:NRD} is
\begin{equation}
\label{eq:grow}
\growth_{\pure}(\strat)
	= \sum_{\lolvl=0}^{\nLvls-1}
		\rate\atlvl{\lolvl}\,
		\bracks{\payv_{\pure}(\strat) - \classpay_{\classof[\lolvl]{\pure}}(\strat)}
	= \payv_{\pure}(\strat)
		- \sum_{\lolvl=0}^{\nLvls-1} \rate\atlvl{\lolvl}
			\,\classpay_{\classof[\lolvl]{\pure}}(\strat)
\end{equation}
\ie $\growth_{\pure}(\strat)$ is a convex combination of payoff differences between $\payv_{\pure}(\strat)$ and the sub-population means $\classpay_{\classof[\lolvl]{\pure}}(\strat)$ evaluated at all levels of resemblance to $\pure$, from the maximum possible degree (when $\lolvl=\nLvls-1$) to no similarities whatsoever (when $\lolvl=0$).
An important consequence of this is that, for any pair of strategies $\pure,\purealt\in\pures$, we have
\begin{equation}
\label{eq:growdiff}
\growth_{\pure}(\strat) - \growth_{\purealt}(\strat)
	= \payv_{\pure}(\strat) - \payv_{\purealt}(\strat)
		- \sum_{\lolvl=0}^{\nLvls-1} \rate\atlvl{\lolvl}
			\,\bracks{\classpay_{\classof[\lolvl]{\pure}}(\strat) - \classpay_{\classof[\lolvl]{\purealt}}(\strat)}
	\,.
\end{equation}
Clearly, if $\pure$ and $\purealt$ are \emph{siblings} \textendash\ that is, children of the same parent, so $\classof[\lolvl]{\pure} = \classof[\lolvl]{\purealt}$ for all $\lolvl= 0,\dotsc,\nLvls-1$ \textendash\ \cref{eq:growdiff} gives $\growth_{\pure}(\strat) - \growth_{\purealt}(\strat) = \payv_{\pure}(\strat) - \payv_{\purealt}(\strat)$, so the dynamics \eqref{eq:NRD} exhibit payoff-monotone growth rates within the \emph{finest} non-trivial similarity class.%
\footnote{Following \citet{San10,San15} and \citet{Wei95}, payoff monotonicity means here that $\growth_{\pure}(\strat) \geq \growth_{\purealt}(\strat)$ if and only if $\payv_{\pure}(\strat) \geq \payv_{\purealt}(\strat)$.
In view of \eqref{eq:imit-NPPI}, \citet{Wei95} would call \eqref{eq:NPPI} an ``imitative'' protocol, but this clashes with \citet[Chap.~5.4]{San10}, where a payoff monotonicity condition is appended to the definition of ``imitative'';
for a detailed discussion of
this issue,
see \citet{MV22}.}
On the other hand, if $\pure$ and $\purealt$ are not siblings, the convex combination of class payoff differences in \eqref{eq:growdiff} could be arbitrarily large relative to $\payv_{\pure}(\strat) - \payv_{\purealt}(\strat)$, and with a different sign.
Thus, in stark contrast to \eqref{eq:RD}, the nested dynamics \eqref{eq:NRD} \emph{are not payoff-monotonic}, a fact which poses significant obstacles to their analysis.
We will revisit this point in the next section.
\endenv
\end{remark}

\section{Rationality analysis}
\label{sec:analysis}

\subsection{Statement of results and discussion}
\label{sec:results}

We are now in a position to state and prove our main results concerning the long-run rationality properties of \eqref{eq:NRD}.
These are as follows:%
\footnote{We recall here very briefly that $\eq\in\strats$ is said to be
\begin{enumerate*}
[\upshape(\itshape i\hspace*{.5pt}\upshape)]
\item
\define{Lyapunov stable} (or simply \define{stable}) if every solution trajectory $\orbit{\time}$ that starts close enough to $\eq$ remains close enough to $\eq$ for all $\time\geq0$;
\item
\define{attracting} if $\lim_{\time\to\infty} \orbit{\time} = \eq$ for every solution trajectory $\orbit{\time}$ that starts close enough to $\eq$;
and
\item
\emph{asymptotically stable} if it is both stable and attracting.
\end{enumerate*}
For an introduction to the theory of dynamical systems, we refer the reader to \citet{Shu87} and \citet{HSD04}.}

\begin{theorem}
\label{thm:NRD}
Let $\game\equiv\gamefull$ be a population game, let $\struct$ be a similarity structure on $\pures$, and let $\orbit{\time}$ be an interior solution of the \acl{NRD} \eqref{eq:NRD} for $\game$.
Then:
\begin{enumerate}
[label={\arabic*.},ref=\arabic*]

\item
\label[part]{part:dom}
If $\pure \in \pures$ is strictly dominated, then $\lim_{\time\to\infty} \orbit[\pure]{\time} = 0$.

\item
\label[part]{part:stat}
A state $\eq\in\strats$ is stationary if and only if it is a \acl{RE} of $\game$.

\item
\label[part]{part:limit}
If $\orbit{\time}$ converges to $\eq$ as $\time\to\infty$, then $\eq$ is a \acl{NE} of $\game$.

\item
\label[part]{part:Lyap}
If $\eq\in\strats$ is Lyapunov stable, then it is a \acl{NE} of $\game$.

\item
\label[part]{part:ESS}
If $\eq$ is an \ac{ESS} of $\game$, it is asymptotically stable.
In particular, if $\eq$ is a \ac{GESS} of $\game$, it attracts all initial conditions with $\supp(\strat(\tstart)) \supseteq \supp(\eq)$.

\item
\label[part]{part:pot}
If $\game$ is a potential game with potential function $\pot$, then $\pot(\orbit{\time})$ is nondecreasing and $\orbit{\time}$ converges to a connected set of restricted equilibria of $\game$.

\item
\label[part]{part:mon}
If $\game$ is strictly monotone, its \textpar{necessarily unique} \acl{NE} attracts almost all initial conditions \textendash\ and, in particular, all initial conditions with full support.
\end{enumerate}
\end{theorem}


Before discussing the proof of \cref{thm:NRD}, it is worth noting that, as we explained in the previous section, the dynamics \eqref{eq:NRD} are \emph{not} payoff-monotonic.
On that account, it is fairly surprising to see that they satisfy a range of properties that have long been associated with payoff monotonicity (such as the elimination of strictly dominated strategies and the stability properties of \aclp{NE}).

To illustrate the intricacies that this entails, consider the simple case where $\pure$ is strictly dominated by $\purealt$ with $\payv_{\purealt}(\strat) - \payv_{\pure}(\strat) = \paydiff > 0$ for all $\strat\in\strats$.
Then, by \eqref{eq:growdiff}, we readily get
\begin{align}
\label{eq:domdiff}
\frac{d}{dt} \log\frac{\strat_{\pure}}{\strat_{\purealt}}
	= \frac{\dot\strat_{\pure}}{\strat_{\pure}} - \frac{\dot\strat_{\purealt}}{\strat_{\purealt}}
	&= -\paydiff
		- \sum_{\lolvl=0}^{\nLvls-1} \rate\atlvl{\lolvl}
			\,\bracks{\classpay_{\classof[\lolvl]{\pure}}(\strat) - \classpay_{\classof[\lolvl]{\purealt}}(\strat)}
	\,.
\end{align}
Now, if $\pure$ and $\purealt$ are similarity siblings (that is, $\pure \sibling[\nLvls-1] \purealt$), we will have $\classof[\lolvl]{\pure} = \classof[\lolvl]{\purealt}$ for all $\lolvl=0,\dotsc,\nLvls-1$, so \eqref{eq:domdiff} gives
\begin{equation}
\label{eq:domdiff-1}
\orbit[\pure]{\time}
	\leq \exp(\offset-\paydiff\time)
\end{equation}
for some constant $\offset\in\R$ depending on the initialization of \eqref{eq:NRD}.
Thus, since \eqref{eq:domdiff-1} can be tight in the limit $\time\to\infty$ if $\orbit[\purealt]{\time}\to1$, we conclude that $\pure$ becomes extinct in the long run and the worst-case rate of extinction is $\Theta(\exp(-\paydiff\time))$.

On the other hand, if $\pure$ and $\purealt$ only share, say, $\nLvls-2$ degrees of similarity (\eg they are different means of transportation but both public, like a bus line and a metro line in our running example), \cref{eq:domdiff} instead yields
\begin{equation}
\label{eq:domdiff-2}
\frac{d}{dt} \log\frac{\strat_{\pure}}{\strat_{\purealt}}
	= -\paydiff
		- \rate_{\nLvls-1} \bracks{\classpay_{\classof[\nLvls-1]{\pure}}(\strat) - \classpay_{\classof[\nLvls-1]{\purealt}}(\strat)}
	\,.
\end{equation}
In this case, there is no a priori reason to assume that, as a group, the parent class $\classof[\nLvls-1]{\purealt}$ of $\purealt$ is doing better in terms of payoffs than the parent class $\classof[\nLvls-1]{\pure}$ of $\pure$ (for instance, it may well happen that a certain metro line is dominated by a specific bus line, while all other metro lines dominate all bus lines).
Because of this, the class payoff difference $\classpay_{\classof[\nLvls-1]{\pure}}(\strat) - \classpay_{\classof[\nLvls-1]{\purealt}}(\strat)$ in \eqref{eq:domdiff-2} could cancel out the domination margin $\paydiff$ between $\purealt$ and $\pure$, in which case
the per capita growth rate of $\pure$ could even \emph{exceed} the per capita growth rate of $\purealt$, although $\pure$ is dominated by $\purealt$.

In view of the above, the fact that dominated strategies become extinct under \eqref{eq:NRD} irrespective of similarities may appear somewhat puzzling \textendash\ and, perhaps, unexpected.
Nonetheless, the degree of similarity between strategies plays an important role in the dynamics;
as we show below, it controls \emph{the rate of extinction} of dominated strategies.

\begin{proposition}
\label{prop:domrate}
Let $\orbit{\time}$ be an interior solution of \eqref{eq:NRD}, and suppose that $\pure\in\pures$ is dominated by $\purealt\in\pures$ with a margin of $\paydiff$, that is,
\begin{equation}
\label{eq:paydiff}
\txs
\paydiff
	= \min_{\strat\in\strats} \bracks{\payv_{\purealt}(\strat) - \payv_{\pure}(\strat)}
	> 0
	\,.
\end{equation}
Then $\pure$ becomes extinct along $\orbit{\time}$ at a rate of
\begin{equation}
\label{eq:domrate}
\txs
\orbit[\pure]{\time}
	\leq \exp\parens*{\offset - \sum_{\lolvl=0}^{\lvl} \rate\atlvl{\lolvl} \cdot \paydiff\time}
\end{equation}
where
$\lvl = \deg(\pure,\purealt) \equiv \max\setdef{\lolvl=0,\dotsc,\nLvls-1}{\pure \sibling[\lolvl] \purealt}$ is the degree of similarity between $\pure$ and $\purealt$,
and
$\offset \in \R$ is a constant depending on $\orbit{\tstart}$.
\end{proposition}

To streamline our presentation, we defer the proof of \cref{prop:domrate} to \cref{app:NRD}.
Instead, for our purposes, it is more important to note that the higher the degree of similarity between a dominated strategy and the strategy that dominates it, the faster the rate of extinction of said strategy.
This is due to the revision mechanism that is driving the dynamics \eqref{eq:NRD}:
because an agent is less likely to initiate a comparison with less similar strategies, the probability of dropping a dominated strategy is likewise reduced if the strategy that dominates it has a small degree of similarity.
This decrease is captured by the coefficient $\sumrate\atlvl{\lvl} \equiv \rate\atlvl{0} + \dotsb + \rate\atlvl{\lvl}$ that modulates the domination margin $\paydiff$ in \eqref{eq:domrate}, and which is strictly less than $1$ if $\pure$ and $\purealt$ are similar to the highest possible degree.


\begin{figure*}[t]
\centering
\footnotesize
\begin{subfigure}{\textwidth}
\includegraphics[width=.48\textwidth]{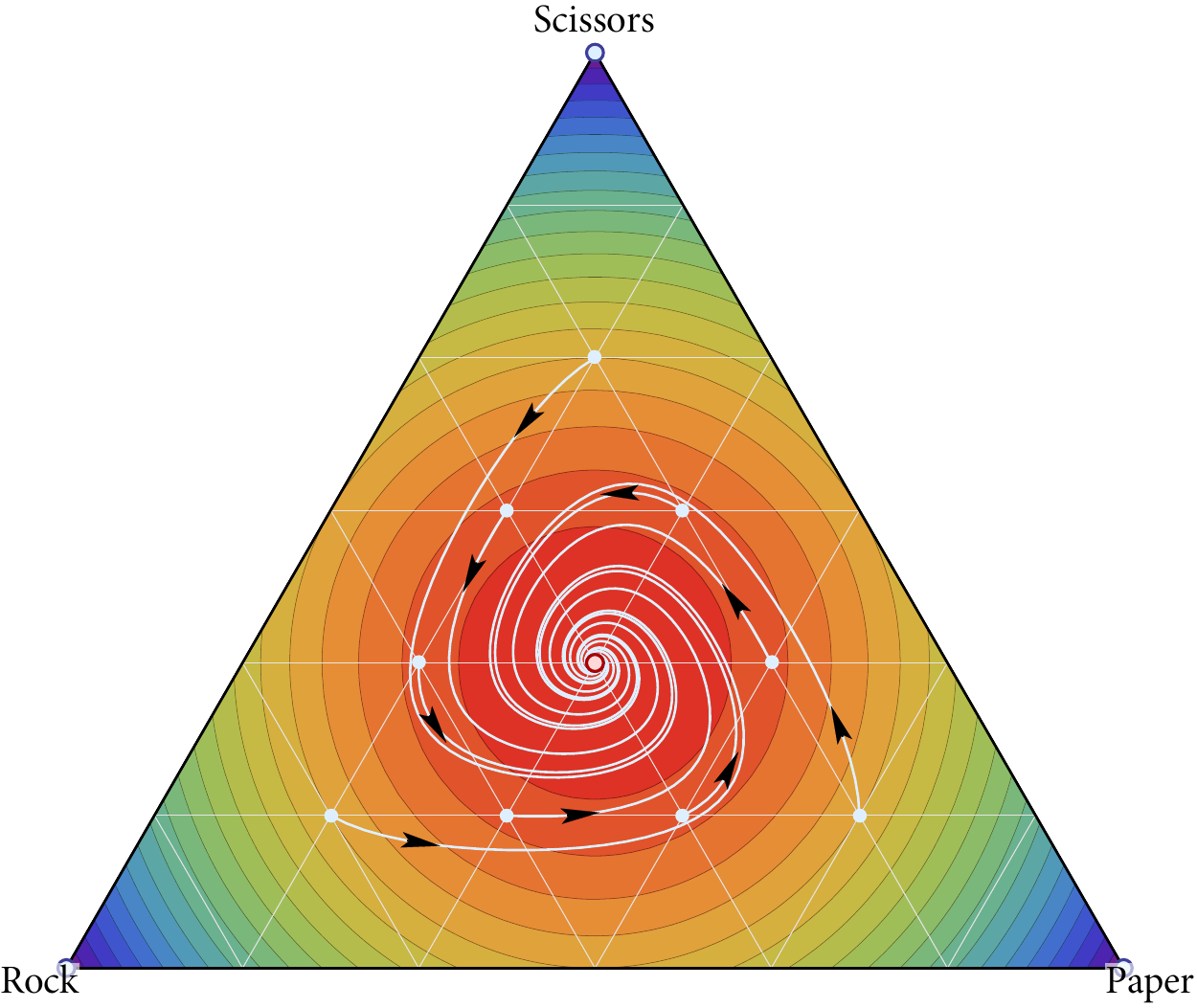}
\hfill
\includegraphics[width=.48\textwidth]{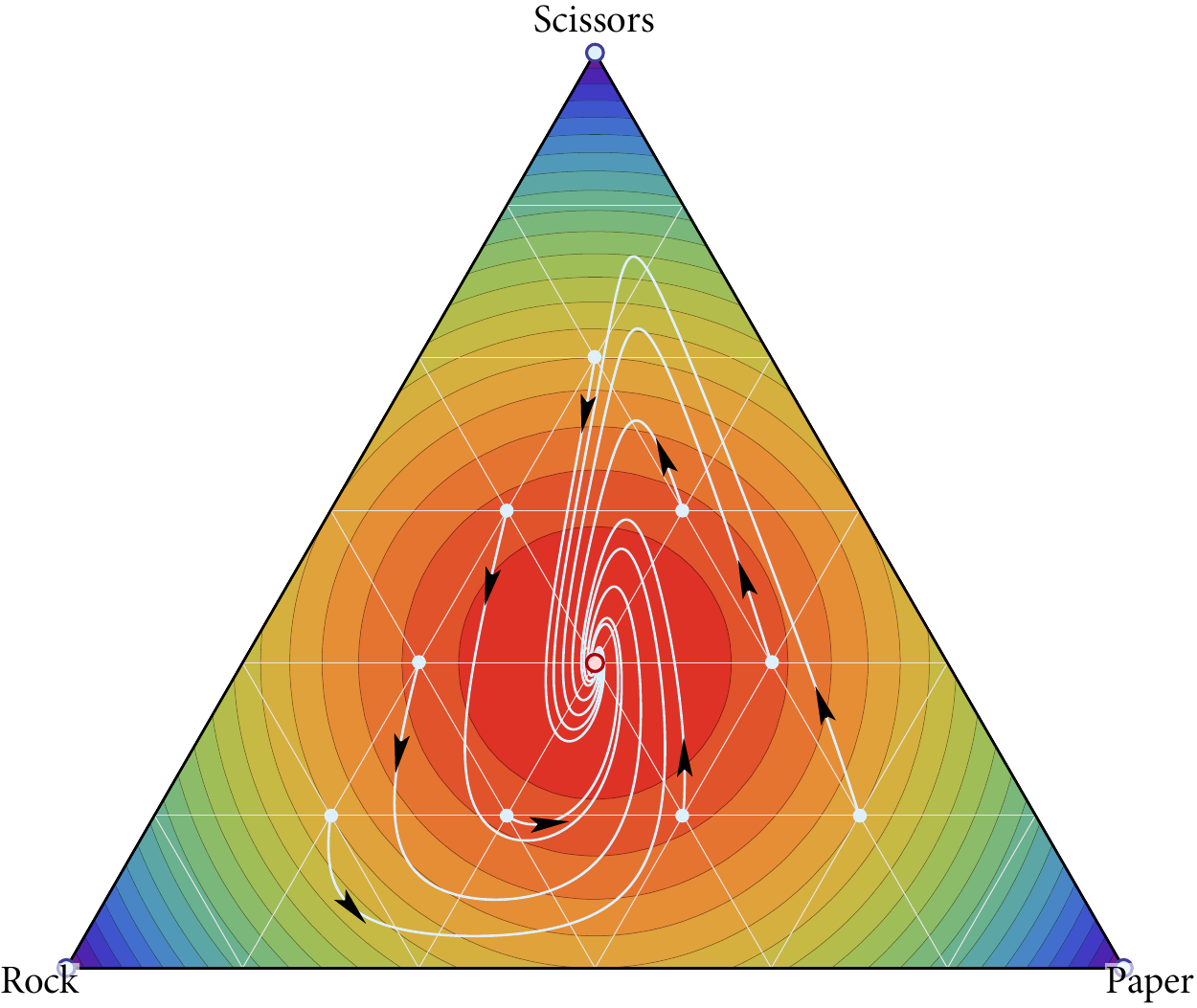}
\caption{Solution orbits of \eqref{eq:RD} and \eqref{eq:NRD} in a ``good Rock-Paper-Scissors'' game.}
\label{fig:portraits-RPS}
\end{subfigure}
\begin{subfigure}{\textwidth}
\includegraphics[width=.48\textwidth]{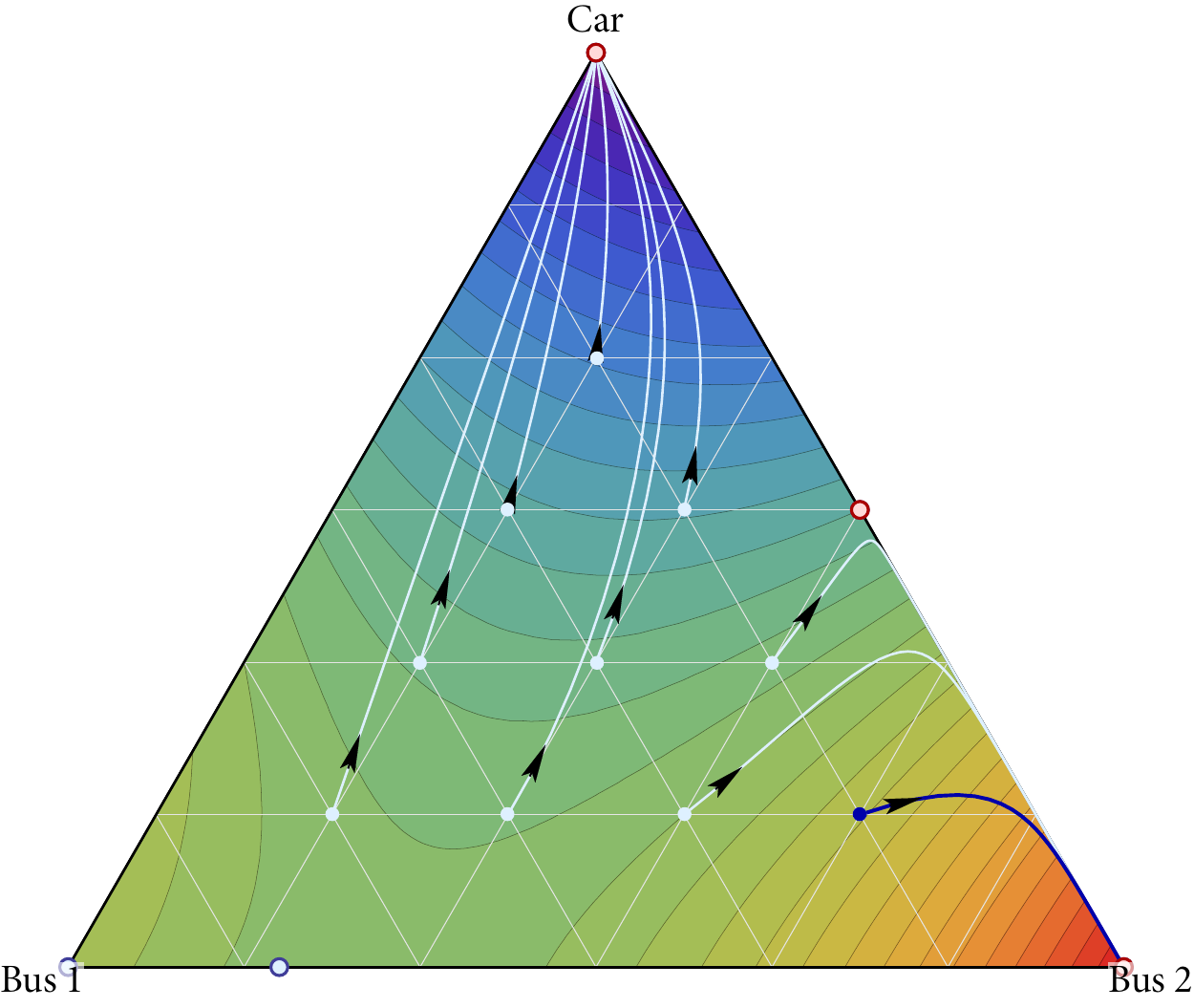}
\hfill
\includegraphics[width=.48\textwidth]{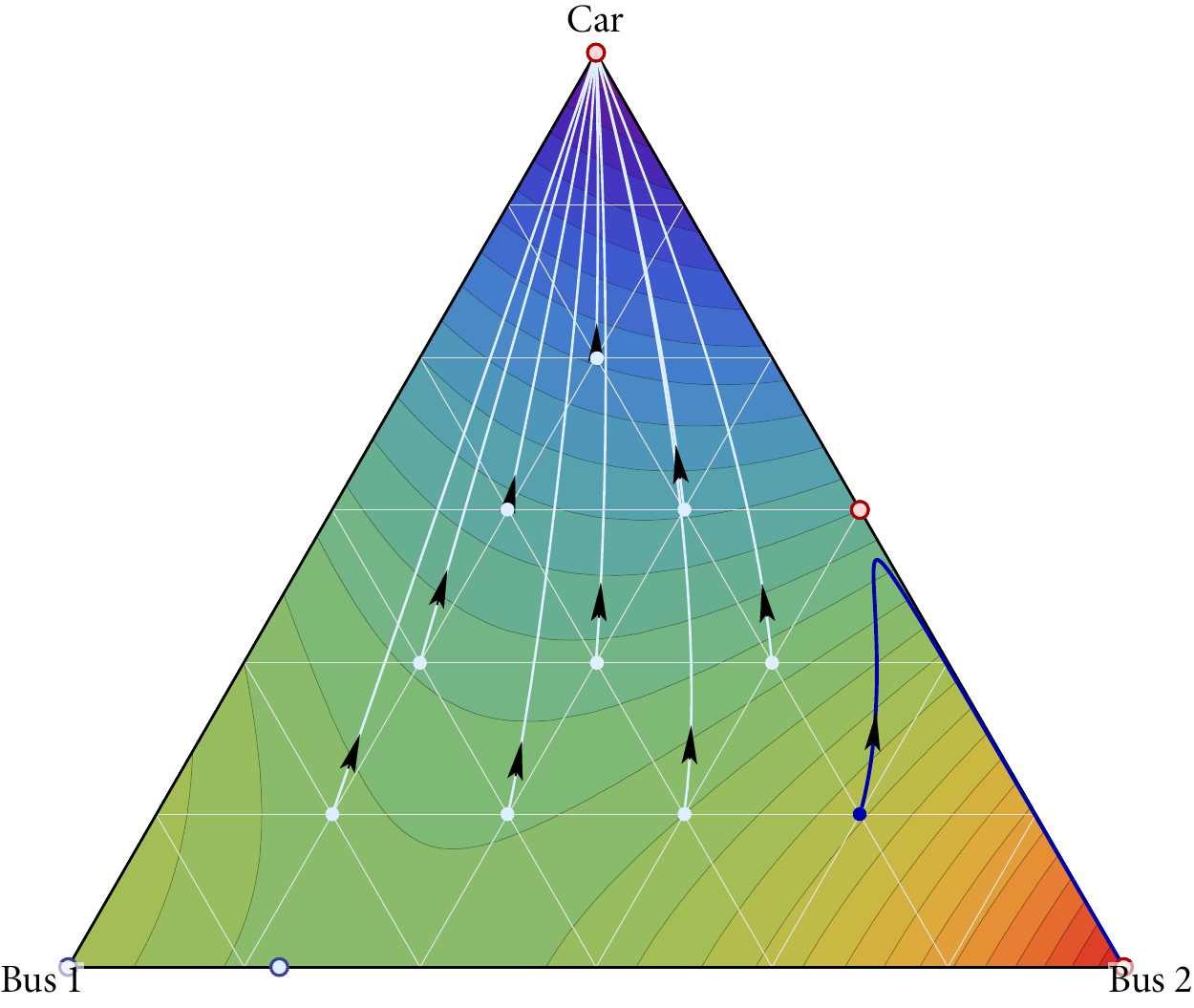}
\caption{Solution orbits of \eqref{eq:RD} and \eqref{eq:NRD} in a commuting game.}
\label{fig:portraits-BBC}
\end{subfigure}
\caption{Solution orbits of the \acl{RD} (left) and the \acl{NRD} (right) in a game of ``good Rock-Paper-Scissors'' (top) and the commuting game described in the text (bottom).
\Aclp{NE} are depicted in red, stationary points in blue;
the contours represent the population mean payoff (higher values shifted to red).
In the RPS game, ``P'' has been arbitrarily grouped with ``S'' to illustrate the distortion incurred by the implicit similarity bias of \eqref{eq:NRD}.
In all cases, \eqref{eq:NRD} has been run with intra-level sampling probabilities $\rate_{0}=1/4$ and $\rate_{1} = 3/4$.
The evolution of the population shares of the highlighted orbit in the commuting game is shown in \cref{fig:domrate}.}
\label{fig:portraits}
\end{figure*}



\begin{figure*}[t]
\centering
\footnotesize
\begin{subfigure}{.48\textwidth}
\includegraphics[width=\textwidth]{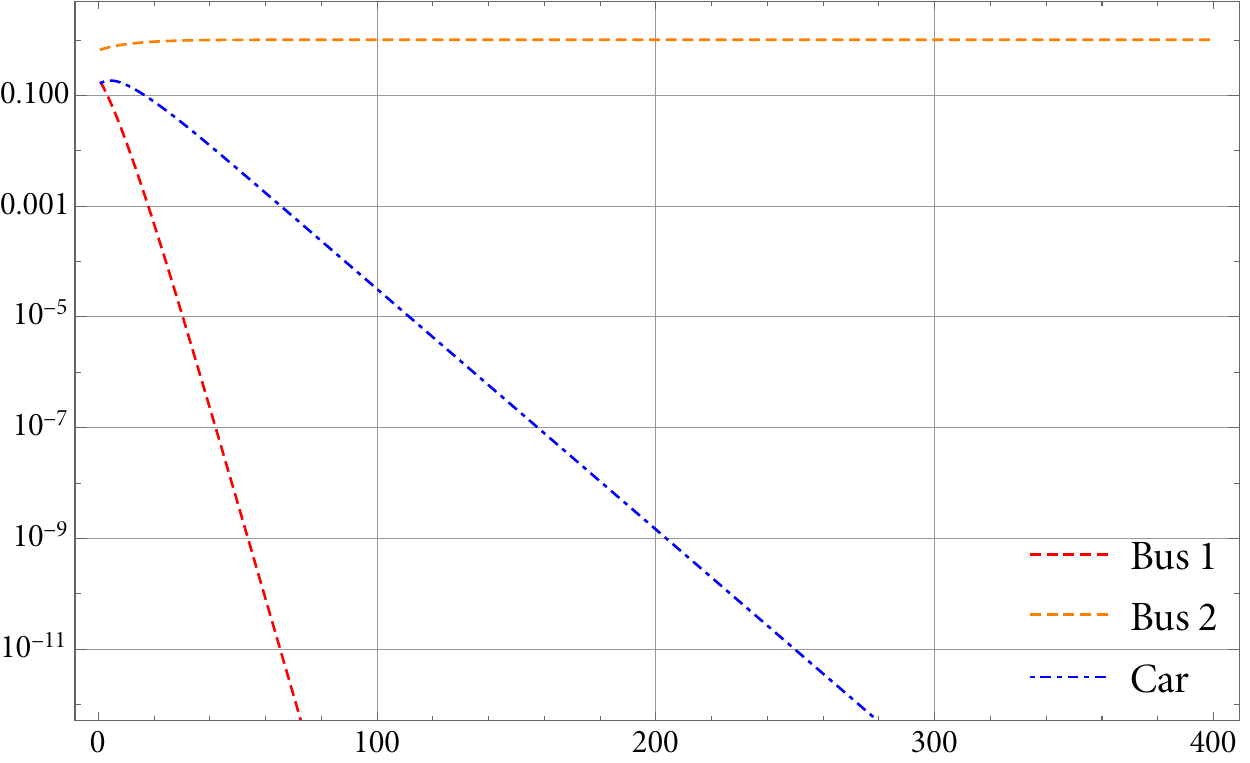}
\caption{Population shares under \eqref{eq:RD}.}
\end{subfigure}
\hfill
\begin{subfigure}{.48\textwidth}
\includegraphics[width=\textwidth]{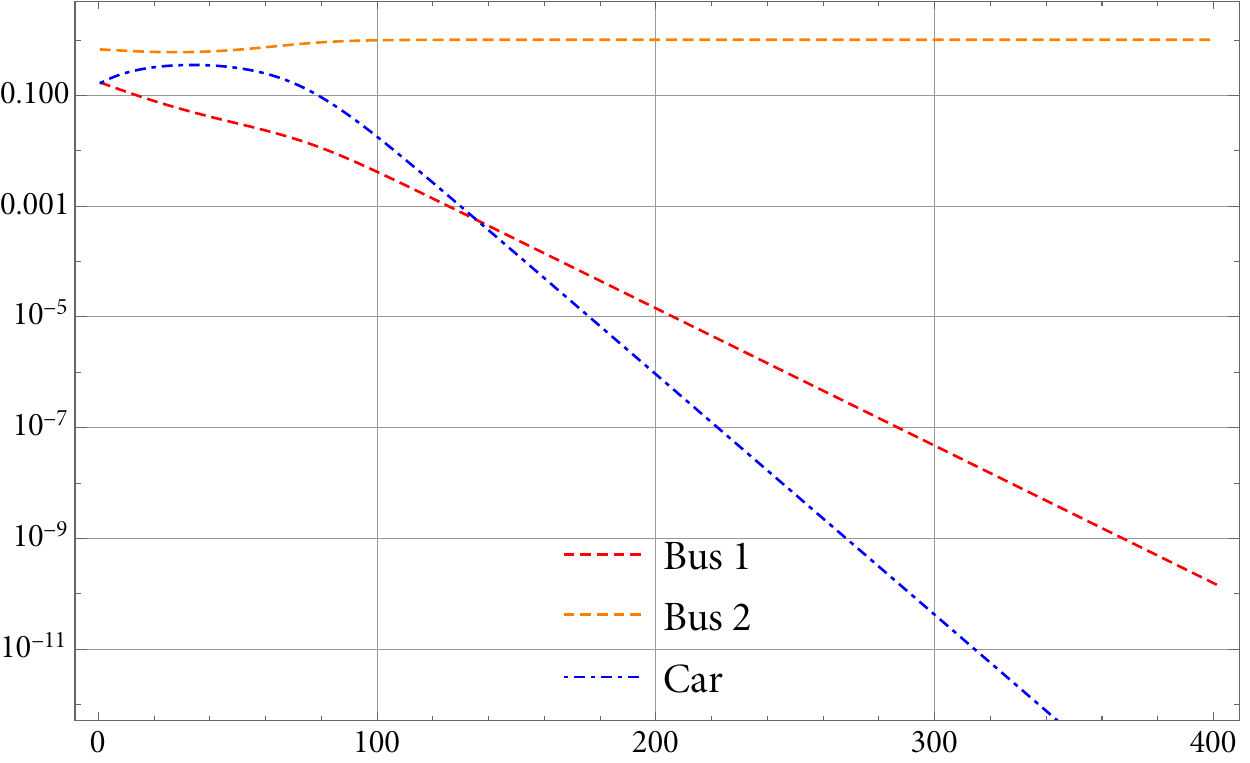}
\caption{Population shares under \eqref{eq:NRD}.}
\end{subfigure}
\caption{The rate of extinction of dominated strategies over time under the \acl{RD} (left) and the \acl{NRD} (right).
Population shares are computed for the highlighted orbits of \cref{fig:portraits};
axes are log-linear, indicating an exponential rate of convergence to equilibrium, and an exponential rate of extinction of dominated strategies (with the slope of each line capturing the extinction exponent).
In tune with \cref{prop:domrate}, we observe that the first bus line becomes extinct at a significantly faster rate when similarities are not taken into account.}
\label{fig:domrate}
\end{figure*}


\begin{remark*}
A similar result can be obtained for the \emph{rate of convergence} of \eqref{eq:NRD} to strict equilibria.
In particular, if $\eq$ is a strict \acl{NE} of $\game$ and $\paydiff = \min_{\purealt\notin\supp(\eq)} \bracks{\pay(\eq) - \payv_{\purealt}(\eq)}$ denotes the minimum payoff difference at equilibrium, it can be shown that, for all $\eps>0$, there exists a neighborhood $\nhd$ of $\eq$ in $\strats$ such that
\begin{equation}
\label{eq:strictrate}
\onenorm{\orbit{\time} - \eq}
	\leq \exp\parens*{\offset - (1-\eps)\sum_{\lolvl=0}^{\lvl} \rate\atlvl{\lolvl} \cdot \paydiff\time}
	\quad
	\text{whenever $\orbit{\tstart} \in \nhd$}
\end{equation}
where $\lvl$ denotes the degree of similarity between $\eq$ and the strategy $\purealt\notin\supp(\eq)$ that realizes the value of $\paydiff$.
Up to the factor $1-\eps$ (which is due to technical reasons, and which can be taken arbitrarily close to $1$), we observe that the degree of similarity has the same effect here as it does on the rate of extinction of dominated strategies.
\endenv
\end{remark*}

As a concrete illustration of the above, consider the following commuting game (loosely based on the blue bus\,/\,red bus paradox).
A population of commuters can choose between taking either of two bus lines (labeled as ``$1$'' and ``$2$'') or their car (strategy ``$3$'') to go to work;
going by car is faster when congestion is light, but taking the bus becomes the better choice as car traffic increases.
We represent this via the affine payoff model $\payv(\strat) = \mat\strat + \offset$ with
\begin{equation}
\label{eq:commute}
\offset
	= -\parens*{
	\begin{array}{c}
		5\\
		5\\
		2
	\end{array}}
\quad
\text{and}
\quad
\mat
	= -\parens*{
	\begin{array}{ccc}
		2 	&4	&6
		\\
		3	&0	&6
		\\
		1	&4	&8
	\end{array}}
	\,.
\end{equation}
A direct calculation (which we omit) shows that this game admits three \aclp{NE}, two strict, $(0,1,0)$ and $(0,0,1)$, and one mixed, $(0,1/2,1/2)$.
In addition, strategy $1$ (the first bus line) is dominated by strategy $3$ (commuting by car) with a margin of $\paydiff=1$, and there are no other dominated strategies.

To demonstrate the impact of similarities, we depict in \cref{fig:portraits} the solution orbits of the \acl{RD} \eqref{eq:RD} and the \acl{NRD} \eqref{eq:NRD} for this game (with the two bus lines regarded as ``similar'' by the commuters with an intra-class sampling probability of $\rate\atlvl{1} = 3/4$).
In tune with \cref{thm:NRD}, the stability characteristics of the game's \aclp{NE} do not change:
the two strict equilibria (which are evolutionarily stable) remain asymptotically stable under both \eqref{eq:RD} and \eqref{eq:NRD}, while the game's mixed equilibrium is unstable.
From a quantiative viewpoint however, the basins of attraction of the game's stable equilibria are quite different in the nested and the non-nested case:
specifically, commuting by car attracts a significantly larger set of initial conditions when the two bus lines are regarded as similar by the population.
This is also reflected in the rate of extinction of dominated strategies:
as we show in \cref{fig:domrate}, the rate of extinction of the first bus line is significantly slower when bus users are more likely to compare strategies with other bus users, precisely because neither bus line dominates the other.
This is captured by the coefficient $\sum_{\lolvl=0}^{\lvl} \rate\atlvl{\lolvl}$ of \eqref{eq:domrate}:
in the case of \eqref{eq:RD} it is equal to $1$ (because similarities are not taken into account), whereas in the case of \eqref{eq:NRD} it is equal to $\rate\atlvl{0}$ and thus, potentially much smaller.%

\subsection{Proof outline}
\label{sec:proof}

We close this section with an outline of the proof of \cref{thm:NRD,prop:domrate} (full proofs are provided in \cref{app:NRD}).
In contrast to standard arguments relying on payoff monotonicity, the enabling machinery for the rationality properties of \eqref{eq:NRD} stated in \cref{thm:NRD} requires the construction of a specific ``energy function'' generalizing the well-known \acf{KL} divergence
\begin{equation}
\label{eq:KL}
\breg(\base,\strat)
	= \sum_{\pure\in\pures} \base_{\pure} \log\frac{\base_{\pure}}{\strat_{\pure}}
	\,.
\end{equation}
With a fair degree of hindsight, we will instead consider the \emph{nested \acl{KL} divergence}
\begin{equation}
\label{eq:NKL}
\tag{NKL}
\breg_{\nLvls}(\base,\strat)
	= \sum_{\lvl=1}^{\nLvls} \diff\atlvl{\lvl}
		\sum_{\class\in\classes\atlvl{\lvl}}
		\base_{\class} \log\frac{\base_{\class}}{\strat_{\class}}.
\end{equation}
where the weights $\diff\atlvl{\lvl}$, $\lvl=1,\dotsc,\nLvls$, are given by the somewhat mysterious expression
\begin{equation}
\label{eq:rate2diff}
\diff\atlvl{\lvl}
	= \begin{cases}
		\rate\atlvl{\lvl} / (\sumrate\atlvl{\lvl-1} \sumrate\atlvl{\lvl})
			&\quad
			\text{if $\lvl=1,\dotsc,\nLvls-1$},
		\\[\smallskipamount]
		1
			&\quad
			\text{if $\lvl=\nLvls$},
	\end{cases}
\end{equation}
with $\sumrate\atlvl{\lvl} = \rate\atlvl{0} + \dotsb + \rate\atlvl{\lvl}$ denoting the total intra-class revision rate at level $\lvl$ and below (recall here our normalization assumption $\sumrate\atlvl{\nLvls-1} = \rate\atlvl{0} + \dotsb + \rate\atlvl{\nLvls-1} = 1$).

Intuitively, $\breg_{\nLvls}$ can be viewed as an affine combination of the ordinary \ac{KL} divergence between $\base$ and $\strat$ at different degrees of similarity:
in particular, the $\lvl$-th level terms in \eqref{eq:NKL} add up to the standard \ac{KL} divergence between the class state variables $(\base_{\class\atlvl{\lvl}})_{\class\atlvl{\lvl}\in\classes\atlvl{\lvl}}$ and $(\strat_{\class\atlvl{\lvl}})_{\class\atlvl{\lvl}\in\classes\atlvl{\lvl}}$ in $\simplex(\classes\atlvl{\lvl})$.
As for the \textendash\ admittedly opaque \textendash\ choice of weights $\diff\atlvl{\lvlalt}$ in the definition of $\breg_{\nLvls}$, the specific assignment \eqref{eq:rate2diff} has been judiciously chosen to provide the following fundamental link between the nested divergence \eqref{eq:NKL} and the nested dynamics \eqref{eq:NRD}:

\begin{proposition}
\label{prop:KL}
Fix a state $\base\in\strats$ and let $\orbit{\time}$ be an interior solution of \eqref{eq:NRD}.
Then:
\begin{equation}
\label{eq:dKL}
\frac{d}{dt} \breg_{\nLvls}(\base,\orbit{\time})
	= \braket{\payv(\orbit{\time})}{\orbit{\time} - \base}
	= \sum_{\pure\in\pures} \payv_{\pure}(\orbit{\time}) \, \parens{\orbit[\pure]{\time} - \base_{\pure}}
\end{equation}
\end{proposition}

In words, \cref{prop:KL} states that, under \eqref{eq:NRD}, the rate of change of the nested \ac{KL} divergence relative to a base point $\base\in\strats$ is the difference between the population average payoff $\pay(\strat) = \braket{\payv(\strat)}{\strat}$ and the average payoff $\braket{\payv(\strat)}{\base}$ of a group of mutants with actions distributed according to $\base$
(for a range of similar results in the context of game dynamics, see \citealp{MS16}, \citealp{MZ19}, and references therein).
The key property highlighted by this interpretation is that two similarity-\emph{dependent} objects \textendash\ the nested dynamics \eqref{eq:NRD} and the nested \ac{KL} divergence \eqref{eq:NKL} \textendash\ give rise to a similarity-\emph{agnostic} quantity that only depends on the game and the population states involved, and is otherwise oblivious of the similarity hierarchy $\struct$.

The main ingredient of the proof of \cref{prop:KL} is that, under the specific weight assignment \eqref{eq:rate2diff}, the time derivative of the nested divergence $\breg_{\nLvls}$ becomes a telescoping sum, with only the finest ($\lvl=\nLvls$) and coarsest ($\lvl=0$) levels surviving in the \acl{RHS} of \eqref{eq:dKL}.
This requires some delicate algebraic manipulations, which we detail below.

\begin{proof}
[Proof of \cref{prop:KL}]
By the definition of the nested \ac{KL} divergence, we have
\begin{equation}
\label{eq:dKL-1}
\frac{d}{dt} \breg_{\nLvls}(\base,\orbit{\time})
	= -\sum_{\lvl=1}^{\nLvls}
		\diff\atlvl{\lvl}\,
		\sum_{\class\in\classes\atlvl{\lvl}} \base_{\class} \frac{\dot \strat_{\class}}{\strat_{\class}}
	\,.
\end{equation}
Thus, by substituting the class dynamics \eqref{eq:NRD-class} into \eqref{eq:dKL-1} and rearranging, we get
\begin{align}
\label{eq:dKL-2}
\frac{d}{dt} \breg_{\nLvls}(\base,\orbit{\time})
	&= - \sum_{\lvl=1}^{\nLvls}
		\diff\atlvl{\lvl}
		\sum_{\class\in\classes\atlvl{\lvl}}
			\base_{\class}
			\sum_{\lolvl=0}^{\lvl-1}
				\rate\atlvl{\lolvl}\,
				\bracks{\classpay_{\class}(\strat) - \classpay_{\classof[\lolvl]{\class}}(\strat)}
	\notag\\
	&= \sum_{\lvl=1}^{\nLvls}
		\diff\atlvl{\lvl}
		\sum_{\lolvl=0}^{\lvl-1}
			\rate\atlvl{\lolvl}
			\sum_{\class\in\classes\atlvl{\lvl}}
				\base_{\class}
				\bracks{\classpay_{\class}(\strat) - \classpay_{\classof[\lolvl]{\class}}(\strat)}
	\,.
\end{align}
To proceed, note that, for all $\lolvl \leq \lvl$, we have
\begin{align}
\sum_{\class\in\classes\atlvl{\lvl}}
	\base_{\class} \classpay_{\classof[\lolvl]{\class}}(\strat)
	&= \sum_{\class\atlvl{\lolvl} \in \classes\atlvl{\lolvl}}
		\sum_{\class \desceq[\lvl] \class\atlvl{\lolvl}}
			\base_{\class} \classpay_{\classof[\lolvl]{\class}}(\strat)
	\notag\\
	&= \sum_{\class\atlvl{\lolvl} \in \classes\atlvl{\lolvl}}
		\classpay_{\class\atlvl{\lolvl}}(\strat)
		\sum_{\class \desceq[\lvl] \class\atlvl{\lolvl}} \base_{\class}
	= \sum_{\class\atlvl{\lolvl} \in \classes\atlvl{\lolvl}}
		\classpay_{\class\atlvl{\lolvl}}(\strat)\,
		\base_{\class\atlvl{\lolvl}}
	= \braket{\payv(\strat)}{\base}_{\lolvl}
\end{align}
where the notation $\braket{\payv(\strat)}{\base}_{\lolvl}$ stands for
\begin{equation}
\label{eq:classprod}
\braket{\payv(\strat)}{\base}_{\lolvl}
	\defeq \sum_{\class\atlvl{\lolvl}\in\classes\atlvl{\lolvl}} \base_{\class\atlvl{\lolvl}} \classpay_{\class\atlvl{\lolvl}}(\strat)
	\quad
	\text{for all $\base,\strat\in\strats$ and all $\lolvl=0,\dotsc,\nLvls$},
\end{equation}
and
we used the facts that
$\union\setdef{\class}{\class\desceq[\lvl]\class\atlvl{\lolvl}} = \class\atlvl{\lolvl}$
and
$\intersect\setdef{\class}{\class\desceq[\lvl]\class\atlvl{\lolvl}} = \varnothing$ whenever $\lolvl \leq \lvl$.
In this way, \eqref{eq:dKL-2} can be unpacked one attribute at a time to yield
\begin{align}
\label{eq:dKL-3}
\eqref{eq:dKL-2}
	&= \sum_{\lvl=1}^{\nLvls}
		\diff\atlvl{\lvl}
		\sum_{\lolvl=0}^{\lvl-1}
			\rate\atlvl{\lolvl}\,
			\bracks{\braket{\payv(\strat)}{\base}_{\lvl} - \braket{\payv(\strat)}{\base}_{\lolvl}}
	\notag\\
	&= \diff\atlvl{1} \rate\atlvl{0} \braket{\payv(\strat)}{\base}_{0}
		- \diff\atlvl{1} \rate\atlvl{0} \braket{\payv(\strat)}{\base}_{1}
	\notag\\
	&+ \diff\atlvl{2} \bracks{\rate\atlvl{0} \braket{\payv(\strat)}{\base}_{0} + \rate\atlvl{1} \braket{\payv(\strat)}{\base}_{1}}
		- \diff\atlvl{2} \parens{\rate\atlvl{0} + \rate\atlvl{1}}
		\, \braket{\payv(\strat)}{\base}_{2}
	\notag\\
	&\;\:\vdots
	\notag\\
	&+ \diff\atlvl{\nLvls}
		\bracks{\rate\atlvl{0} \braket{\payv(\strat)}{\base}_{0} + \dotsb + \rate\atlvl{\nLvls-1}
		\braket{\payv(\strat)}{\base}_{\nLvls-1}}
	- \diff\atlvl{\nLvls}
		\parens{\rate\atlvl{0} + \dotsb + \rate\atlvl{\nLvls-1}}
		\, \braket{\payv(\strat)}{\base}_{\nLvls}
\end{align}
and hence, regrouping all same-level payoff terms, we obtain:
\begin{align}
\label{eq:dKL-4}
\frac{d}{dt} \breg_{\nLvls}(\base,\orbit{\time})
	&= \rate\atlvl{0}
		\parens{\diff\atlvl{1} + \dotsb + \diff\atlvl{\nLvls}}
		\cdot \braket{\payv(\strat)}{\base}_{0}
	\notag\\
	&\;\:\vdots
	\notag\\
	&+ \bracks{\rate\atlvl{\lolvl} \parens{\diff\atlvl{\lolvl+1} + \dotsb + \diff\atlvl{\nLvls}}
		- \diff\atlvl{\lolvl} \parens{\rate\atlvl{0} + \dotsb + \rate\atlvl{\lolvl-1}}}
		\cdot \braket{\payv(\strat)}{\base}_{\lolvl}
	\notag\\
	&\;\:\vdots
	\notag\\
	&- \diff\atlvl{\nLvls}
		\parens{\rate\atlvl{0} + \dotsb + \rate\atlvl{\nLvls-1}}
		\cdot \braket{\payv(\strat)}{\base}_{\nLvls}
	\notag\\
	&= \sum_{\lolvl=0}^{\nLvls}
		\bracks*{\rate\atlvl{\lolvl} \sum_{\lvl=\lolvl+1}^{\nLvls} \diff\atlvl{\lvl}
			- \diff\atlvl{\lolvl} \sum_{\lvl=0}^{\lolvl-1} \rate\atlvl{\lvl}}
		\, \braket{\payv(\strat)}{\base}_{\lolvl}
\end{align}
where, in the last line, we used the empty sum convention $\sum_{j\in\varnothing} \alpha_{j} = 0$.

To proceed, we claim that
\begin{equation}
\label{eq:sumweight}
\diff\atlvl{\lolvl+1} + \dotsb + \diff\atlvl{\nLvls}
	= \frac{1}{\rate\atlvl{0} + \dotsb + \rate\atlvl{\lolvl}}
	= \frac{1}{\sumrate\atlvl{\lolvl}}
	\quad
	\text{for all $\lolvl = 0,\dotsc,\nLvls-1$}
	\,.
\end{equation}
Indeed, for $\lolvl=\nLvls-1$, we recover the definition \eqref{eq:rate2diff} of $\diff\atlvl{\nLvls}$.
Otherwise, assume inductively that \eqref{eq:sumweight} holds for some $\lolvl < \nLvls-1$;
then, by \eqref{eq:rate2diff}, we get
\begin{equation}
\label{eq:sumweight-induction}
\diff\atlvl{\lolvl} + \diff\atlvl{\lolvl+1} + \dotsb + \diff\atlvl{\nLvls}
	= \frac{\rate\atlvl{\lolvl}}{\sumrate\atlvl{\lolvl}\sumrate\atlvl{\lolvl-1}}
		+ \frac{1}{\sumrate\atlvl{\lolvl}}
	= \frac{\rate\atlvl{\lolvl} + \sumrate\atlvl{\lolvl-1}}{\sumrate\atlvl{\lolvl}\sumrate\atlvl{\lolvl-1}}
	= \frac{\sumrate\atlvl{\lolvl}}{\sumrate\atlvl{\lolvl}\sumrate\atlvl{\lolvl-1}}
	= \frac{1}{\sumrate\atlvl{\lolvl-1}},
\end{equation}
which completes the inductive step and proves our claim.
Hence, combining \eqref{eq:sumweight} with \eqref{eq:rate2diff}, we get
\begin{equation}
\label{eq:weightrate}
\rate\atlvl{\lolvl} \parens{\diff\atlvl{\lolvl+1} + \dotsb + \diff\atlvl{\nLvls}}
	= \frac{\rate\atlvl{\lolvl}}{\sumrate\atlvl{\lolvl}}
	= \diff\atlvl{\lolvl} \sumrate\atlvl{\lolvl-1}
	= \diff\atlvl{\lolvl} \parens{\rate\atlvl{0} + \dotsb + \rate\atlvl{\lolvl-1}}
\end{equation}
for all $\lvlalt=1,\dotsc,\nLvls-1$.
Thus, substituting \eqref{eq:sumweight} and \eqref{eq:weightrate} back into \eqref{eq:dKL-4}, we finally obtain
\begin{equation}
\label{eq:dKL-5}
\frac{d}{dt} \breg_{\nLvls}(\base,\orbit{\time})
	= \braket{\payv(\strat)}{\base}_{0} - \braket{\payv(\strat)}{\base}_{\nLvls}.
\end{equation}
However, since the only $0$-th level similarity class is $\pures$ itself (and $\base_{\pures} = 1$), we readily get
\begin{equation}
\braket{\payv(\strat)}{\base}_{0}
	= \base_{\pures} \; \classpay_{\pures}(\strat)
	= \sum_{\pure\in\pures} \strat_{\pure} \payv_{\pure}(\strat)
	= \braket{\payv(\strat)}{\strat}
\end{equation}
and our assertion follows.
\end{proof}

\cref{prop:KL} plays a crucial role in establishing the long-run rationality properties of \eqref{eq:NRD} recorded in \cref{thm:NRD}.
As an illustration, we sketch below the proof of the theorem's first item concerning the elimination of strictly dominated strategies.
For this, suppose that action $\pure\in\pures$ is strictly dominated by action $\purealt\in\pures$.
Then, if we let $\bvec_\pure\in \strats$ denote the ``pure'' population state corresponding to $\pure$ (and likewise for $\purealt$), the difference
\begin{equation}
\label{eq:TheBregDif}
\breg_{\nLvls}(\bvec_{\pure},\orbit{\time}) - \breg_{\nLvls}(\bvec_{\purealt},\orbit{\time})
	= \sum_{\lvl=1}^{\nLvls} \diff\atlvl{\lvl}
		\log\frac
			{\orbit[{\classof[\lvl]{\purealt}}]{\time}}
			{\orbit[{\classof[\lvl]{\pure}}]{\time}},
\end{equation}
is a weighted sum of log ratios of the masses of agents using actions in classes containing $\pure$ and $\purealt$.
It then follows from \cref{prop:KL} that
\begin{align}
\frac{d}{dt} \left(\vphantom{I^I}\breg_{\nLvls}(\bvec_{\pure},\orbit{\time}) - \breg_{\nLvls}(\bvec_{\purealt},\orbit{\time})\right)
	&= \braket{\payv(\orbit{\time})}{\orbit{\time} - \bvec_{\pure}} - \braket{\payv(\orbit{\time})}{\orbit{\time} - \bvec_{\purealt}}\notag\\
	&= \payv_{\purealt}(\orbit{\time}) - \payv_{\pure}(\orbit{\time}).
	\label{eq:ThePayoffDiff}
\end{align}
Since $\purealt$ strictly dominates $\pure$, the payoff difference \eqref{eq:ThePayoffDiff} is positive and bounded away from zero.
Because of this, the difference \eqref{eq:TheBregDif} must grow without bound, which in turn implies that the population share $\strat_\pure$ of the dominated strategy $\pure$ vanishes, and a more detailed algebraic manipulation allows us to get the rate estimate described in \cref{prop:domrate}.

For a detailed proof of the above claims and the remaining parts of \cref{thm:NRD,prop:domrate}, we refer the reader to \cref{app:NRD}.

\section{Evolution through reinforcement and \acl{NLC}}
\label{sec:learning}

In this section, we discuss a continuous-time model for learning in games, which reveals a deep relation between the \acl{NRD} and the \acl{NLC} rule of \cite{BA73} and \cite{McF78}.
To establish a baseline, we begin by describing the non-nested case in \cref{sec:LC} below;
subsequently, we present the details of this relation in full in \cref{sec:NLC,sec:NEW} right after.

\subsection{Logit choice, \acl{EW}, and the replicator dynamics}
\label{sec:LC}

Our starting point is the following stimulus-response model of evolution in the spirit of \citet{ER98}:
\begin{itemize}
\item
At each instance $\time\geq0$, the agents observe the cumulative payoff $\score_{\pure}(\time)$ of each pure strategy $\pure\in\pures$ up to time $\time$, \viz
\begin{equation}
\label{eq:score}
\score_{\pure}(\time)
	= \int_{0}^{\time} \payv_{\pure}(\strat(\timealt)) \dd\timealt
\end{equation}
\item
Subsequently, players select actions with probabilities that are exponentially proportional to the aggregate propensity scores $\score_{\pure}$, reinforcing in this way those actions that seem to be performing better.
As a result, in our nonatomic setting, the relative share of $\pure$-strategists at any given time are determined by the exponential choice rule
\begin{equation}
\label{eq:score2strat}
\strat_{\pure}(\time)
	= \frac{\exp(\score_{\pure}(\time)/\temp)}{\sum_{\purealt\in\pures} \exp(\score_{\purealt}(\time)/\temp)}
\end{equation}
where $\temp>0$ is a ``temperature parameter'' that indicates the players' sensitivity to changes in the observed payoffs.
When $\temp\to\infty$, players are relatively indifferent and tend to employ all strategies in a near-uniform mix;
by contrast,
when $\temp\to0$, players are extremely sensitive to payoff changes and essentially ``best-respond'' to the vector of aggregate scores $\score(\time)$.
\end{itemize}

The main difference between this model and the revisionist viewpoint of the previous sections is that, instead of specifying the inflows and outflows to and from each strategy, \cref{eq:score2strat} directly prescribes the distribution of the population as a function of the cumulative propensity scores \eqref{eq:score}.
Nevertheless, despite these very different starting points, these two processes yield the same evolutionary dynamics:
indeed, by expressing \cref{eq:score,eq:score2strat} in differential form, we obtain the \acdef{EW} dynamics
\begin{equation}
\label{eq:EW}
\tag{EW}
\dot\score
	= \payv(\strat)
	\qquad
\strat
	= \logit(\score)
\end{equation}
where
\begin{equation}
\label{eq:LC}
\tag{LC}
\logit_{\pure}(\score)
	=\frac{\exp(\score_{\pure}/\temp)}{\sum_{\purealt \in \pures} \exp(\score_{\purealt}/\temp)}
\end{equation}
denotes the logit choice map underlying \eqref{eq:score2strat}.
Then, by a straightforward differentiation, we readily get the system
\begin{equation}
\label{eq:EW2RD}
\dot\strat_{\pure}
	= \temp^{-1} \, \strat_{\pure} \bracks{\payv_{\pure}(\strat) - \pay(\strat)}
\end{equation}
which, up to the multiplicative factor $\temp^{-1}$, is simply the replicator equation \eqref{eq:RD} induced by the revision protocol \eqref{eq:PPI}.

\begin{remark*}
The \acl{EW} scheme \eqref{eq:EW} is one of the staples of the online learning literature and its origins can be traced to the work of \citet{Vov90}, \citet{LW94} and \citet{ACBFS95,ACBF02} on multi-armed bandits.
The specific derivation of the replicator dynamics that we presented here goes back at least as far as \citet{Rus99} but versions of this model have also been studied in different contexts by \citet{BS97}, \citet{Pos97}, \citet{Hop02}, \citet{Beg05}, \citet{LC05}, \citet{HSV09}, \citet{MM10}, \citet{CGM15}, and many others;
for a partial survey, see \citet{MHC24} and references therein.
\end{remark*}

A key byproduct of this equivalence between \eqref{eq:RD} and \eqref{eq:EW} is that the latter immediately inherits all the rationality properties of the former, including \textendash\ but not limited to \textendash\ \cref{thm:NRD}. 
This suite of properties has contributed significantly to the popularity of \eqref{eq:EW} in the literature on learning in games, so a natural question that arises is whether a similar link can be established between the nested dynamics \eqref{eq:NRD} and a suitable model for learning in the presence of similarities.
We address this question in the following sections.

\subsection{\Acl{NLC}}
\label{sec:NLC}

Going back to the derivation of the dynamics \eqref{eq:EW}, a key observation is that the action selection rule \eqref{eq:LC} puts all alternatives on an equal footing and is oblivious to any exogenous similarities they might possess.
Albeit reasonable in certain contexts, this feature of \eqref{eq:LC} can become problematic when the inclusion of irrelevant alternatives ends up conflating the action choices of the population and lead to undesirable behaviors.

For an illustration, consider the following adaptation of the classical ``red bus\,/\,blue bus'' paradox of \citet{McF74a}.
To wit, suppose that a population of commuters can take a car or bus to work and, on average,
commuting by car takes half as long as commuting by bus.
Accordingly, if the induced propensities \eqref{eq:score} for employing each alternative are, say, $\score_{\car} = -1/2$ and $\score_{\bus} = -1$, the ratio of commuters taking a car over a bus under \eqref{eq:LC} would be $\strat_{\car} / \strat_{\bus} = \exp(-1/2) / \exp(-1) \approx 1.6$ in favor of taking a car, in tune with the fact that commuting by bus takes much longer.
The paradox now comes up if the company operating the bus network paints half of the buses blue and the other half red.
This change has absolutely no effect on the travel time of a bus, but since the commuters' new set of alternatives is $\pures = \{\car,\redbus,\bluebus\}$, the ratio of commuters taking a car under \eqref{eq:LC} would drop to $\strat_{\car} / [\strat_{\redbus} + \strat_{\bluebus}] = \exp(-1/2) / [\exp(-1) + \exp(-1)] \approx 0.8$, \ie half of what it was before the addition of an irrelevant feature (the color of the bus), and no longer consistent with the associated travel times.

To avoid this artificial inflation of similar alternatives, \citet{BA73} and \citet{McF74,McF78} introduced the so-called \ac{NLC} rule, which does not treat all alternatives together at the same time;
instead, it singles out an alternative by progressively filtering out similarities within the set of all possible alternatives, one attribute at a time.
In more detail, given a similarity structure $\struct$ and a profile of propensity scores
$\score = (\score_{1},\dotsc,\score_{\nPures}) \in \R^{\pures}$ on $\pures$,
the \acl{NLC} model unfolds as follows:
\begin{enumerate}
\item
First, the propensity score of a similarity class $\class\atlvl{\lolvl}\in\classes\atlvl{\lolvl}$, $\lolvl=0,\dotsc,\nLvls-1$, is defined inductively as
\begin{equation}
\label{eq:score-class}
\score_{\class\atlvl{\lvl-1}}
	= \temp\atlvl{\lvl}
		\log \!\!\sum_{\class\atlvl{\lvl} \childof \class\atlvl{\lvl-1}}\!\!
			\exp\parens*{\frac{\score_{\class\atlvl{\lvl}}}{\temp\atlvl{\lvl}}}
\end{equation}
where
$\temp\atlvl{\lvl} > 0$, $\lvl=1,\dotsc,\nLvls$, measures the players' \emph{uncertainty level} relative to the $\lvl$-th tier $\classes\atlvl{\lvl}$ of $\struct$.
Thus, starting with the individual alternatives of $\pures$ \textendash\ that is, the \emph{leaves} $\{\pure\} \in \classes\atlvl{\nLvls}$ of $\struct$ \textendash\ propensity scores are backpropagated along $\struct$, one attribute at a time, until reaching the root $\classes\atlvl{0} = \{\pures\}$ of $\struct$.
The reason for defining the class scores $\score_{\class\atlvl{\lolvl}}$ in this precise way is discussed in detail below;
for now, we only note that \eqref{eq:score-class} defines the score of a class as the weighted softmax of the scores of its children, meaning in particular that a class with higher-scoring options will score higher overall.

\item
An alternative $\pure\in\pures$ is drawn by specifying a chain of progressively finer similarity classes $\source \equiv \class\atlvl{0} \parentof \class\atlvl{1} \parentof \dotsb \parentof \class\atlvl{\nLvls} \equiv \{\pure\}$, each selected via the family of conditional probabilities
\begin{equation}
\label{eq:NLC-cond}
\choice[\class\atlvl{\lvl} \vert \class\atlvl{\lvl-1}](\score)
	= \frac
		{\exp(\score_{\class\atlvl{\lvl}} / \temp\atlvl{\lvl})}
		{\exp(\score_{\class\atlvl{\lvl-1}} / \temp\atlvl{\lvl})}
	= \frac
		{\exp(\score_{\class\atlvl{\lvl}} / \temp\atlvl{\lvl})}
		{\sum_{\alt\class\atlvl{\lvl} \sibling \class\atlvl{\lvl}}\exp(\score_{\alt\class\atlvl{\lvl}} / \temp\atlvl{\lvl})}
\end{equation}
for all $\lvl=1,\dotsc,\nLvls$.
In particular, \eqref{eq:NLC-cond} first prescribes choice probabilities $\strat_{\class\atlvl{1}}$ for all classes $\class\atlvl{1}\in\classes\atlvl{1}$ (\ie the coarsest ones);
subsequently, once a class $\class\atlvl{1} \in \classes\atlvl{1}$ has been selected, \eqref{eq:NLC-cond} prescribes the conditional choice probabilities $\strat_{\class\atlvl{2} \given \class\atlvl{1}}$ for all children $\class\atlvl{2}$ of $\class\atlvl{1}$ and draws a class from $\classes\atlvl{2}$ based on $\strat_{\class\atlvl{2} \given \class\atlvl{1}}$.
The process then continues downward along $\struct$ until reaching the finest partition $\classes\atlvl{\nLvls}$ and selecting an alternative $\{\pure\} \equiv \class\atlvl{\nLvls} \in \classes\atlvl{\nLvls}$.

\end{enumerate}

In this way, unrolling the conditional probabilities \eqref{eq:NLC-cond} over the lineage of any given action $\pure\in\pures$, we obtain the \acli{NLC} model
\begin{equation}
\label{eq:NLC-elem}
\tag{NLC}
\choice[\pure](\score)
	= \prod_{\lvl=1}^{\nLvls}\;
	\frac{\exp\parens[\big]{\score_{\class\atlvl{\lvl}}/\temp\atlvl{\lvl}}}
		{\exp\parens[\big]{\score_{\class\atlvl{\lvl-1}}/\temp\atlvl{\lvl}}}
	\quad
	\text{for all $\pure\in\pures$}
\end{equation}
with each $\score_{\class\atlvl{\lvl}}$, $\lvl=1,\dotsc,\nLvls$, given by \eqref{eq:score-class}.
This action selection rule \textendash\ originally due to \citet{BA73} and \citet{McF78} \textendash\ will be the mainstay of our analysis, so some remarks are in order.

The first thing of note is the precise expression \eqref{eq:score-class} of the class propensities $\score_{\class}$.
The reason for this definition is the conditional probability rule \eqref{eq:NLC-cond}, which posits that each attribute selection step of \eqref{eq:NLC-elem} is an instance of the non-nested choice rule \eqref{eq:LC}:
in particular, treating the parent class $\class\atlvl{\lvl-1}$ as a set of alternatives on its own right, the expression \eqref{eq:score-class} is essentially imposed by the requirement $\sum_{\class\atlvl{\lvl} \childof \class\atlvl{\lvl-1}} \choice[\class\atlvl{\lvl} \given \class\atlvl{\lvl-1}](\score) = 1$.
As for the uncertainty parameters $\temp\atlvl{1},\dotsc,\temp\atlvl{\nLvls} > 0$, these are implicitly assumed to model the level of the agents' uncertainty when selecting a particular attribute:
for large values of $\temp\atlvl{\lvl}$, agents tend to be uniformly distributed among the constituent classes of $\classes\atlvl{\lvl}$;
by contrast, when $\temp\atlvl{\lvl}\to0$, nearly the entire population is employing the similarity class $\class\atlvl{\lvl}\in\classes\atlvl{\lvl}$ with the highest aggregate propensity $\score_{\class\atlvl{\lvl}}$.
Since coarser attributes inherently come with higher levels of uncertainty, we will assume throughout that $\temp\atlvl{1} \geq \dotsm \geq \temp\atlvl{\nLvls} > 0$;
this point will play an important role in the sequel.

Now, other than selecting a primitive alternative $\pure\in\pures$, unrolling the chain of probabilities \eqref{eq:NLC-cond} over the lineage $\source \equiv \class\atlvl{0} \parentof \class\atlvl{1} \parentof \dotsb \parentof \class\atlvl{\lvl} \equiv \class$ of a target class $\class\in\classes\atlvl{\lvl}$ yields
\begin{equation}
\label{eq:NLC-class}
\choice[\class](\score)
	= \frac{\exp(\score_{\class}/\temp\atlvl{\lvl})}{\exp(\score_{\class\atlvl{\lvl-1}}/\temp\atlvl{\lvl})}
	\times \dotsm
	\times \frac{\exp(\score_{\class\atlvl{1}}/\temp\atlvl{1})}{\exp(\score_{\source}/\temp\atlvl{1})}
	= \prod\nolimits_{\lvlalt=1}^{\lvl}
		\frac
			{\exp(\score_{\class\atlvl{\lvlalt}}/\temp\atlvl{\lvlalt})}
			{\exp(\score_{\class\atlvl{\lvlalt-1}}/\temp\atlvl{\lvlalt})}
\end{equation}
for the relative share of individuals playing an element of $\class$ under the propensity score profile $\score\in\R^{\pures}$.
Clearly, the element-wise expression \eqref{eq:NLC-elem} is a special case of the more general model \eqref{eq:NLC-class}, and both are a direct consequence of \eqref{eq:NLC-cond}.
In fact, as we show below, the two definitions are equivalent:

\begin{lemma}
\label{lem:NLC-class}
Fix a propensity score profile $\score\in\R^{\nElems}$.
Then, under \eqref{eq:NLC-elem}, we have:
\begin{enumerate}
[label={\arabic*.},ref=\arabic*]
\item
For all $\class\in\struct$, the class selection probabilities $\choice[\class](\score) = \sum_{\pure\in\class} \choice[\pure](\score)$ satisfy \eqref{eq:NLC-class}.
\item
For every child/parent pair $\child \childof \parent\in\struct$, the conditional class selection probabilities $\choice[\child\given\parent](\score) = \choice[\child](\score) / \choice[\parent](\score)$ satisfy \eqref{eq:NLC-cond}.
\end{enumerate}
\end{lemma}

\begin{proof}
Our proof is by reverse induction on the level $\lvl = \attof{\class}$ of $\class$.
For our first claim, the case $\attof{\class} = \nLvls$ is simply the definition of $\choice[\class](\score)$ so there is nothing to show.
Then, assuming inductively that our assertion holds for some $\lvl < \nLvls$, a simple summation gives
\begin{align}
\choice[\class\atlvl{\lvl-1}](\score)
	&= \smashoperator[r]{\sum_{\class\atlvl{\lvl} \childof \class\atlvl{\lvl-1}}}
	\;\;
	\choice[\class\atlvl{\lvl}](\score)
	= \smashoperator[r]{\sum_{\class\atlvl{\lvl} \childof \class\atlvl{\lvl-1}}}
	\;\;\;
	\prod_{\lvlalt=1}^{\lvl}
		\frac
			{\exp(\score_{\class\atlvl{\lvlalt}}/\temp\atlvl{\lvlalt})}
			{\exp(\score_{\class\atlvl{\lvlalt-1}}/\temp\atlvl{\lvlalt})}
	\notag\\
	&= \prod_{\lvlalt=1}^{\lvl-1}
		\frac
			{\exp(\score_{\class\atlvl{\lvlalt}}/\temp\atlvl{\lvlalt})}
			{\exp(\score_{\class\atlvl{\lvlalt-1}}/\temp\atlvl{\lvlalt})}
	\times\;\;
		\smashoperator{\sum_{\class\atlvl{\lvl} \childof \class\atlvl{\lvl-1}}}
		\;\;\;
		\frac{\exp\parens[\big]{\score_{\class\atlvl{\lvl}}/\temp\atlvl{\lvl}}}
			{\exp\parens[\big]{\score_{\class\atlvl{\lvl-1}}/\temp\atlvl{\lvl}}}
		= \prod_{\lvlalt=1}^{\lvl-1}
		\frac
			{\exp(\score_{\class\atlvl{\lvlalt}}/\temp\atlvl{\lvlalt})}
			{\exp(\score_{\class\atlvl{\lvlalt-1}}/\temp\atlvl{\lvlalt})}
\end{align}
where the last equality is an immediate consequence of the definition \eqref{eq:score-class} of the class scores $\score_{\class\atlvl{\lvl}}$.
This completes the induction and proves our first claim;
for our second claim, simply take the quotient $\choice[\child](\score) / \choice[\parent](\score)$ and apply the first.
\end{proof}


By \cref{lem:NLC-class}, \eqref{eq:NLC-elem} and \cref{eq:NLC-cond,eq:NLC-class} are all equivalent in terms of action distributions, so we will refer to either one as the definition of \acl{NLC}.

\subsection{\Acl{NEW} and the \acl{NRD}}
\label{sec:NEW}

Going back to our learning considerations, combining \eqref{eq:NLC-elem} with the propensity scoring mechanism \eqref{eq:score} yields the \acli{NEW}\acused{NEW} scheme
\begin{equation}
\label{eq:NEW}
\tag{NEW}
\dot\score
	= \payv(\strat)
	\qquad
\strat
	= \choice(\score)
\end{equation}
with $\choice\from\R^{\pures}\to\strats$ defined viasatisfies \eqref{eq:NLC-elem}.
In words, agents following \eqref{eq:NEW} continue to tally the payoffs of their individual actions as they accrue over time;
however, instead of disregarding intrinsic similarities between actions, they now
\begin{enumerate*}
[\itshape a\upshape)]
\item
iteratively assign a propensity score via \eqref{eq:score-class} to each class of strategies that share one or more similarity attributes;
and
\item
they employ the \acl{NLC} rule \eqref{eq:NLC-cond} to select an action based on this iterative assignment of propensities.
\end{enumerate*}

In the rest of this section, we provide a broad generalization of the link between \eqref{eq:RD} and \eqref{eq:EW} that we described in \cref{sec:LC}.
Specifically, we show that, despite their completely different origins, the dynamics \eqref{eq:NRD} and \eqref{eq:NEW} are equivalent:

\begin{theorem}
\label{thm:NEW}
If $\strat(\time) = \choice(\score(\time))$ is a solution orbit of \eqref{eq:NEW}, it satisfies \eqref{eq:NRD} with revision rates
\begin{equation}
\label{eq:temp2rate}
\rate_{\lolvl}
	=
	\begin{cases}
	1/\temp\atlvl{1}
		&\quad
		\text{if $\lolvl=0$},
		\\[\smallskipamount]
	1/\temp\atlvl{\lolvl+1} - 1/\temp\atlvl{\lolvl}
		&\quad
		\text{if $\lolvl=1,\dotsc,\nLvls-1$}.
	\end{cases}
\end{equation}
Conversely, if $\strat(\time)$ is an interior solution orbit of \eqref{eq:NRD}, it satisfies \eqref{eq:NEW} with uncertainty parameters
\begin{equation}
\label{eq:rate2temp}
\temp\atlvl{\lvl}
	= \frac{1}{\sumrate\atlvl{\lvl-1}}
	\equiv \frac{1}{\sum_{\lolvl=0}^{\lvl-1} \rate\atlvl{\lolvl}}
	\quad
	\text{for all $\lvl=1,\dotsc,\nLvls$}.
\end{equation}
\end{theorem}

Importantly, as an immediate consequence of \cref{thm:NEW}, we have:

\begin{corollary}
\label{cor:NEW}
Let $\game\equiv\gamefull$ be a population game, and let $\struct$ be a similarity structure on $\pures$.
Then the conclusions of \cref{thm:NRD} apply as stated to the nested dynamics \eqref{eq:NEW}.
\end{corollary}

From a modeling viewpoint,
it is worth noting that \cref{thm:NEW} reveals a somewhat unexpected relation between the revision rates $\rate\atlvl{\lolvl}$ of the nested protocol \eqref{eq:NPPI} and the uncertainty parameters $\temp\atlvl{\lolvl}$ of the nested choice model \eqref{eq:NLC-elem}.
Mutatis mutandis, the uncertainty inherent in a given attribute partition $\classes\atlvl{\lvl}$ of $\pures$ is \emph{inversely proportional} to the \emph{total revision rate} of agents looking to switch outside the classes defined by $\classes\atlvl{\lvl}$.
Since these switches are more likely to inform the agents making them about the payoffs they could encounter by disregarding the similarities of level $\lvl$ or lower, it is not unreasonable to expect that, in a discrete choice context, any corresponding ``uncertainty'' should be likewise lower when these revision rates are higher.
The conversion recipe \eqref{eq:rate2temp} can be seen as a formal version of this high-level, heuristic analogy.

The key element in the proof of \cref{thm:NEW} is \cref{prop:NEW-class} below, which establishes the following property of independent interest:
under \eqref{eq:NEW}, the rate of change of the (recursively defined) propensity score $\score_{\class}$ of any similarity class $\class\in\struct$ is given by the mean class payoff $\classpay_{\class}(\strat)$ of the class in question.%
\footnote{The non-nested analogue of this observation is that, under \eqref{eq:RD}/\eqref{eq:EW}, the changes in each $\score_{\pure}$, $\pure\in\pures$, under \eqref{eq:EW} are given by the corresponding payoffs $\payv_{\pure}(\strat)$ at the current mixed strategy $\strat = \logit(\score)$.}
More formally, we have:

\begin{proposition}
\label{prop:NEW-class}
Under \eqref{eq:NEW}, we have
\begin{equation}
\label{eq:NEW-class}
\dot\score_{\class}
	= \classpay_{\class}(\strat)
	\quad
	\text{for all $\class\in\struct$}.
\end{equation}
\end{proposition}

\begin{proof}[Proof of \cref{prop:NEW-class}]
Our proof is by induction on the level $\lvl$ of the class $\class\atlvl{\lvl}$.
To begin with, for $\lvl=\nLvls$, \eqref{eq:NEW-class} is simply the definition of \eqref{eq:NEW}, so there is nothing to show.
Moving forward, to establish the inductive step, assume that $\dot\score_{\class\atlvl{\lvl}} = \classpay_{\class\atlvl{\lvl}}$ for some $\lvl < \nLvls$ and for all $\class\atlvl{\lvl}\in\classes\atlvl{\lvl}$.
Then, for all $\class\atlvl{\lvl-1} \in \classes\atlvl{\lvl-1}$, we have
\begin{align}
\dot\score_{\class\atlvl{\lvl-1}}
	= \frac{d}{dt} \temp\atlvl{\lvl}
		\log \sum_{\class\atlvl{\lvl}\childof\class\atlvl{\lvl-1}}
			\exp\parens*{\frac{\score_{\class\atlvl{\lvl}}}{\temp\atlvl{\lvl}}}
	&= \temp\atlvl{\lvl} \frac
		{\sum_{\class\atlvl{\lvl}\childof\class\atlvl{\lvl-1}} \dot\score_{\class\atlvl{\lvl}}/\temp\atlvl{\lvl} \cdot \exp\parens{\score_{\class\atlvl{\lvl}}/\temp\atlvl{\lvl}}}
		{\sum_{\class\atlvl{\lvl}'\childof\class\atlvl{\lvl-1}} \exp\parens{\score_{\class\atlvl{\lvl}'}/\temp\atlvl{\lvl}}}
	\explain{by \eqref{eq:score-class}}
	\\
	&=\;\;\;
	\sum_{\class\atlvl{\lvl}\childof\class\atlvl{\lvl-1}}
	\dot\score_{\class\atlvl{\lvl}}
		\frac{\strat_{\class\atlvl{\lvl}}}{\strat_{\class\atlvl{\lvl-1}}}
	\explain{by \eqref{eq:NLC-cond}}
	\\
	&=\;\;\;
	\smashoperator{\sum_{\class\atlvl{\lvl}\childof\class\atlvl{\lvl-1}}}
	\;\;\;
	\strat_{\class\atlvl{\lvl} \vert \class\atlvl{\lvl-1}}
		\classpay_{\class\atlvl{\lvl}}(\strat)
	\explain{by induction}\\
	&=\;\;\;
	\smashoperator{\sum_{\class\atlvl{\lvl}\childof\class\atlvl{\lvl-1}}}
	\;\;\;\;
	\sum_{\pure\in\class\atlvl{\lvl}}\strat_{\pure \vert \class\atlvl{\lvl-1}}
		\payv_\pure(\strat)
	= \classpay_{\class\atlvl{\lvl-1}}(\strat),
\end{align}
where the last line following from the definition \eqref{eq:pay-class} of $\payv_{\class}$.
This completes the induction and our proof.
\end{proof}

With this structural proposition in hand, the proof of \cref{thm:NEW} proceeds as follows:

\begin{proof}[Proof of \cref{thm:NEW}]
Fix an alternative $\pure\in\pures$ with lineage $\{\pure\} \equiv \class\atlvl{\nLvls} \childof \dotsb \childof \class\atlvl{1} \childof \class\atlvl{0} \equiv \pures$.
Then, by the definition of the \acl{NLC} rule, we have
\begin{align}
\strat_{\pure}
	&= \prod_{\lvl=1}^{\nLvls}\;
	\frac{\exp\parens[\big]{\score_{\class\atlvl{\lvl}}/\temp\atlvl{\lvl}}}
		{\exp\parens[\big]{\score_{\class\atlvl{\lvl-1}}/\temp\atlvl{\lvl}}}
	\explain{by \eqref{eq:NLC-elem}}
	\\
	&= \exp\parens*{-\frac{\score_{\pures}}{\temp\atlvl{1}}}
	\times \prod_{\lolvl=1}^{\nLvls-1}
		\exp\parens*{\parens*{\frac{1}{\temp\atlvl\lolvl} - \frac{1}{\temp\atlvl{\lolvl+1}}} \, \score_{\class\atlvl{\lolvl}}}
	\times \exp\parens*{\frac{\score_{\pure}}{\temp\atlvl{\nLvls}}}
	\notag
	\\
	&\label{eq:NLC-regrouped}
	= \exp\parens*{\score_{\pure}/\temp_{0}}
	\times \prod_{\lolvl=0}^{\nLvls-1}
		\exp\parens*{-\rate\atlvl{\lolvl} \, \score_{\class\atlvl{\lolvl}}},
\end{align}
with $\rate\atlvl{\lolvl}$ defined as in \eqref{eq:temp2rate}.
Now, if we use \eqref{eq:NLC-regrouped} to view $\strat_{\pure}$ as a function of the associated class scores $\score_{\class\atlvl{\lvl}}$ (viewed momentarily as independent variables), a straightforward differentiation yields
\begin{equation}
\frac{\pd \strat_{\pure}}{\pd \score_{\class\atlvl{\lolvl}}}
	=
	\begin{cases}
		- \rate\atlvl{\lolvl} \strat_{\pure}
			&\quad
			\text{if $\lolvl\in\{0,\dotsc,\nLvls-1\}$},
			\\
		\strat_{\pure}/\temp\atlvl{\nLvls}
			&\quad
			\text{if $\lvl=\nLvls$}.
	\end{cases}
\end{equation}
Thus, by \cref{prop:NEW-class}, we obtain
\begin{align}
\dot\strat_{\pure}
	= \sum_{\lvl=0}^{\nLvls}
		\frac{\pd \strat_{\pure}}{\pd \score_{\class\atlvl{\lolvl}}}
		\dot\score_{\class\atlvl{\lolvl}}
	&= \strat_{\pure}
		\bracks*{
			\frac{\payv_{\pure}(\strat)}{\temp\atlvl{\nLvls}}
			- \sum_{\lolvl=0}^{\nLvls-1} \rate\atlvl{\lolvl} \classpay_{\class\atlvl{\lolvl}}(\strat)}
	\notag\\
	&= \strat_{\pure}
		\bracks*{
			\sum_{\lolvl=0}^{\nLvls-1} \rate\atlvl{\lolvl} \cdot \payv_{\pure}(\strat)
			- \sum_{\lolvl=0}^{\nLvls-1} \rate\atlvl{\lolvl} \classpay_{\class\atlvl{\lolvl}}(\strat)}
	\explain{by \eqref{eq:rate2temp}}
	\\
	&= \strat_{\pure} \sum_{\lolvl=0}^{\nLvls-1}
		\rate\atlvl{\lolvl}\bracks*{\payv_{\pure}(\strat) - \classpay_{\class\atlvl{\lolvl}}(\strat)}.
\end{align}
This shows that $\strat_{\pure}(\time)$ follows the nested dynamics \eqref{eq:NRD}, as was to be shown.
\end{proof}

\section{Links with regularized learning}
\label{sec:FTRL}

We conclude our paper with a discussion of an important link between \eqref{eq:NRD} and a more general class of game dynamics driven by reinforcement and regularization.
To provide the necessarily context, \citet{McF78} presents \eqref{eq:NLC-elem} as the choice probabilities of an \ac{ARUM} whose payoff disturbances follow a generalized-extreme-value joint distribution, and shows that choice probabilities under \eqref{eq:NLC-elem} can be obtained as the logarithmic derivatives of the (recursively-defined) generator of this distribution.%
\footnote{For its part, the logit rule \eqref{eq:LC} 
requires extreme-value distributed payoff disturbances that are independent across strategies;
for more details, see \cite{AdPT92}.}
Expanding on previous work by \citet{Ver96} and \citet{MMRZ22}, we provide here a representation of the choice probabilities $\choice[\pure](\score)$ of \eqref{eq:NLC-elem} and the smoothed maximal scores $\score_{\class\atlvl{\lvl}}$ of \eqref{eq:score-class} as solutions to a nonrecursive optimization problem, and we connect them to a model of regularized learning in games known as \acdef{FTRL}.

To set the stage for all this, recall that the logit choice rule \eqref{eq:LC} can be rewritten as
\begin{align}
\label{eq:LC-sol}
\logit(\score)
	= \nabla\score_{\pures}
	= \argmax_{\strat\in\strats} \braces*{\braket{\score}{\strat} - \hreg(\strat)}
\shortintertext{where}
\label{eq:LC-val}
\score_{\pures}
	= \log \sum_{\pure\in\pures} \exp\parens*{\score_{\pure}}
	= \max_{\strat\in\strats} \braces*{\braket{\score}{\strat} - \hreg(\strat)}.
\end{align}
In the above,
\begin{equation}
\label{eq:entropy}
\hreg(\strat)
	= \temp\insum_{\pure\in\pures} \strat_{\pure} \log \strat_{\pure}
\end{equation}
denotes the negative entropy of the population state $\strat\in\strats$, and we are using the standard continuity convention $0\log0=0$.
Thus, given a propensity score vector $\score\in\R^{\pures}$, \eqref{eq:LC-sol} states that logit choice maximizes the average score $\braket{\score}{\strat}$ minus an entropic penalty term;
the corresponding maximum value is then provided by the softmax function \eqref{eq:LC-val}.

From an evolutionary perspective, this implies that the dynamics \eqref{eq:RD}/\eqref{eq:EW} can be seen as a special case of the \acli{RL} dynamics
\begin{equation}
\label{eq:RL}
\tag{RL}
\dot\score
	= \payv(\strat)
	\qquad
\strat
	= \argmax\nolimits_{\stratalt\in\strats} \{\braket{\score}{\stratalt} - \hreg(\stratalt) \}
\end{equation}
where $\hreg\from\strats\to\R$ is a generic penalty term, referred to as the method's ``regularizer'' \citep{HS09,SS11,MS16}.
In words, \eqref{eq:RL} admits a straightforward interpretation:
agents simply adjust their strategies over time by best-responding to their actions' cumulative payoffs minus a regularization penalty intended to incentivize exploration.
This regularization idea is essentially the one underlying the stochastic \textendash\ or ``smooth'', depending on the context \textendash\ fictitious play process of \citet{FK93}, which has been studied in a wide variety of contexts by \citet{FL95,FL99}, \citet{HS09}, \citet{BM17}, \citet{GVM21}, and many others;
for a partial survey, see \citet{HS09} and references therein.%
\footnote{In the online learning literature, \eqref{eq:RL} is often referred to as \acdef{FTRL}, \cf \citet{SSS06} and \citet{SS11}.}

In the remaining of this section, we extend this interpretation to the nested dynamics \eqref{eq:NRD}/\eqref{eq:NEW}.
This extension is quite intricate and hinges on a delicately tuned, nested variant of the entropic regularizer \eqref{eq:entropy} defined as
\begin{equation}
\label{eq:NER}
\hreg(\strat)
	= \sum_{\lolvl=0}^{\nLvls}
		\diff\atlvl{\lolvl} \sum_{\class\atlvl{\lolvl} \in \classes\atlvl{\lolvl}}
		\strat_{\class\atlvl{\lolvl}} \log\strat_{\class\atlvl{\lolvl}}
\end{equation}
where the weight coefficients $\diff_{\lolvl}$, $\lolvl=0,\dotsc,\nLvls$, are given by
\begin{equation}
\label{eq:temp2diff}
\diff\atlvl{\lolvl}
	= \begin{cases}
		\temp\atlvl{\lolvl} - \temp\atlvl{\lolvl+1}
			&\quad
			\text{for $\lolvl=0,\dotsc,\nLvls-1$},
		\\
		\temp\atlvl{\nLvls}
			&\quad
			\text{for $\lolvl=\nLvls$}.
	\end{cases}
\end{equation}
A version of this ``nested entropy'' was first derived by \citet{MerSan18} and \citet{MMRZ22} in the context of Riemannian game dynamics and multi-armed bandits respectively.
In the present setting, an interpretation of this definition is as follows:
at the finest level ($\lolvl = \nLvls$), the term $\sum_{\pure\in\pures} \strat_{\pure} \log\strat_{\pure}$ measures the (negative) entropy of the population state $\strat$, with no similarities taken into account.
Then, at each progressively coarser level $\lolvl = \nLvls-1,\dotsc,0$, the various summands of the form $\sum_{\class\atlvl{\lolvl} \in \classes\atlvl{\lolvl}} \strat_{\class\atlvl{\lolvl}} \log\strat_{\class\atlvl{\lolvl}}$ measure the entropy of the probability distribution induced by $\strat$ on $\simplex(\classes\atlvl{\lolvl})$, that is, restricted to $\lolvl$-tier similarity classes.
These individual entropy measures do not enter the definition of $\hreg_{\class}$ with uniform weight, but are suitably adjusted by the uncertainty level $\temp\atlvl{\lvl}$ asssociated to each similarity partition $\classes\atlvl{\lvl}$ of $\struct$;
the specific choice of the weight coefficients involves a significant degree of hindsight, but the fact that it completes the triangle with \eqref{eq:rate2diff} and \eqref{eq:temp2rate} is, of course, not a coincidence.


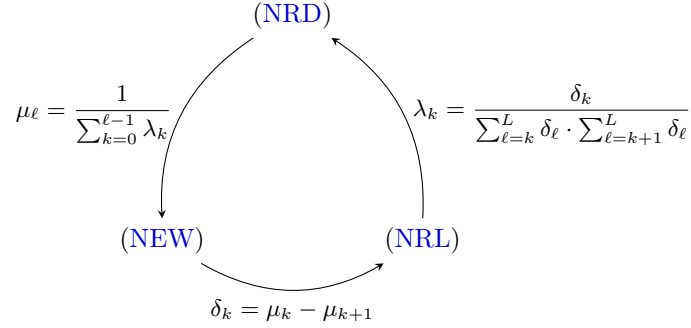
\begin{figure}[tbp]

\begin{tikzpicture}
[scale=1,
class/.style={inner sep=2pt},
desc/.style={rounded corners=2pt,fill=white,inner sep=2pt,align=left},
edgestyle/.style={->},
label/.style={circle,fill=white,inner sep=0pt},
nest/.style={densely dashed,rounded corners=2pt,draw=black},
connect/.style={draw=black,-},
edgestyle/.style={-},
>=stealth]

\def\side{2}
\def\legendpos{-5}
\def\costhirty{0.8660256}
\def\cosfortyfive{0.7071068}

\node (NRD) at (0,\side) {\eqref{eq:NRD}};
\node (NEW) at (-\costhirty*\side,-\side/2) {\eqref{eq:NEW}};
\node (NRL) at (\costhirty*\side,-\side/2) {\eqref{eq:NRL}};



\draw (NRD) [->,bend right] to node [midway,left] {\small$\hphantom{\sum_{\lvl=\lolvl+1}^{\nLvls} \diff\atlvl{\lvl}}\dis\temp\atlvl{\lvl} = \frac{1}{\sum_{\lolvl=0}^{\lvl-1}\rate\atlvl{\lolvl}}$} (NEW);

\draw (NEW) [->,bend right] to node [midway,below] {\small$\diff\atlvl{\lolvl} = \temp\atlvl{\lolvl} - \temp\atlvl{\lolvl+1}$} (NRL);

\draw (NRL) [->,bend right] to node [midway,right] {\small$\dis\rate\atlvl{\lolvl} = \frac{\diff\atlvl{\lolvl}}{\sum_{\lvl=\lolvl}^{\nLvls} \diff\atlvl{\lvl} \cdot \sum_{\lvl=\lolvl+1}^{\nLvls} \diff\atlvl{\lvl}}$} (NRD);

\end{tikzpicture}
\caption{The triple equivalence between \eqref{eq:NRD}, \eqref{eq:NEW}, and \eqref{eq:NRL}.}
\label{fig:triple}
\end{figure}


With this definition in hand, we have the following equivalence:

\begin{theorem}
\label{thm:NLC}
For all $\score\in\R^{\pures}$, the probability distribution \eqref{eq:NLC-elem} can be written as
\begin{equation}
\label{eq:NLC-reg}
\choice(\score)
	= \argmax_{\strat\in\strats} \braces*{\braket{\score}{\strat} - \hreg(\strat)}
\end{equation}
with $\hreg(\strat)$ given by \eqref{eq:NER}.
Consequently $\orbit{\time} = \choice(\score(\time))$ is a solution orbit of \eqref{eq:NEW} if and only if it is a solution orbit of the \acl{NRL} dynamics
\begin{equation}
\label{eq:NRL}
\tag{NRL}
\dot\score
	= \payv(\strat)
	\qquad
\strat
	= \argmax_{\stratalt\in\strats}
		\braces*{
			\braket{\score}{\stratalt}
			- \sum_{\lolvl=0}^{\nLvls} \diff\atlvl{\lolvl}
				\sum_{\class\atlvl{\lolvl} \in \classes\atlvl{\lolvl}}
					\stratalt_{\class\atlvl{\lolvl}} \log\stratalt_{\class\atlvl{\lolvl}}
		}
\end{equation}
with $\diff\atlvl{\lolvl}$ given by \eqref{eq:temp2diff}.
\end{theorem}

Thus, combining \cref{thm:NEW,thm:NLC}, we obtain a triple equivalence between the revisionist, reinforcement, and regularized learning viewpoints, as represented schematically in \cref{fig:triple} (the proof of \cref{thm:NLC} is presented in \cref{app:NLC}).
Importantly, the reagularized dynamics \eqref{eq:RL} \textendash\ of which \eqref{eq:NRL} constitute a special case \textendash\ have been studied extensively in the online learning literature, so this triple equivalence can be used to infer a number of additional properties for \eqref{eq:NRD} which would otherwise be particularly difficult to piece together \textendash\ from recurrence in zero-sum \citep{MPP18} and harmonic games \citep{LMP24} to evolutionary permanence and impermanence in the spirit of \citet{HS88}.
We defer such examinations to future work.

\appendix
\setcounter{remark}{0}
\numberwithin{equation}{section}	
\numberwithin{lemma}{section}	
\numberwithin{proposition}{section}	
\numberwithin{theorem}{section}	
\numberwithin{corollary}{section}	

\section{Proof of \cref{thm:NRD}}
\label{app:NRD}

Our aim in this appendix is to prove \cref{thm:NRD,prop:domrate} on the rationality and convergence properties of the dynamics \eqref{eq:NRD}/\eqref{eq:NEW}.
We begin with a short calculation which is very useful in several parts of our proof:

\begin{lemma}
\label{lem:paydiff}
Let $\orbit{\time}$ be a solution orbit of \eqref{eq:NRD}, and set
\begin{equation}
\label{eq:KL-diff}
\lyap_{\pure\purealt}(\time)
	= \breg_{\nLvls}(\bvec_{\pure},\orbit{\time}) - \breg_{\nLvls}(\bvec_{\purealt},\orbit{\time})
	= \sum_{\lvl=1}^{\nLvls} \diff\atlvl{\lvl} \log\frac{\orbit[{\classof[\lvl]{\purealt}}]{\time}}{\orbit[{\classof[\lvl]{\pure}}]{\time}},
\end{equation}
for all $\pure,\purealt\in\pures$.
Then $\dot\lyap_{\pure\purealt}(\time) = \payv_{\purealt}(\orbit{\time}) - \payv_{\pure}(\orbit{\time})$.
\end{lemma}

\begin{proof}
By \cref{prop:KL}, we readily get
\begin{equation*}
\dot\lyap_{\pure\purealt}
	= \braket{\payv(\strat)}{\strat - \bvec_{\pure}} - \braket{\payv(\strat)}{\strat - \bvec_{\purealt}}
	= \braket{\payv(\strat)}{\bvec_{\purealt} - \bvec_{\pure}}
	= \payv_{\purealt}(\strat) - \payv_{\pure}(\strat).\qedhere
\end{equation*}
\end{proof}

We are now in a position to carry out the proof of \cref{thm:NRD}, one property at a time.

\para{Proof of \cref{part:dom}: Elimination of strictly dominated strategies}
Suppose that $\pure\in\pures$ is strictly dominated by $\purealt\in\pures$ and let $\orbit{\time}$ be a solution of \eqref{eq:NRD} with $\strat(\tstart) \in \relint(\strats)$, so $\lyap_{\pure\purealt}(\tstart)<\infty$.
Then, \cref{lem:paydiff} implies that $\dot\lyap_{\pure\purealt}\geq \paydiff$
for some $\paydiff>0$, from which it follows that
\begin{equation}
\label{eq:dom-Lyap-1}
-\sum_{\lvl=1}^{\nLvls} \diff\atlvl{\lvl} \log \orbit[{\classof[\lvl]{\pure}}]{\time}
	\geq \lyap_{\pure\purealt}(\time) \geq \lyap_{\pure\purealt}(\tstart) + \paydiff\time
	\to \infty
	\quad
	\text{as $t\to\infty$}.
\end{equation}
Since $\strat_{\classof[\lvl]{\pure}} \geq \strat_{\pure}$ for all $\lvl=1,\dotsc,\nLvls$, we obtain
\begin{equation}
\sum_{\lvl=1}^{\nLvls} \diff\atlvl{\lvl} \log \orbit[{\classof[\lvl]{\pure}}]{\time}
	\geq \sum_{\lvl=1}^{\nLvls} \diff\atlvl{\lvl} \log\orbit[\pure]{\time}
\end{equation}
so, by \eqref{eq:dom-Lyap-1}, we conclude that $\orbit[\pure]{\time} \to 0$ as $\time\to\infty$, \ie $\pure$ becomes extinct along $\orbit{\time}$.
\hfill
\qed

\para{Proof of \cref{part:stat}: Stationary points}
By definition, $\eq$ is a \acl{RE} of $\game$ if and only if $\payv_{\pure}(\eq) = \payv_{\purealt}(\eq)$ whenever $\pure,\purealt\in\supp(\eq)$.
Thus, stationarity of \aclp{RE} in \eqref{eq:NRD} follows immediately by observing that $\payv_{\pure}(\eq) = \classpay_{\classof[\lvl]{\pure}}(\eq)$ whenever $\eq_{\pure}>0$.

Conversely, assume that $\eq$ is a stationary point of \eqref{eq:NRD} which is not a \acl{RE} of $\game$.
Then, for all $\pure\in\argmax_\purealt\payv_{\purealt}(\eq)$ and all $\lvl=1,\dotsc,\nLvls$, we have $\payv_{\pure}(\eq) \geq \classpay_{\classof[\lvl]{\pure}}(\eq)$, with the inequality being strict for $\lvl < \nLvls$.
Since $\rate_{0}>0$, it follows that the \acl{RHS} of \eqref{eq:NRD} is strictly positive at $\eq$, contradicting the assumption that $\eq$ is stationary.
\hfill
\qed

\para{Proof of \cref{part:limit,part:Lyap}: Limits of interior solutions and Lyapunov stability}
We will prove both claims in tandem by showing the following more general result:
if every neighborhood $\nhd$ of $\eq\in\strats$ admits an interior solution $\orbit{\time}$ that remains in $\nhd$ for all $\time\geq0$, then $\eq$ is a \acl{NE}.

To establish this more general claim, assume to the contrary that $\eq$ is not a \acl{NE}, so $\payv_{\pure}(\eq) < \payv_{\purealt}(\eq)$ for some $\pure\in\supp(\eq)$, $\purealt\in\pures$.
Then there exists a neighborhood $\nhd$ of $\eq$ in $\strats$ such that
$\strat_{\pure} > \eq_{\pure}/2$
and
$\payv_{\purealt}(\strat) - \payv_{\pure}(\strat) \geq \paydiff$ for some $\paydiff>0$ and for all $\strat\in\nhd$.
By assumption, there exists an interior solution orbit $\orbit{\time}$ of \eqref{eq:NRD} that is contained in $\nhd$ for all $\time\geq0$.  
Since $\payv_{\pure}(\orbit{\time}) -\payv_{\purealt}(\orbit{\time}) <- \paydiff$, invoking \cref{lem:paydiff} in the same way as in the proof of Part 1 implies that $\lyap_{\pure\purealt}(\time)\to\infty$ as $t\to\infty$.
However, since $\orbit[\pure]{\time} > \eq_{\pure}/2$ for all $\time\geq0$ by the definition of $\nhd$, we also have that
\begin{equation}
\lyap_{\pure\purealt}(\time)
	\leq \sum_{\lvl=1}^{\nLvls} \log\frac{\strat_{\classof[\lvl]{\purealt}}(\time)}{\eq_{\pure}/2}
	\leq \sum_{\lvl=1}^{\nLvls} \log\frac{2}{\eq_{\pure}}
	< \infty,
\end{equation}
a contradiction (recall here that $\pure\in\supp(\eq)$).
In turn, going back to our original working assumption, we conclude that $\eq$ must be a \acl{NE} of $\game$, as claimed.
\hfill
\qed

\para{Proof of \cref{part:ESS}: Convergence to \aclp{ESS}}
We consider the case in which $\eq$ is a \acs{GESS}; our claim for non-global \acp{ESS} follows by a straightforward localization of our argument to a sufficiently small neighborhood of $\eq$ defined by a level set of $\breg_{\nLvls}(\eq,\cdot)$.

To begin, define the \emph{effective domain} $\good(\eq)$ of $\breg_{\nLvls}(\eq,\cdot)$ as
\begin{equation}
\label{eq:good}
\good(\eq)
	\equiv \setdef{\strat\in\strats}{\breg_{\nLvls}(\eq,\strat) < \infty}
	= \setdef{\strat\in\strats}{\supp(\eq) \subseteq \supp(\strat)}.
\end{equation}
Since every solution of \eqref{eq:NRD} has constant support, $\good(\eq)$ is an invariant set of \eqref{eq:NRD}.
Moreover, for any solution $\orbit{\time}$ contained in $\good(\eq)$, \cref{prop:KL} and the definition of \acs{GESS} imply that
\begin{equation}
\label{eq:dKL-monotone}
\frac{d}{dt} \breg_{\nLvls}(\eq,\orbit{\time})
	= \braket{\payv(\orbit{\time})}{\orbit{\time} - \eq}
	\leq 0,
\end{equation}
with equality if and only if $\orbit{\time} = \eq$.%
\footnote{In other words, $\breg_{\nLvls}(\eq,\orbit{\time})$ is a strict Lyapunov function for \eqref{eq:NRD} on $\good(\eq)$.}

Consequently, to establish that $\eq$ is asymptotically stable with basin $\good(\eq)$, it suffices to show that $\orbit{\time}$ has no $\omega$-limit points in $\strats\setminus\good(\eq)$ (\cf \citealp[Appendix 7.B]{San10}).
To do so, assume to the contrary that $\orbit{\time}$ admits an $\omega$-limit $\olim\neq\eq$, so $\strat(\time_{\run})\to\olim$ for some divergent sequence of times $\time_{\run}\uparrow\infty$.
By the definition of \eqref{eq:NRD}, we further have
\begin{align}
\abs{\dot\strat_{\pure}}
	&\leq \strat_{\pure}
		\sum_{\lolvl=0}^{\nLvls-1}
			\rate\atlvl{\lolvl}
			\abs{\payv_{\pure}(\strat) - \payv_{\classof[\lolvl]{\pure}}(\strat)}
	\notag\\
	&\leq \sum_{\lolvl=0}^{\nLvls-1} \rate\atlvl{\lolvl}
		\cdot \max_{\strat\in\strats} \abs{\payv_{\pure}(\strat) - \payv_{\classof[\lolvl]{\pure}}(\strat)}
	\leq 2 \max_{\strat\in\strats} \supnorm{\payv(\strat)}
\end{align}
where, in the last line, we used the facts that $\strat_{\pure} \leq 1$ and that $\abs{\classpay_{\class}(\strat)} \leq \max_{\pure\in\class} \abs{\payv_{\pure}(\strat)}$.%
\footnote{The notation $\supnorm{\cdot}$ stands here for the sup-norm on $\R^{\pures}$, viz. $\supnorm{\py} = \max_{\pure} \abs{\py_{\pure}}$.}
It thus follows that there exists an open neighborhood $\nhd$ of $\olim$, scalars $\paydiff,\delta > 0$ and an index $\run_{0}\geq1$ such that $\orbit{\time} \in \nhd$ and $\braket{\payv(\orbit{\time})}{\orbit{\time} - \eq} \leq -\paydiff < 0$ for all $\time\in[\time_{\run},\time_{\run}+\delta]$ and all $\run\geq \run_{0}$.
Hence, using \eqref{eq:dKL-monotone}, we get
\begin{equation}\label{eq:and}
\breg_{\nLvls}(\eq,\strat(\time_{\run}+\delta)) - \breg_{\nLvls}(\eq,\strat(\tstart))
	\leq \int_{0}^{\time_{\run}+\delta} \braket{\payv(\strat(s))}{\strat(\timealt) - \eq} \dd\timealt
	\leq -\paydiff \delta (\run - \run_{0}).
\end{equation}
Hence, letting $\run\to\infty$, we get $\liminf_{t\to\infty} \breg_{\nLvls}(\eq,\orbit{\time}) = -\infty$, a contradiction which proves our original assertion.
\hfill
\qed

\para{Proof of \cref{part:pot}: Convergence in potential games}
We will establish the dynamics' convergence in potential games by showing that the game's potential is a strict (increasing) Lyapunov function for \eqref{eq:NRD}.
Indeed, if $\payv = \nabla\pot$ for some potential function $\pot$ on $\strats$, we have:
\begin{flalign}
\dot\pot
	= \sum_{\pure\in\pures} \frac{\pd\pot}{\pd \strat_{\pure}} \,\dot \strat_{\pure}
	&= \sum_{\pure\in\pures} \payv_{\pure}(\strat) \, \strat_{\pure}
		\, \sum_{\lolvl=0}^{\nLvls-1}
			\rate_{\lolvl}
			\, \bracks*{\payv_{\pure}(\strat) - \classpay_{\classof[\lolvl]{\pure}}(\strat)}
	\explain{by \eqref{eq:NRD}}
	\\
	&= \sum_{\lolvl=0}^{\nLvls-1} \rate_{\lolvl}
		\sum_{\pure\in\pures} \strat_{\pure}
			\payv_{\pure}(\strat)
			\, \bracks*{\payv_{\pure}(\strat) - \classpay_{\classof[\lolvl]{\pure}}(\strat)}
	\explain{rearrange product}
	\\
	&= \sum_{\lolvl=0}^{\nLvls-1} \rate_{\lolvl}
		\sum_{\class\atlvl{\lolvl}\in\classes\atlvl{\lolvl}}
		\sum_{\pure\in\class\atlvl{\lolvl}}
			\strat_{\pure}
			\payv_{\pure}(\strat)
			\, \bracks*{\payv_{\pure}(\strat) - \classpay_{\class\atlvl{\lolvl}}(\strat)}
	\explain{collect parent classes}
	\\
	&\label{eq:dpot}
	= \sum_{\lolvl=0}^{\nLvls-1}
		\rate_{\lolvl}
		\sum_{\class_{\lolvl}\in\classes_{\lolvl}}
			\strat_{\class_{\lolvl}}
			\bracks*{\sum_{\pure\in\class_{\lolvl}}
				\strat_{\pure \vert \class_{\lolvl}}
				\payv_{\pure}^{2}(\strat)
			- \bracks*{\sum_{\pure\in\class_{\lolvl}} \strat_{\pure \vert \class_{\lolvl}}
			\payv_{\pure}(\strat)}^{2}}.
\end{flalign}
By Jensen's inequality applied to the vector of conditional probabilities $(\strat_{\pure\vert\class\atlvl{\lolvl}})_{\pure\in\class\atlvl{\lolvl}} \in \simplex(\class\atlvl{\lolvl})$, we conclude that every inner summand in \eqref{eq:dpot} is non-negative and, moreover, it vanishes if and only if $\strat_{\pure\vert\class\atlvl{\lolvl}} = 1$ for some $\pure\in\class\atlvl{\lolvl}$.%
\footnote{Note also that each of these terms can be interpreted as the conditional variance of $\payv_{\pure}(\strat)$ restricted to class $\class\atlvl{\lolvl}$ of $\classes\atlvl{\lolvl}$.}
We thus conclude that $\dot\pot(\orbit{\time})\geq0$ with equality if and only if $\orbit{\time}$ is a \acl{RE} of $\game$.
Our claim then follows from \cref{part:stat} of our theorem combined with standard Lyapunov function arguments \textendash\ cf.~\citet[Appendix 7.B]{San10}.
\hfill
\qed

\para{Proof of \cref{part:mon}: Convergence in strictly monotone games}
Let $\eq$ be a \acl{NE} of $\game$ (that every population game $\game\equiv\gamefull$ admits an equilibrium follows from standard fixed point arguments, \cf \citealp{San10}).
Then, by combining the variational characterization \eqref{eq:Nash-var} of $\sol$ and the monotonicity postulate \eqref{eq:monotone} for $\game$, we readily get
\begin{equation}
\braket{\payv(\strat)}{\strat - \eq}
	\leq \braket{\payv(\eq)}{\strat - \eq}
	\leq 0
\end{equation}
with equality holding if and only $\strat = \eq$ (by the strict monotonicity of $\game$).
In turn, this implies that the (necessarily unique) equilibrium $\eq$ of a strictly monotone game is a \acs{GESS} thereof, so our convergence claim follows from \cref{part:ESS} of the theorem.
\hfill
\qed

We conclude this appendix with the proof of our result on the rate of extinction of dominated strategies in the presence of similarities.

\begin{proof}[Proof of \cref{prop:domrate}]
Let $\lvl=\deg(\pure,\purealt)$ be the degree of similarity between $\pure$ and $\purealt$.
Then, by the definition \eqref{eq:KL-diff} of $\lyap_{\pure\purealt}$, we have
\begin{equation}
\lyap_{\pure\purealt}(\time)
	= \sum_{\lvlalt=1}^{\nLvls} \diff\atlvl{\lvlalt} \log\frac{\orbit[{\classof[\lvlalt]{\purealt}}]{\time}}{\orbit[{\classof[\lvlalt]{\pure}}]{\time}}
	= \sum_{\lvlalt=\lvl+1}^{\nLvls} \diff\atlvl{\lvlalt} \log\frac{\orbit[{\classof[\lvlalt]{\purealt}}]{\time}}{\orbit[{\classof[\lvlalt]{\pure}}]{\time}}
\end{equation}
on account of the fact that $\classof[\lvlalt]{\pure} = \classof[\lvlalt]{\purealt}$ for all $\lvlalt=1,\dotsc,\lvl$.
Since $\strat_{\classof[\lvlalt]{\purealt}} \leq 1$ and $\strat_{\classof[\lvlalt]{\pure}} \geq \strat_{\pure}$ for all $\lvlalt=\lvl+1,\dotsc,\nLvls$, we further get
\begin{equation}
\label{eq:dom-Lyap-2}
-\sum_{\lvlalt=\lvl+1}^{\nLvls} \diff\atlvl{\lvlalt} \log\orbit[\pure]{\time}
	\geq \sum_{\lvlalt=\lvl+1}^{\nLvls} \diff\atlvl{\lvlalt} \log\frac{\orbit[{\classof[\lvlalt]{\purealt}}]{\time}}{\orbit[{\classof[\lvlalt]{\pure}}]{\time}}
	= \lyap_{\pure\purealt}(\time)
	\geq \lyap_{\pure\purealt}(\tstart) + \paydiff\time
\end{equation}
by \eqref{eq:dom-Lyap-1}.
However, by \eqref{eq:sumweight}, we have
\begin{equation}
\sum_{\lvlalt=\lvl+1}^{\nLvls} \diff\atlvl{\lvlalt}
	= \frac{1}{\sumrate\atlvl{\lvl}}
	= \frac{1}{\sum_{\lolvl=0}^{\lvl} \rate_{\lolvl}}
\end{equation}
so our assertion follows by substituting the above in \eqref{eq:dom-Lyap-2} and solving for $\orbit[\pure]{\time}$.
\end{proof}

\section{Proof of \cref{thm:NLC}}
\label{app:NLC}

Our aim in this appendix is to prove the representation of the \acl{NLC} rule \eqref{eq:NLC-elem} as a regularized best response model (\cf \cref{thm:NLC} in \cref{sec:NLC}).%
\footnote{The proof presented here closely follows the conference paper of \cite{MMRZ22}, which was published when the current paper was under review.}
To that end, given a similarity structure $\struct$ on $\pures$ and a sequence of uncertainty parameters $\temp\atlvl{1} \geq \dotsm \geq \temp\atlvl{\nLvls} > 0$ (with $\temp\atlvl{\nLvls+1} = 0$ by convention), it will be convenient to introduce the following (conditional) measures of entropy:

\begin{enumerate}
\addtolength{\itemsep}{\smallskipamount}

\item
The \emph{nested entropy} of $\strat\in\simplex(\pures)$ relative to $\class\in\classes\atlvl{\lvl}$:
\begin{equation}
\label{eq:entropy-nest}
\hreg_{\class}(\strat)
	= \sum_{\lvlalt=\lvl}^{\nLvls} \diff\atlvl{\lvlalt}
		\;\sum_{\mathclap{\class\atlvl{\lvlalt} \desceq[\lvlalt] \class}}\;
			\strat_{\class\atlvl{\lvlalt}} \log\strat_{\class\atlvl{\lvlalt}}
\end{equation}
where $\diff\atlvl{\lvlalt} = \temp\atlvl{\lvlalt} - \temp\atlvl{\lvlalt+1}$ for all $\lvlalt = 1,\dotsc,\nLvls$.

\item
The \emph{restricted entropy} of $\strat\in\simplex(\pures)$ relative to $\class\in\classes\atlvl{\lvl}$:
\begin{equation}
\label{eq:entropy-restr}
\hreg_{\vert\class}(\strat)
	= \hreg_{\class}(\strat)
		+ \chi_{\simplex(\class)}(\strat)
	= \begin{cases}
		\hreg_{\class}(\strat)
			&\quad
			\text{if $\strat\in\simplex(\class)$},
		\\
		\infty
			&\quad
			\text{otherwise},
	\end{cases}
\end{equation}
where $\chi_{\simplex(\class)}$ denotes the characteristic function of $\simplex(\class)$, \viz 
$\chi_{\simplex(\class)}(\strat) = 0$ if $\strat\in\simplex(\class)$ and $\chi_{\simplex(\class)}(\strat) = \infty$ otherwise.
[Obviously, $\hreg_{\vert\class}(\strat) = \hreg_{\class}(\strat)$ whenever $\strat \in \simplex(\class)$.]

\item
The \emph{conditional entropy} of $\strat\in\simplex(\pures)$ relative to a target class $\class\in\classes\atlvl{\lvl}$:
\begin{equation}
\label{eq:entropy-cond}
\hreg(\strat \vert \class)
	= \temp\atlvl{\lvl+1} \sum_{\classalt \childof \class}
		\strat_{\classalt} \log \frac{\strat_{\classalt}}{\strat_{\class}}
	= \temp\atlvl{\lvl+1} \, \strat_{\class}
		\sum_{\classalt \childof \class}
			\strat_{\classalt \vert \class} \log\strat_{\classalt \vert \class}.
\end{equation}

\end{enumerate}

\begin{remark}
As per our standard conventions, we are treating $\class$ interchangeably as a subset of $\pures$ or as an element of $\struct$;
by analogy, to avoid notational inflation, we are also viewing $\simplex(\class)$ as a subset of $\simplex(\pures)$ \textendash\ more precisely, a face thereof.
Finally, in all cases, the functions $\hreg(\strat\vert\class)$, $\hreg_{\class}(\strat)$ and $\hreg_{\vert\class}(\strat)$ are assumed to take the value $+\infty$ for $\strat \in \R^{\pures} \setminus \simplex(\pures)$.
\endenv
\end{remark}

\begin{remark}
For posterity, we also note that the nested and restricted entropy functions ($\hreg_{\class}(\strat)$ and $\hreg_{\vert\class}(\strat)$ respectively) are both convex over $\simplex(\pures)$.
This is a consequence of the fact that each summand $\strat_{\class} \log\strat_{\class}$ in \eqref{eq:entropy-nest} is convex in $\strat$ and that $\diff\atlvl{\lvlalt} = \temp\atlvl{\lvlalt} - \temp\atlvl{\lvlalt+1} \geq 0$ for all $\lvlalt=1,\dotsc,\nLvls$.
Of course, any two distributions $\strat,\stratalt\in\simplex(\pures)$ that assign the same probabilities to elements of $\class$ but not otherwise have $\hreg_{\class}(\strat) = \hreg_{\class}(\stratalt)$, so $\hreg_{\class}$ is \emph{not} strictly convex over $\simplex(\pures)$ if $\class \neq \pures$.
However, since the function $\sum_{\pure\in\class} \strat_{\pure} \log\strat_{\pure}$ is strictly convex over $\simplex(\class)$, it follows that $\hreg_{\class}$ \textendash\ and hence $\hreg_{\vert\class}$ \textendash\ \emph{is} strictly convex over $\simplex(\class)$.
\endenv
\end{remark}

As a preamble to proving \cref{thm:NLC}, we will require the technical result below, which illustrates the precision connection between the nested and conditional entropy functions:
\smallskip

\begin{proposition}
\label{prop:nest2cond}
For all $\class \in \classes\atlvl{\lvl}$, $\lvl=1,\dotsc,\nLvls$, and for all $\strat\in\simplex(\pures)$, we have:
\begin{align}
\label{eq:nest2cond}
\hreg_{\class}(\strat)
	&= \sum_{\classalt\desceq\class} \hreg(\strat \vert \classalt)
	+ \temp\atlvl{\lvl} \, \strat_{\class} \log\strat_{\class}.
\intertext{Consequently, for all $\strat\in\simplex(\class)$, we have:}
\label{eq:restr2cond}
\hreg_{\vert\class}(\strat)
	&= \sum_{\classalt\desceq\class} \hreg(\strat \vert \classalt).
\end{align}
\end{proposition}

\begin{proof}
Let $\lvl = \attof{\class}$, and fix some attribute label $\lvlalt > \lvl$.
We will proceed inductively by collecting all terms in \eqref{eq:nest2cond} associated to the attribute $\classes\atlvl{\lvlalt}$ and then summing everything together.
Indeed, we have:
\begin{subequations}
\label{eq:ent-nest1}
\begin{align}
\temp\atlvl{\lvlalt}
	\sum_{\classalt \desceq[\lvlalt] \class}
		\strat_{\classalt} \log\strat_{\classalt}
	&= \temp\atlvl{\lvlalt}
		\sum_{\class\atlvl{\lvlalt-1} \desceq[\lvlalt-1] \class}
			\bracks*{
				\sum_{\child \childof \class\atlvl{\lvlalt-1}}
					\strat_{\child} \log\strat_{\child}
				}
	\explain{collect attributes}
	\\
	&= \temp\atlvl{\lvlalt}
		\sum_{\class\atlvl{\lvlalt-1} \desceq[\lvlalt-1] \class}
			\bracks*{
				\sum_{\child \childof \class\atlvl{\lvlalt-1}}
					\strat_{\child \vert \class\atlvl{\lvlalt-1}} \strat_{\class\atlvl{\lvlalt-1}}
					\log(\strat_{\child \vert \class\atlvl{\lvlalt-1}} \strat_{\class\atlvl{\lvlalt-1}})
				}
	\notag\\
	&= \temp\atlvl{\lvlalt}
		\sum_{\class\atlvl{\lvlalt-1} \desceq[\lvlalt-1] \class} 
			\bracks*{
				\sum_{\child \childof \class\atlvl{\lvlalt-1}}
					\strat_{\child \vert \class\atlvl{\lvlalt-1}} \strat_{\class\atlvl{\lvlalt-1}}
					\log\strat_{\child \vert \class\atlvl{\lvlalt-1}}
				}
	\label{eq:ent-nest1a}\\
	&+ \temp\atlvl{\lvlalt}
		\sum_{\class\atlvl{\lvlalt-1} \desceq[\lvlalt-1] \class}
			\bracks*{
				\sum_{\child \childof \class\atlvl{\lvlalt-1}}
					\strat_{\child \vert \class\atlvl{\lvlalt-1}} \strat_{\class\atlvl{\lvlalt-1}}
					\log\strat_{\class\atlvl{\lvlalt-1}}
				}
	\label{eq:ent-nest1b}
\end{align}
\end{subequations}
with the tacit understanding that any empty sum that appears above is taken equal to zero.

Now, by the definition of the nested entropy, we readily obtain that
\begin{subequations}
\label{eq:ent-nest2}
\begin{equation}
\label{eq:ent-nest2a}
\eqref{eq:ent-nest1a}
	= \sum_{\class\atlvl{\lvlalt-1} \desceq[\lvlalt-1] \class}
			\hreg(\strat \vert \class\atlvl{\lvlalt-1})
\end{equation}
whereas, by noting that $\sum_{\child \childof \class\atlvl{\lvlalt-1}} \strat_{\child \vert \class\atlvl{\lvlalt-1}} = 1$ (by the definition of conditional class choice probabilities), \cref{eq:ent-nest1b} becomes
\begin{equation}
\label{eq:ent-nest2b}
\eqref{eq:ent-nest1b}
	= \temp\atlvl{\lvlalt}
		\sum_{\class\atlvl{\lvlalt-1} \desceq[\lvlalt-1] \class}
			\strat_{\class\atlvl{\lvlalt-1}} \log\strat_{\class\atlvl{\lvlalt-1}}.
\end{equation}
\end{subequations}
Hence, combining \cref{eq:ent-nest1,eq:ent-nest2a,eq:ent-nest2b}, we get:
\begin{equation}
\label{eq:ent-nest3}
\temp\atlvl{\lvlalt}
	\sum_{\classalt \desceq[\lvlalt] \class}
		\strat_{\classalt} \log\strat_{\classalt}
	= \sum_{\class\atlvl{\lvlalt-1} \desceq[\lvlalt-1] \class}
			\hreg(\strat \vert \class\atlvl{\lvlalt-1})
	+ \temp\atlvl{\lvlalt}
		\sum_{\class\atlvl{\lvlalt-1} \desceq[\lvlalt-1] \class}
			\strat_{\class\atlvl{\lvlalt-1}} \log\strat_{\class\atlvl{\lvlalt-1}}.
\end{equation}

The above expression is the basic building block of our inductive construction.
Indeed, summing \eqref{eq:ent-nest3} over all $\lvlalt = \nLvls,\dotsc,\lvl = \attof{\class}$, we obtain:
\begin{align}
\hreg_{\class}(\strat)
	&= \sum_{\lvlalt=\lvl}^{\nLvls}
		(\temp\atlvl{\lvlalt} - \temp\atlvl{\lvlalt+1})
		\;\sum_{\mathclap{\classalt \desceq[\lvlalt] \class}}\;
			\strat_{\classalt} \log\strat_{\classalt}
	\explain{by definition}
	\\
	&= \sum_{\lvlalt=\nLvls}^{\lvl+1}
		\bracks*{
			\temp\atlvl{\lvlalt} \sum_{\mathclap{\classalt \desceq[\lvlalt] \class}}
				\strat_{\classalt} \log\strat_{\classalt}
			- \temp\atlvl{\lvlalt+1} \sum_{\mathclap{\classalt \desceq[\lvlalt] \class}}
				\strat_{\classalt} \log\strat_{\classalt}
			}
	+ (\temp\atlvl{\lvl} - \temp\atlvl{\lvl+1}) \, \strat_{\class} \log\strat_{\class}
	\explain{isolate $\class$}
	\\
	&= \sum_{\lvlalt=\nLvls}^{\lvl+1}
		\bracks*{
			\sum_{\class\atlvl{\lvlalt-1} \desceq[\lvlalt-1] \class}
				\hreg(\strat \vert \class\atlvl{\lvlalt-1})
			+ \temp\atlvl{\lvlalt}
				\sum_{\class\atlvl{\lvlalt-1} \desceq[\lvlalt-1] \class}
					\strat_{\class\atlvl{\lvlalt-1}} \log\strat_{\class\atlvl{\lvlalt-1}}
			- \temp\atlvl{\lvlalt+1} \sum_{\mathclap{\classalt \desceq[\lvlalt] \class}}
				\strat_{\classalt} \log\strat_{\classalt}
			}
	\notag\\
	&\qquad
		+ (\temp\atlvl{\lvl} - \temp\atlvl{\lvl+1}) \, \strat_{\class} \log\strat_{\class}
	\explain{by \eqref{eq:ent-nest3}}
	\\
	&= \sum_{\lvlalt=\lvl}^{\nLvls-1}
		\sum_{\classalt \desceq[\lvlalt] \class} \hreg(\strat \vert \classalt)
	+ \temp\atlvl{\lvl} \,
			\strat_{\class} \log\strat_{\class}
	- \temp\atlvl{\nLvls+1}
		\;\sum_{\mathclap{\classalt \desceq[\nLvls] \class}}\;
			\strat_{\classalt} \log\strat_{\classalt}
\label{eq:ent-nest4}
\end{align}
with the last equality following by telescoping the terms involving $\temp\atlvl{\lvlalt}$.
Now, given that $\temp\atlvl{\nLvls+1} = 0$ by convention, the third sum above is zero.
Finally, since the conditional entropy of $\strat$ relative to any childless class is zero by definition, the first sum in \eqref{eq:ent-nest4} can be rewritten as $\sum_{\lvlalt=\lvl}^{\nLvls-1} \sum_{\classalt \desceq[\lvlalt] \class} \hreg(\strat \vert \classalt) = \sum_{\classalt \desceq \class} \hreg(\strat \vert \classalt)$, and our claim follows.

Finally, \eqref{eq:restr2cond} is a consequence of the fact that $\strat_{\class} = 1$ whenever $\strat \in \simplex(\class)$ \textendash\ \ie whenever $\supp(\strat) \subseteq \class$.
\end{proof}

%

We are now in a position to state and prove a class-based version of \cref{thm:NLC}, which is of independent interest, and which contains \cref{thm:NLC} as a special case:

\begin{proposition}
\label{prop:NLC}
The following holds for all $\class \in \struct$ and all $\score\in\R^{\pures}$:
\begin{enumerate}
[label={\arabic*.},ref=\arabic*]
\item
\label[part]{part:score}
The recursively defined propensity score $\score_{\class}$ of $\class$ as given by \eqref{eq:score-class} can be expressed as
\begin{equation}
\label{eq:ent-max}
\score_{\class}
	= \max_{\strat \in \simplex(\class)}
		\braces{\braket{\score}{\strat} - \hreg_{\class}(\strat)}
\end{equation}

\item
\label[part]{part:choice}
The conditional probability of choosing $\pure\in\pures$ given that $\class$ has already been selected under \eqref{eq:NLC-cond} is given by
\begin{equation}
\label{eq:ent-diff}
\choice[\pure \vert \class](\score)
	= \frac{\pd\score_{\class}}{\pd\score_{\pure}}
\end{equation}
and the conditional probability vector $\choice[\vert\class](\score) = (\choice[\pure\vert\class](\score))_{\pure\in\pures}$ solves the problem \eqref{eq:ent-max}, \viz
\begin{equation}
\label{eq:ent-argmax}
\choice[\vert\class](\score)
	= \argmax_{\strat \in \simplex(\class)}
		\braces{\braket{\score}{\strat} - \hreg_{\class}(\strat)}
\end{equation}
\end{enumerate}
\end{proposition}


\begin{proof}
We begin by noting that the optimization problem \eqref{eq:ent-max} can be written more explicitly as
\begin{equation}
\label{eq:opt-class}
\tag{$\Opt_{\class}$}
\begin{aligned}
\textup{maximize}
	&\quad
	\braket{\score}{\strat} - \hreg_{\class}(\strat),
	\\
\textup{subject to}
	&\txs
	\quad
	\strat \in \simplex(\pures)
	\text{ and }
	\supp(\strat) \subseteq \class.
\end{aligned}
\end{equation}
We will proceed to show that the (unique) solution of \eqref{eq:opt-class} is given by the vector of conditional probabilities $(\choice[\pure \vert \class](\score))_{\pure\in\pures}$.
The expression \eqref{eq:ent-max} for the maximal value of \eqref{eq:opt-class} will then be derived from \cref{prop:nest2cond}, and the differential representation \eqref{eq:ent-diff} will follow from Legendre's identity.
We make all this precise in a series of individual steps below.

\para{Step 1: Optimality conditions for \eqref{eq:opt-class}}

Fix an alternative $\pure\in\class$ with lineage $\{\pure\} \equiv \class\atlvl{\nLvls} \childof \dotsb \childof \class\atlvl{\lvl} \equiv \class$.
Then the definition of the nested entropy gives
\begin{align}
\label{eq:nest-grad}
\frac{\pd\hreg_{\class}}{\pd\strat_{\pure}}
	&= \sum_{\lvlalt=\lvl}^{\nLvls} \diff\atlvl{\lvlalt}
		\;\sum_{\mathclap{\classalt \desceq[\lvlalt] \class}}\;
			\frac{\pd}{\pd\strat_{\pure}}
				\parens{\strat_{\classalt} \log\strat_{\classalt}}
	= \sum_{\lvlalt=\lvl}^{\nLvls} \diff\atlvl{\lvlalt}
		\;\sum_{\mathclap{\classalt \desceq[\lvlalt] \class}}\;
			(1 + \log\strat_{\classalt})
			\frac{\pd\strat_{\classalt}}{\pd\strat_{\pure}}
	\notag\\
	&= \sum_{\lvlalt=\lvl}^{\nLvls} \diff\atlvl{\lvlalt}
		\;\sum_{\mathclap{\classalt \desceq[\lvlalt] \class}}\;
			(1 + \log\strat_{\classalt})
			\oneof{\pure \in \classalt}
	\notag\\
	&= \sum_{\lvlalt=\lvl}^{\nLvls} \diff\atlvl{\lvlalt}
				(1 + \log\strat_{\class\atlvl{\lvlalt}})
	\notag\\
	&= \temp\atlvl{\lvl}
		+ \sum_{\lvlalt=\lvl}^{\nLvls} \diff\atlvl{\lvlalt} \log\strat_{\class\atlvl{\lvlalt}}
\end{align}
This implies that $\pd_{\pure}\hreg_{\class}(\strat) \to -\infty$ whenever $\strat_{\pure}\to0$, so any solution $\strat$ of \eqref{eq:opt-class} must have $\strat_{\pure} > 0$ for all $\pure\in\class$.
In view of this, the first-order optimality conditions for \eqref{eq:opt-class} become
\begin{equation}
\label{eq:KKT}
\score_{\pure}
	- \frac{\pd\hreg_{\class}}{\pd\strat_{\pure}}
	= \score_{\pure}
		- \temp\atlvl{\lvl}
		- \sum_{\lvlalt=\lvl}^{\nLvls} \diff\atlvl{\lvlalt} \log\strat_{\class\atlvl{\lvlalt}}
	= \coef
	\quad
	\text{for all $\pure\in\class$},
\end{equation}
where $\coef$ is the Lagrange multiplier for the equality constraint $\sum_{\pure\in\pures} \strat_{\pure} = 1$.%
\footnote{Since $\strat_{\pure} > 0$ for all $\pure\in\class$, 
the multipliers for the corresponding inequality constraints all vanish by complementary slackness.}
Thus, after rearranging terms and exponentiating, we get
\begin{equation}
\label{eq:chain}
\strat_{\class\atlvl{\nLvls}}^{\diff\atlvl{\nLvls}}
	\cdot \strat_{\class\atlvl{\nLvls-1}}^{\diff\atlvl{\nLvls-1}}
	\dotsm
	\strat_{\class\atlvl{\lvl}}^{\diff\atlvl{\lvl}}
	= \frac{\exp(\score_{\pure})}{\pf},
\end{equation}
for some proportionality constant $\pf \equiv \pf(\score) > 0$.

\para{Step 2: Solving \eqref{eq:opt-class}}

The next step of our proof will focus on unrolling the chain \eqref{eq:chain}, one attribute at a time.
To start, recall that $\diff\atlvl{\nLvls} = \temp\atlvl{\nLvls}$, so \eqref{eq:chain} becomes
\begin{equation}
\label{eq:chain1}
\strat_{\class\atlvl{\nLvls}}
	\cdot \strat_{\class\atlvl{\nLvls-1}}^{\diff\atlvl{\nLvls-1} / \temp\atlvl{\nLvls}}
	\dotsm \strat_{\class\atlvl{\lvl}}^{\diff\atlvl{\lvl} / \temp\atlvl{\nLvls}}
	= \frac
		{\exp(\score_{\class\atlvl{\nLvls}}/\temp\atlvl{\nLvls})}
		{\pf^{1/\temp\atlvl{\nLvls}}},
\end{equation}
where we used the fact that $\class\atlvl{\nLvls} = \pure$ by definition.
Now, since $\class\atlvl{\nLvls-1} \desceq \class\atlvl{\lvl} = \class$, it follows that all children of $\class\atlvl{\nLvls-1}$ are also desendants of $\class$, so \eqref{eq:chain1} applies to all siblings of $\class\atlvl{\nLvls}$ as well.
Hence, summing \eqref{eq:chain1} over $\class\atlvl{\nLvls} \childof \class\atlvl{\nLvls-1}$, we get
\begin{equation}
\label{eq:chain2}
\strat_{\class\atlvl{\nLvls-1}}
	\cdot \strat_{\class\atlvl{\nLvls-1}}^{\diff\atlvl{\nLvls-1} / \temp\atlvl{\nLvls}}
	\dotsm \strat_{\class\atlvl{\lvl}}^{\diff\atlvl{\lvl} / \temp\atlvl{\nLvls}}
	= \frac
		{\exp(\score_{\class\atlvl{\nLvls-1}} / \temp\atlvl{\nLvls})}
		{\pf^{1/\temp\atlvl{\nLvls}}},
\end{equation}
where we used
the definition \eqref{eq:strat-class} of $\strat_{\class\atlvl{\nLvls-1}} = \sum_{\class\atlvl{\nLvls} \childof \class\atlvl{\nLvls-1}} \strat_{\class\atlvl{\nLvls}}$
and the recursive definition \eqref{eq:score-class} for $\score_{\class\atlvl{\nLvls-1}}$, \ie the fact that
$\exp(\score_{\class\atlvl{\nLvls-1}} / \temp\atlvl{\nLvls}) = \sum_{\class\atlvl{\nLvls} \childof \class\atlvl{\nLvls-1}} \exp(\score_{\class\atlvl{\nLvls}} / \temp\atlvl{\nLvls})$.
Therefore, noting that
\begin{equation}
1 + \frac{\diff\atlvl{\nLvls-1}}{\temp\atlvl{\nLvls}}
	= 1 + \frac{\temp\atlvl{\nLvls-1} - \temp\atlvl{\nLvls}}{\temp\atlvl{\nLvls}}
	= \frac{\temp\atlvl{\nLvls-1}}{\temp\atlvl{\nLvls}}
\end{equation}
the product \eqref{eq:chain2} becomes
\begin{equation}
\strat_{\class\atlvl{\nLvls-1}}^{\temp\atlvl{\nLvls-1}}
	\cdot \strat_{\class\atlvl{\nLvls-2}}^{\diff\atlvl{\nLvls-2}}
	\dotsm \strat_{\class\atlvl{\lvl}}^{\diff\atlvl{\lvl}}
	= \frac
		{\exp(\score_{\class\atlvl{\nLvls-1}})}
		{\pf}
\end{equation}
or, equivalently
\begin{equation}
\label{eq:chain3}
\strat_{\class\atlvl{\nLvls-1}}
	\cdot \strat_{\class\atlvl{\nLvls-2}}^{\diff\atlvl{\nLvls-2} / \temp\atlvl{\nLvls-1}}
	\dotsm \strat_{\class\atlvl{\lvl}}^{\diff\atlvl{\lvl} / \temp\atlvl{\nLvls-1}}
	= \frac
		{\exp(\score_{\class\atlvl{\nLvls-1}} / \temp\atlvl{\nLvls-1})}
		{\pf^{1/\temp\atlvl{\nLvls-1}}}.
\end{equation}
This last equation has the same form as \eqref{eq:chain2} applied to the chain $\class\atlvl{\lvl} \parentof \class\atlvl{\lvl+1} \parentof \dotsm \parentof \class\atlvl{\nLvls-1}$ instead of $\class\atlvl{\lvl} \parentof \class\atlvl{\lvl+1} \parentof \dotsm \parentof \class\atlvl{\nLvls}$.
Thus, proceeding inductively, we conclude that
\begin{equation}
\label{eq:chain4}
\strat_{\class\atlvl{\lvlalt}}^{\temp\atlvl{\lvlalt}} \prod_{j=\lvlalt-1}^{\lvl} \strat_{\class\atlvl{j}}^{\diff\atlvl{j}}
	= \frac{\exp(\score_{\class\atlvl{\lvlalt}})}{\pf}
	\quad
	\text{for all $\lvlalt=\nLvls,\dotsc,\lvl$}
\end{equation}
with the empty product $\prod_{j\in\varnothing} \strat_{\class\atlvl{j}}^{\diff\atlvl{j}}$ taken equal to $1$ by standard convention.

Now, substituting $\lvlalt \gets \lvlalt+1$ in \eqref{eq:chain4}, we readily get
\begin{equation}
\label{eq:chain5}
\strat_{\class\atlvl{\lvlalt+1}}^{\temp\atlvl{\lvlalt+1}}
	\cdot \strat_{\class\atlvl{\lvlalt}}^{\diff\atlvl{\lvlalt}} \prod_{j=\lvlalt-1}^{\lvl} \strat_{\class\atlvl{j}}^{\diff\atlvl{j}}
	= \frac{\exp(\score_{\class\atlvl{\lvlalt+1}})}{\pf}
	\quad
	\text{for all $\lvlalt=\nLvls-1,\dotsc,\lvl$}.
\end{equation}
Consequently, recalling that $\diff\atlvl{\lvlalt} = \temp\atlvl{\lvlalt} - \temp\atlvl{\lvlalt+1}$ and dividing \eqref{eq:chain4} by \eqref{eq:chain5}, we get
\begin{equation}
\label{eq:chain6}
\frac
	{\strat_{\class\atlvl{\lvlalt+1}}^{\temp\atlvl{\lvlalt+1}}}
	{\strat_{\class\atlvl{\lvlalt}}^{\temp\atlvl{\lvlalt+1}}}
	= \frac
		{\exp(\score_{\class\atlvl{\lvlalt+1}})}
		{\exp(\score_{\class\atlvl{\lvlalt}})},
\end{equation}
and hence
\begin{equation}
\label{eq:chain-cond}
\frac
	{\strat_{\class\atlvl{\lvlalt+1}}}
	{\strat_{\class\atlvl{\lvlalt}}}
	= \frac
		{\exp(\score_{\class\atlvl{\lvlalt+1}} / \temp\atlvl{\lvlalt+1})}
		{\exp(\score_{\class\atlvl{\lvlalt}} / \temp\atlvl{\lvlalt+1})}
	= \choice[\class\atlvl{\lvlalt+1} \vert \class\atlvl{\lvlalt}](\score)
\end{equation}
by the definition of the conditional logit choice model \eqref{eq:NLC-cond}.
Thus, by unrolling the chain
\begin{align}
\label{eq:chain7}
\strat_{\pure \vert \class}
	= \frac{\strat_{\pure}}{\strat_{\class}}
	&= \frac{\strat_{\class\atlvl{\nLvls}}}{\strat_{\class\atlvl{\nLvls-1}}}
		\cdot \frac{\strat_{\class\atlvl{\nLvls-1}}}{\strat_{\class\atlvl{\nLvls-2}}}
		\dotsm \frac{\strat_{\class\atlvl{\lvl+1}}}{\strat_{\class\atlvl{\lvl}}}
	\notag\\
	&= \choice[\class\atlvl{\nLvls} \vert \class\atlvl{\nLvls-1}](\score)
		\times \choice[\class\atlvl{\nLvls-1} \vert \class\atlvl{\nLvls-2}](\score)
		\times \dotsm
		\times \choice[\class\atlvl{\lvl+1} \vert \class\atlvl{\lvl}](\score)
\end{align}
we obtain the nested expression
\begin{equation}
\strat_{\pure}
	= \strat_{\class} \prod_{\lvlalt=\lvl}^{\nLvls-1} \choice[\class\atlvl{\lvlalt+1} \vert \class\atlvl{\lvlalt}](\score)
	\quad
	\text{for all $\pure\in\class$}.
\end{equation}
Hence, with $\strat_{\class} = 1$ (by the fact that $\supp(\strat) = \class$), we finally conclude that
\begin{equation}
\label{eq:strat-anc}
\strat_{\pure}
	= \prod_{\lvlalt=\lvl}^{\nLvls-1} \choice[\class\atlvl{\lvlalt+1} \vert \class\atlvl{\lvlalt}](\score)
	= \choice[\pure \vert \class](\score)
	\quad
	\text{for all $\pure\in\class$}.
\end{equation}

\para{Step 3: The maximal value of \eqref{eq:opt-class}}

To obtain the value of the maximization problem \eqref{eq:opt-class}, we will proceed to substitute \eqref{eq:strat-anc} in the expression \eqref{eq:nest2cond} provided by \cref{prop:nest2cond} for $\hreg_{\class}(\strat)$.
To that end, for all $\lvlalt = \lvl,\dotsc,\nLvls-1$ and all $\class\atlvl{\lvlalt} \desceq[\lvlalt] \class$, the definition \eqref{eq:entropy-cond} of the conditional entropy gives:
\begin{align}
\hreg(\strat \vert \class\atlvl{\lvlalt})
	&= \temp\atlvl{\lvlalt+1} \, \strat_{\class\atlvl{\lvlalt}}
		\sum_{\class\atlvl{\lvlalt+1} \childof \class\atlvl{\lvlalt}}
			\strat_{\class\atlvl{\lvlalt+1} \vert \class\atlvl{\lvlalt}}
			\log \strat_{\class\atlvl{\lvlalt+1} \vert \class\atlvl{\lvlalt}}
	\explain{by definition}
	\\
	&= \temp\atlvl{\lvlalt+1} \, \strat_{\class\atlvl{\lvlalt}}
		\sum_{\class\atlvl{\lvlalt+1} \childof \class\atlvl{\lvlalt}}
			\strat_{\class\atlvl{\lvlalt+1} \vert \class\atlvl{\lvlalt}}
			\log \frac
				{\exp(\score_{\class\atlvl{\lvlalt+1}} / \temp\atlvl{\lvlalt+1})}
				{\exp(\score_{\class\atlvl{\lvlalt}} / \temp\atlvl{\lvlalt+1})}
	\explain{by \eqref{eq:chain-cond}}
	\\
	&= \strat_{\class\atlvl{\lvlalt}}
		\sum_{\class\atlvl{\lvlalt+1} \childof \class\atlvl{\lvlalt}}
			\strat_{\class\atlvl{\lvlalt+1} \vert \class\atlvl{\lvlalt}}
			\score_{\class\atlvl{\lvlalt+1}}
		- \strat_{\class\atlvl{\lvlalt}} \score_{\class\atlvl{\lvlalt}}
	\explain{since $\sum_{\class\atlvl{\lvlalt+1} \childof \class\atlvl{\lvlalt}} \strat_{\class\atlvl{\lvlalt+1} \vert \class\atlvl{\lvlalt}} = 1$}\\
	&= \sum_{\class\atlvl{\lvlalt+1} \childof \class\atlvl{\lvlalt}}
			\strat_{\class\atlvl{\lvlalt+1}}
			\score_{\class\atlvl{\lvlalt+1}}
		- \strat_{\class\atlvl{\lvlalt}} \score_{\class\atlvl{\lvlalt}}
\end{align}
and hence
\begin{align}
\label{eq:ent-cond}
\sum_{\class\atlvl{\lvlalt} \desceq[\lvlalt] \class}
	\hreg(\strat \vert \class\atlvl{\lvlalt})
	&= \sum_{\class\atlvl{\lvlalt} \desceq[\lvlalt] \class}
		\bracks*{
			\sum_{\class\atlvl{\lvlalt+1} \childof \class\atlvl{\lvlalt}}
				\strat_{\class\atlvl{\lvlalt+1}}
				\score_{\class\atlvl{\lvlalt+1}}
			- \strat_{\class\atlvl{\lvlalt}} \score_{\class\atlvl{\lvlalt}}
			}
	\notag\\
	&= \sum_{\class\atlvl{\lvlalt+1} \desceq[\lvlalt+1] \class}
			\strat_{\class\atlvl{\lvlalt+1}}
			\score_{\class\atlvl{\lvlalt+1}}
		- \sum_{\class\atlvl{\lvlalt} \desceq[\lvlalt] \class}
			\strat_{\class\atlvl{\lvlalt}}
			\score_{\class\atlvl{\lvlalt}}.
\end{align}
Thus, telescoping this last releation over $\lvlalt = \lvl,\dotsc,\nLvls$ and invoking \cref{prop:nest2cond}, we obtain:
\begin{align}
\hreg_{\class}(\strat)
	&= \sum_{\classalt \desceq \class} \hreg(\strat \vert \classalt)
		+ \temp\atlvl{\lvlalt} \, \strat_{\class} \log\strat_{\class}
	\explain{by \cref{prop:nest2cond}}
	\\
	&= \sum_{\lvlalt = \lvl}^{\nLvls-1} \sum_{\class\atlvl{\lvlalt} \desceq[\lvlalt] \class}
		\hreg(\strat \vert \class\atlvl{\lvlalt})
	\explain{collect parent classes}
	\\
	&= \sum_{\lvlalt = \lvl}^{\nLvls-1}
		\bracks*{
			\sum_{\class\atlvl{\lvlalt+1} \desceq[\lvlalt+1] \class}
				\strat_{\class\atlvl{\lvlalt+1}}
				\score_{\class\atlvl{\lvlalt+1}}
			- \sum_{\class\atlvl{\lvlalt} \desceq[\lvlalt] \class}
				\strat_{\class\atlvl{\lvlalt}}
				\score_{\class\atlvl{\lvlalt}}
				}
	\explain{by \eqref{eq:ent-cond}}
	\\
	&= \braket{\score}{\strat} - \strat_{\class} \score_{\class}
\end{align}
where, in the second line, we used the fact that the conditional entropy $\hreg(\strat \vert \class\atlvl{\nLvls})$ relative to any childless class $\class\atlvl{\nLvls} \in \classes\atlvl{\nLvls}$ is zero by definition.
Accordingly, substituting back to \eqref{eq:opt-class} we conclude that
\begin{equation}
\val\eqref{eq:opt-class}
	= \braket{\score}{\strat} - \hreg_{\class}(\strat)
	= \strat_{\class} \score_{\class}
	= \score_{\class},
\end{equation}
as claimed.

\para{Step 4: Differential representation of conditional probabilities}

To prove the second part of the proposition, recall that the restricted entropy function $\hreg_{\vert\class}$ is convex, and let
\begin{equation}
\label{eq:hconj}
\hconj_{\vert\class}(\score)
	= \max_{\strat \in \simplex(\pures)}
		\braces{\braket{\score}{\strat} - \hreg_{\vert\class}(\strat)}
\end{equation}
denote its convex conjugate.%
\footnote{Note here that $\hconj_{\vert\class}(\score)$ is bounded from above by the convex conjugate $\hconj_{\class}(\score)$ of $\hreg_{\class}(\strat)$ because the latter does not include the constraint $\supp(\strat) \subseteq \class$.}
By standard results in convex analysis \citep[Theorem 23.5]{Roc70}, $\hconj_{\vert\class}$ is differentiable in $\score$ and we have the Legendre identity:
\begin{equation}
\label{eq:Legendre}
\strat
	= \nabla\hconj_{\vert\class}(\score)
	\iff
\score
	\in \subd\hreg_{\vert\class}(\strat)
	\iff
\strat
	\in \argmax_{\stratalt\in\simplex(\pures)}
		\braces{\braket{\score}{\stratalt} - \hreg_{\vert\class}(\stratalt)}
\end{equation}
Now, by \eqref{eq:strat-anc}, we have $\strat_{\pure} = \choice[\pure \vert \class](\score)$ whenever $\strat$ solves \eqref{eq:opt-class} and hence, by Fermat's rule, whenever $\score - \subd\hreg_{\vert\class}(\strat) \ni 0$.
Our claim then follows by noting that $\hconj_{\vert\class}(\score) = \score_{\class}$ and combining the first and third legs of the equivalence \eqref{eq:Legendre}.
\end{proof}

The last remaining step is now trivial.

\begin{proof}[Proof of \cref{thm:NLC}]
Simply invoke \cref{prop:NLC} with $\class \gets \pures$.
\end{proof}

\section*{Acknowledgments}
\begingroup
\small
%
%
This research was supported in part by 
the French National Research Agency (ANR) in the framework of
the PEPR IA FOUNDRY project (ANR-23-PEIA-0003),
the ``Investissements d'avenir'' program (ANR-15-IDEX-02),
the LabEx PERSYVAL (ANR-11-LABX-0025-01),
and
MIAI@Grenoble Alpes (ANR-19-P3IA-0003).
PM is also a member of the Archimedes/Athena RC, and was partially supported by project MIS 5154714 of the National Recovery and Resilience Plan Greece 2.0 funded by the European Union under the NextGenerationEU Program.WHS was supported by U.S. NSF Grant SES\textendash1458992 and U.S. ARO Grant MSN201957.
\endgroup

\bibliographystyle{icml}
\bibliography{Bibliography,IEEEabrv}

\begin{thebibliography}{75}
\providecommand{\natexlab}[1]{#1}
\providecommand{\url}[1]{\texttt{#1}}
\expandafter\ifx\csname urlstyle\endcsname\relax
  \providecommand{\doi}[1]{doi: #1}\else
  \providecommand{\doi}{doi: \begingroup \urlstyle{rm}\Url}\fi

\bibitem[Anderson et~al.(1992)Anderson, de~Palma, and Thisse]{AdPT92}
Anderson, S.~P., de~Palma, A., and Thisse, J.-F.
\newblock \emph{Discrete Choice Theory of Product Differentiation}.
\newblock MIT Press, Cambridge, MA, 1992.

\bibitem[Auer et~al.(1995)Auer, Cesa-Bianchi, Freund, and Schapire]{ACBFS95}
Auer, P., Cesa-Bianchi, N., Freund, Y., and Schapire, R.~E.
\newblock Gambling in a rigged casino: The adversarial multi-armed bandit
  problem.
\newblock In \emph{Proceedings of the 36th Annual Symposium on Foundations of
  Computer Science}, 1995.

\bibitem[Auer et~al.(2002)Auer, Cesa-Bianchi, and Fischer]{ACBF02}
Auer, P., Cesa-Bianchi, N., and Fischer, P.
\newblock Finite-time analysis of the multiarmed bandit problem.
\newblock \emph{Machine Learning}, 47:\penalty0 235--256, 2002.

\bibitem[Beggs(2005)]{Beg05}
Beggs, A.~W.
\newblock On the convergence of reinforcement learning.
\newblock \emph{Journal of Economic Theory}, 122:\penalty0 1--36, 2005.

\bibitem[Ben-Akiva(1973)]{BA73}
Ben-Akiva, M.
\newblock \emph{Structure of Passenger Travel Demand Models}.
\newblock PhD thesis, MIT, 1973.

\bibitem[Ben-Akiva \& Lerman(1985)Ben-Akiva and Lerman]{BAL85}
Ben-Akiva, M. and Lerman, S.~R.
\newblock \emph{Discrete Choice Analysis: Theory and Application to Travel
  Demand}.
\newblock MIT Press, Cambridge, 1985.

\bibitem[Bena{\"\i}m \& Weibull(2003)Bena{\"\i}m and Weibull]{BW03}
Bena{\"\i}m, M. and Weibull, J.~W.
\newblock Deterministic approximation of stochastic evolution in games.
\newblock \emph{Econometrica}, 71\penalty0 (3):\penalty0 873--903, May 2003.

\bibitem[Binmore \& Samuelson(1997)Binmore and Samuelson]{BinSam97}
Binmore, K. and Samuelson, L.
\newblock Muddling through: {Noisy} equilibrium selection.
\newblock \emph{Journal of Economic Theory}, 74\penalty0 (2):\penalty0
  235--265, June 1997.

\bibitem[Bj{\"o}rnerstedt \& Weibull(1996)Bj{\"o}rnerstedt and Weibull]{BW96}
Bj{\"o}rnerstedt, J. and Weibull, J.~W.
\newblock Nash equilibrium and evolution by imitation.
\newblock In Arrow, K.~J., Colombatto, E., Perlman, M., and Schmidt, C. (eds.),
  \emph{The Rational Foundations of Economic Behavior}, pp.\  155--181. St.
  Martin's Press, New York, NY, 1996.

\bibitem[B{\"o}rgers \& Mailath(2020)B{\"o}rgers and Mailath]{BorMai20}
B{\"o}rgers, T. and Mailath, G.~J.
\newblock Similarity-based learning and similarity equilibria.
\newblock \url{https://www.youtube.com/watch?v=nvAqFT6WDTU}, 2020.

\bibitem[B{\"o}rgers \& Sarin(1997)B{\"o}rgers and Sarin]{BS97}
B{\"o}rgers, T. and Sarin, R.
\newblock Learning through reinforcement and replicator dynamics.
\newblock \emph{Journal of Economic Theory}, 77\penalty0 (1):\penalty0 1--14,
  1997.

\bibitem[Bravo \& Mertikopoulos(2017)Bravo and Mertikopoulos]{BM17}
Bravo, M. and Mertikopoulos, P.
\newblock On the robustness of learning in games with stochastically perturbed
  payoff observations.
\newblock \emph{Games and Economic Behavior}, 103\penalty0 (John Nash Memorial
  issue):\penalty0 41--66, May 2017.

\bibitem[Coucheney et~al.(2015)Coucheney, Gaujal, and Mertikopoulos]{CGM15}
Coucheney, P., Gaujal, B., and Mertikopoulos, P.
\newblock Penalty-regulated dynamics and robust learning procedures in games.
\newblock \emph{Mathematics of Operations Research}, 40\penalty0 (3):\penalty0
  611--633, August 2015.

\bibitem[Debreu(1960)]{Deb60}
Debreu, G.
\newblock ``{I}ndividual choice behavior: A theoretical analysis'' by {R}.
  {D}uncan {L}uce.
\newblock \emph{American Economic Review}, 50\penalty0 (1):\penalty0 186--188,
  March 1960.

\bibitem[Duvocelle et~al.(2023)Duvocelle, Mertikopoulos, Staudigl, and
  Vermeulen]{DMSV23}
Duvocelle, B., Mertikopoulos, P., Staudigl, M., and Vermeulen, D.
\newblock Multi-agent online learning in time-varying games.
\newblock \emph{Mathematics of Operations Research}, 48\penalty0 (2):\penalty0
  914--941, May 2023.

\bibitem[Erev \& Roth(1998)Erev and Roth]{ER98}
Erev, I. and Roth, A.~E.
\newblock Predicting how people play games: Reinforcement learning in
  experimental games with unique, mixed strategy equilibria.
\newblock \emph{American Economic Review}, 88:\penalty0 848--881, 1998.

\bibitem[Farrell \& Klemperer(2007)Farrell and Klemperer]{FarKle07}
Farrell, J. and Klemperer, P.
\newblock Switching costs and network effects.
\newblock In Armstrong, M. and Porter, R. (eds.), \emph{Handbook of Industrial
  Organization}, volume~3, pp.\  1967--2072. Elsevier, Amsterdam, 2007.

\bibitem[Friedman(1991)]{Fri91}
Friedman, D.
\newblock Evolutionary games in economics.
\newblock \emph{Econometrica}, 59\penalty0 (3):\penalty0 637--666, 1991.

\bibitem[Fudenberg \& Imhof(2006)Fudenberg and Imhof]{FI06}
Fudenberg, D. and Imhof, L.~A.
\newblock Imitation processes with small mutations.
\newblock \emph{Journal of Economic Theory}, 131:\penalty0 251--262, 2006.

\bibitem[Fudenberg \& Imhof(2008)Fudenberg and Imhof]{FI08}
Fudenberg, D. and Imhof, L.~A.
\newblock Monotone imitation dynamics in large populations.
\newblock \emph{Journal of Economic Theory}, 140\penalty0 (1):\penalty0
  229--245, May 2008.

\bibitem[Fudenberg \& Kreps(1993)Fudenberg and Kreps]{FK93}
Fudenberg, D. and Kreps, D.~M.
\newblock Learning mixed equilibria.
\newblock \emph{Games and Economic Behavior}, 5\penalty0 (320-367), 1993.

\bibitem[Fudenberg \& Levine(1995)Fudenberg and Levine]{FL95}
Fudenberg, D. and Levine, D.~K.
\newblock Consistency and cautious fictitious play.
\newblock \emph{Journal of Economic Dynamics and Control}, 19\penalty0
  (5-7):\penalty0 1065--1089, 1995.

\bibitem[Fudenberg \& Levine(1999)Fudenberg and Levine]{FL99}
Fudenberg, D. and Levine, D.~K.
\newblock Conditional universal consistency.
\newblock \emph{Games and Economic Behavior}, 29\penalty0 (1):\penalty0
  104--130, 1999.

\bibitem[Giannou et~al.(2021)Giannou, Vlatakis-Gkaragkounis, and
  Mertikopoulos]{GVM21}
Giannou, A., Vlatakis-Gkaragkounis, E.~V., and Mertikopoulos, P.
\newblock Survival of the strictest: {Stable} and unstable equilibria under
  regularized learning with partial information.
\newblock In \emph{COLT '21: Proceedings of the 34th Annual Conference on
  Learning Theory}, 2021.

\bibitem[Gilboa \& Schmeidler(1995)Gilboa and Schmeidler]{GS95}
Gilboa, I. and Schmeidler, D.
\newblock Case-based decision theory.
\newblock \emph{Quarterly Journal of Economics}, 110\penalty0 (3):\penalty0
  605--639, August 1995.

\bibitem[Hadikhanloo et~al.(2022)Hadikhanloo, Laraki, Mertikopoulos, and
  Sorin]{HLMS22}
Hadikhanloo, S., Laraki, R., Mertikopoulos, P., and Sorin, S.
\newblock Learning in nonatomic games, {Part I}: {Finite} action spaces and
  population games.
\newblock \emph{Journal of Dynamics and Games}, 9\penalty0 (4, William H.
  Sandholm memorial issue):\penalty0 433--460, October 2022.

\bibitem[Helbing(1992)]{Hel92}
Helbing, D.
\newblock A mathematical model for behavioral changes by pair interactions.
\newblock In Haag, G., Mueller, U., and Troitzsch, K.~G. (eds.), \emph{Economic
  Evolution and Demographic Change: Formal Models in Social Sciences}, pp.\
  330--348. Springer, Berlin, 1992.

\bibitem[Hirsch et~al.(2004)Hirsch, Smale, and Devaney]{HSD04}
Hirsch, M.~W., Smale, S., and Devaney, R.~L.
\newblock \emph{Differential Equations, Dynamical Systems, and an Introduction
  to Chaos}.
\newblock Elsevier, London, UK, 2 edition, 2004.

\bibitem[Hofbauer \& Sandholm(2009)Hofbauer and Sandholm]{HS09}
Hofbauer, J. and Sandholm, W.~H.
\newblock Stable games and their dynamics.
\newblock \emph{Journal of Economic Theory}, 144\penalty0 (4):\penalty0
  1665--1693, July 2009.

\bibitem[Hofbauer \& Sigmund(1988)Hofbauer and Sigmund]{HS88}
Hofbauer, J. and Sigmund, K.
\newblock \emph{The Theory of Evolution and Dynamical Systems}.
\newblock Cambridge University Press, 1988.

\bibitem[Hofbauer et~al.(1979)Hofbauer, Schuster, and Sigmund]{HSS79}
Hofbauer, J., Schuster, P., and Sigmund, K.
\newblock A note on evolutionarily stable strategies and game dynamics.
\newblock \emph{Journal of Theoretical Biology}, 81\penalty0 (3):\penalty0
  609--612, 1979.

\bibitem[Hofbauer et~al.(2009)Hofbauer, Sorin, and Viossat]{HSV09}
Hofbauer, J., Sorin, S., and Viossat, Y.
\newblock Time average replicator and best reply dynamics.
\newblock \emph{Mathematics of Operations Research}, 34\penalty0 (2):\penalty0
  263--269, May 2009.

\bibitem[Hopkins(2002)]{Hop02}
Hopkins, E.
\newblock Two competing models of how people learn in games.
\newblock \emph{Econometrica}, 70\penalty0 (6):\penalty0 2141--2166, November
  2002.

\bibitem[Jehiel(2005)]{Jeh05}
Jehiel, P.
\newblock Analogy-based expectation equilibrium.
\newblock \emph{Journal of Economic Theory}, 123\penalty0 (2):\penalty0
  81--104, August 2005.

\bibitem[Jehiel \& Samet(2005)Jehiel and Samet]{JS05}
Jehiel, P. and Samet, D.
\newblock Learning to play games in extensive form by valuation.
\newblock \emph{Journal of Economic Theory}, 124\penalty0 (2):\penalty0
  129--148, October 2005.

\bibitem[Jehiel \& Samet(2007)Jehiel and Samet]{JS07}
Jehiel, P. and Samet, D.
\newblock Valuation equilibrium.
\newblock \emph{Theoretical Economics}, 2\penalty0 (2):\penalty0 163--185,
  2007.

\bibitem[Katz \& Shapiro(1994)Katz and Shapiro]{KS94}
Katz, M.~L. and Shapiro, C.
\newblock Systems competition and network effects.
\newblock \emph{The Journal of Economic Perspectives}, 8\penalty0 (2):\penalty0
  93--115, 1994.

\bibitem[Legacci et~al.(2024)Legacci, Mertikopoulos, and Pradelski]{LMP24}
Legacci, D., Mertikopoulos, P., and Pradelski, B. S.~R.
\newblock A geometric decomposition of finite games: {Convergence} vs.
  recurrence under exponential weights.
\newblock In \emph{ICML '24: Proceedings of the 41st International Conference
  on Machine Learning}, 2024.

\bibitem[Leslie \& Collins(2005)Leslie and Collins]{LC05}
Leslie, D.~S. and Collins, E.~J.
\newblock Individual {$Q$}-learning in normal form games.
\newblock \emph{SIAM Journal on Control and Optimization}, 44\penalty0
  (2):\penalty0 495--514, 2005.

\bibitem[Littlestone \& Warmuth(1994)Littlestone and Warmuth]{LW94}
Littlestone, N. and Warmuth, M.~K.
\newblock The weighted majority algorithm.
\newblock \emph{Information and Computation}, 108\penalty0 (2):\penalty0
  212--261, 1994.

\bibitem[Luce(1956)]{Luc56}
Luce, R.~D.
\newblock Semiorders and a theory of utility discrimination.
\newblock \emph{Econometrica}, 24\penalty0 (2):\penalty0 178--191, April 1956.

\bibitem[Luce(1959)]{Luc59}
Luce, R.~D.
\newblock \emph{Individual Choice Behavior: A Theoretical Analysis}.
\newblock Wiley, New York, 1959.

\bibitem[Martin et~al.(2022)Martin, Mertikopoulos, Rahier, and Zenati]{MMRZ22}
Martin, M., Mertikopoulos, P., Rahier, T., and Zenati, H.
\newblock Nested bandits.
\newblock In \emph{ICML '22: Proceedings of the 39th International Conference
  on Machine Learning}, 2022.

\bibitem[McFadden(1978)]{McF78}
McFadden, D.
\newblock Modelling the choice of residential location.
\newblock In Karlqvist, A., Lundqvist, L., Snickars, F., and Weibull, J.~W.
  (eds.), \emph{Spatial Interaction Theory and Planning Models}. North Holland,
  Amsterdam, 1978.

\bibitem[McFadden(2001)]{McF01}
McFadden, D.
\newblock Economic choices.
\newblock \emph{American Economic Review}, 91:\penalty0 351--378, 2001.

\bibitem[McFadden(1974{\natexlab{a}})]{McF74}
McFadden, D.~L.
\newblock The measurement of urban travel demand.
\newblock \emph{Journal of Public Economics}, 3\penalty0 (4):\penalty0
  303--328, 1974{\natexlab{a}}.

\bibitem[McFadden(1974{\natexlab{b}})]{McF74a}
McFadden, D.~L.
\newblock Conditional logit analysis of qualitative choice behavior.
\newblock In Zarembka, P. (ed.), \emph{Frontiers in Econometrics}, pp.\
  105--142. Academic Press, New York, NY, 1974{\natexlab{b}}.

\bibitem[McFadden(1981)]{McF81}
McFadden, D.~L.
\newblock Econometric models of probabilistic choice.
\newblock In Manski, C.~F. and McFadden, D.~L. (eds.), \emph{Structural
  Analysis of Discrete Data with Econometric Applications}, pp.\  198--272. MIT
  Press, Cambridge, MA, 1981.

\bibitem[Mengel(2012)]{Men12}
Mengel, F.
\newblock Learning across games.
\newblock \emph{Games and Economic Behavior}, 74\penalty0 (2):\penalty0
  601--619, March 2012.

\bibitem[Mertikopoulos \& Moustakas(2010)Mertikopoulos and Moustakas]{MM10}
Mertikopoulos, P. and Moustakas, A.~L.
\newblock The emergence of rational behavior in the presence of stochastic
  perturbations.
\newblock \emph{The Annals of Applied Probability}, 20\penalty0 (4):\penalty0
  1359--1388, July 2010.

\bibitem[Mertikopoulos \& Sandholm(2016)Mertikopoulos and Sandholm]{MS16}
Mertikopoulos, P. and Sandholm, W.~H.
\newblock Learning in games via reinforcement and regularization.
\newblock \emph{Mathematics of Operations Research}, 41\penalty0 (4):\penalty0
  1297--1324, November 2016.

\bibitem[Mertikopoulos \& Sandholm(2018)Mertikopoulos and Sandholm]{MerSan18}
Mertikopoulos, P. and Sandholm, W.~H.
\newblock Riemannian game dynamics.
\newblock \emph{Journal of Economic Theory}, 177:\penalty0 315--364, September
  2018.

\bibitem[Mertikopoulos \& Viossat(2022)Mertikopoulos and Viossat]{MV22}
Mertikopoulos, P. and Viossat, Y.
\newblock Survival of dominated strategies under imitation dynamics.
\newblock \emph{Journal of Dynamics and Games}, 9\penalty0 (4, William H.
  Sandholm memorial issue):\penalty0 499--528, October 2022.

\bibitem[Mertikopoulos \& Zhou(2019)Mertikopoulos and Zhou]{MZ19}
Mertikopoulos, P. and Zhou, Z.
\newblock Learning in games with continuous action sets and unknown payoff
  functions.
\newblock \emph{Mathematical Programming}, 173\penalty0 (1-2):\penalty0
  465--507, January 2019.

\bibitem[Mertikopoulos et~al.(2018)Mertikopoulos, Papadimitriou, and
  Piliouras]{MPP18}
Mertikopoulos, P., Papadimitriou, C.~H., and Piliouras, G.
\newblock Cycles in adversarial regularized learning.
\newblock In \emph{SODA '18: Proceedings of the 29th annual ACM-SIAM Symposium
  on Discrete Algorithms}, 2018.

\bibitem[Mertikopoulos et~al.(2024)Mertikopoulos, Hsieh, and Cevher]{MHC24}
Mertikopoulos, P., Hsieh, Y.-P., and Cevher, V.
\newblock A unified stochastic approximation framework for learning in games.
\newblock \emph{Mathematical Programming}, 203:\penalty0 559--609, January
  2024.

\bibitem[Nachbar(1990)]{Nac90}
Nachbar, J.~H.
\newblock Evolutionary selection dynamics in games.
\newblock \emph{International Journal of Game Theory}, 19:\penalty0 59--89,
  1990.

\bibitem[Posch(1997)]{Pos97}
Posch, M.
\newblock Cycling in a stochastic learning algorithm for normal form games.
\newblock \emph{Journal of Evolutionary Economics}, 7:\penalty0 193--207, 1997.

\bibitem[Rockafellar(1970)]{Roc70}
Rockafellar, R.~T.
\newblock \emph{Convex Analysis}.
\newblock Princeton University Press, Princeton, NJ, 1970.

\bibitem[Rustichini(1999)]{Rus99}
Rustichini, A.
\newblock Optimal properties of stimulus-response learning models.
\newblock \emph{Games and Economic Behavior}, 29\penalty0 (1-2):\penalty0
  244--273, 1999.

\bibitem[Samuelson \& Zhang(1992)Samuelson and Zhang]{SZ92}
Samuelson, L. and Zhang, J.
\newblock Evolutionary stability in asymmetric games.
\newblock \emph{Journal of Economic Theory}, 57:\penalty0 363--391, 1992.

\bibitem[Sandholm(2001)]{San01}
Sandholm, W.~H.
\newblock Potential games with continuous player sets.
\newblock \emph{Journal of Economic Theory}, 97:\penalty0 81--108, 2001.

\bibitem[Sandholm(2010{\natexlab{a}})]{San10}
Sandholm, W.~H.
\newblock \emph{Population Games and Evolutionary Dynamics}.
\newblock MIT Press, Cambridge, MA, 2010{\natexlab{a}}.

\bibitem[Sandholm(2010{\natexlab{b}})]{San10b}
Sandholm, W.~H.
\newblock Pairwise comparison dynamics and evolutionary foundations for {Nash}
  equilibrium.
\newblock \emph{Games}, 1:\penalty0 3--17, 2010{\natexlab{b}}.

\bibitem[Sandholm(2015)]{San15}
Sandholm, W.~H.
\newblock Population games and deterministic evolutionary dynamics.
\newblock In Young, H.~P. and Zamir, S. (eds.), \emph{Handbook of Game Theory
  {IV}}, pp.\  703--778. Elsevier, 2015.

\bibitem[Schlag(1998)]{Sch98}
Schlag, K.~H.
\newblock Why imitate, and if so, how? {A} boundedly rational approach to
  multi-armed bandits.
\newblock \emph{Journal of Economic Theory}, 78\penalty0 (1):\penalty0
  130--156, 1998.

\bibitem[Shalev-Shwartz(2011)]{SS11}
Shalev-Shwartz, S.
\newblock Online learning and online convex optimization.
\newblock \emph{Foundations and Trends in Machine Learning}, 4\penalty0
  (2):\penalty0 107--194, 2011.

\bibitem[Shalev-Shwartz \& Singer(2006)Shalev-Shwartz and Singer]{SSS06}
Shalev-Shwartz, S. and Singer, Y.
\newblock Convex repeated games and {Fenchel} duality.
\newblock In \emph{NIPS' 06: Proceedings of the 19th Annual Conference on
  Neural Information Processing Systems}, pp.\  1265--1272. MIT Press, 2006.

\bibitem[Shub(1987)]{Shu87}
Shub, M.
\newblock \emph{Global Stability of Dynamical Systems}.
\newblock Springer-Verlag, Berlin, 1987.

\bibitem[Sorin \& Wan(2016)Sorin and Wan]{SW16}
Sorin, S. and Wan, C.
\newblock Finite composite games: Equilibria and dynamics.
\newblock \emph{Journal of Dynamics and Games}, 3\penalty0 (1):\penalty0
  101--120, January 2016.

\bibitem[Steiner \& Stewart(2008)Steiner and Stewart]{SteSte08}
Steiner, J. and Stewart, C.
\newblock Contagion through learning.
\newblock \emph{Theoretical Economics}, 3:\penalty0 431--458, 2008.

\bibitem[Taylor \& Jonker(1978)Taylor and Jonker]{TJ78}
Taylor, P.~D. and Jonker, L.~B.
\newblock Evolutionary stable strategies and game dynamics.
\newblock \emph{Mathematical Biosciences}, 40\penalty0 (1-2):\penalty0
  145--156, 1978.

\bibitem[Verboven(1996)]{Ver96}
Verboven, F.
\newblock The nested logit model and representative consumer theory.
\newblock \emph{Economics Letters}, 50:\penalty0 57--63, 1996.

\bibitem[Vovk(1990)]{Vov90}
Vovk, V.~G.
\newblock Aggregating strategies.
\newblock In \emph{COLT '90: Proceedings of the 3rd Workshop on Computational
  Learning Theory}, pp.\  371--383, 1990.

\bibitem[Weibull(1995)]{Wei95}
Weibull, J.~W.
\newblock \emph{Evolutionary Game Theory}.
\newblock MIT Press, Cambridge, MA, 1995.

\end{thebibliography}

\end{document}